\documentclass[11pt,a4paper]{article}
\usepackage[margin=25mm]{geometry}
\usepackage[T1]{fontenc}
\usepackage{lmodern}
\usepackage{microtype}
\usepackage{calc}

\usepackage{graphicx}
\usepackage{amsmath,amssymb}
\usepackage{amsthm}

\usepackage{natbib}
\usepackage{bm}
\usepackage{subcaption}
\usepackage{tabularx}
\usepackage{adjustbox}
\usepackage{booktabs}
\usepackage{array}
\usepackage{multirow}

\usepackage{enumitem}
\usepackage{xcolor,framed}
\definecolor{shadecolor}{rgb}{1,1,0.75}

\usepackage{placeins}
\usepackage{float}
\usepackage{flafter}
\usepackage{tikz}
\usetikzlibrary{arrows.meta,fit}
\usepackage[colorlinks=true,linkcolor=blue!45!black,citecolor=blue!45!black,urlcolor=blue!45!black]{hyperref}
\hypersetup{pdftitle={Physics-informed neural networks (PINNs) for kinetic shock waves: simultaneous accuracy of lower- and higher-order moments, with and without moment supervision},pdfauthor={Ehsan Roohi}}
\usepackage[capitalise]{cleveref}

\providecommand{\href}[2]{#2}
\crefname{appsec}{Appendix}{Appendices}
\Crefname{appsec}{Appendix}{Appendices}

\graphicspath{{figures/}{./}}
\newcommand{\Kn}{\mathrm{Kn}}
\newcommand{\dd}{\mathrm{d}}

\newcommand{\bv}{\bm{v}}
\newcommand{\bc}{\bm{c}}

\newcommand{\M}{\mathcal{M}}
\newcommand{\R}{\mathbb{R}}
\newcommand{\norm}[1]{\left\lVert #1\right\rVert}
\newtheorem{lemma}{Lemma}

\title{Physics-informed neural networks (PINNs) for kinetic shock waves: simultaneous accuracy of lower- and higher-order moments, with and without moment supervision}

\author{Ehsan Roohi\thanks{Correspondence: \href{mailto:roohie@umass.edu}{roohie@umass.edu}.}\\[0.5em]
\small Department of Mechanical and Industrial Engineering\\
\small University of Massachusetts Amherst, Amherst, MA, USA}
\date{}

\newif\ifappendix
\let\originalappendix\appendix
\renewcommand{\appendix}{\originalappendix\appendixtrue\numberwithin{equation}{section}}

\begin{document}
\maketitle

\begin{abstract}
We investigate physics-informed neural networks (PINNs) for one-dimensional kinetic shock waves governed by the Bhatnagar--Gross--Krook (BGK) and Shakhov collision models, with and without moment supervision. The trace-free third-order and the fourth-order moments enter the fluxes of stress and heat flux through a shock; the aim, to our knowledge not previously examined for kinetic PINNs, is simultaneous accuracy of these and the lower-order moments of one predicted distribution. Data-assisted calculations, including anchor and closure studies, combine kinetic residuals with moment observations. For the BGK shock tube, a positive eight-moment representation whose spatial envelope follows the far-field decay of the initial condition gives heat-flux, stress and higher-moment discrepancies of $0.33\%$, $0.67\%$ and $1.8$--$2.0\%$ against the 6400-cell discrete-velocity method (DVM) reference supplying its observations; a five-field representation leaves the latter uncontrolled. Equation-trained characteristic-integral PINNs use analytic initial and incoming-boundary distributions and their own predicted moments. Their BGK and Shakhov shock-tube heat-flux and stress discrepancies are $0.41\%$ and $0.36\%$ against 3200-cell references; against the 6400-cell reference the BGK values are $0.21\%$ and $0.18\%$, and $0.4$--$0.9\%$ for the higher moments. Stationary neural continuation covers Shakhov Mach 2--5 and BGK Mach 6--7. Characteristic correction reduces the Mach-7 heat-flux and stress discrepancies from $3.37\%$ and $1.34\%$ to $0.42\%$ and $0.43\%$; at Mach~8 they are $0.62\%$ and $0.79\%$, and at Mach~12, across references, $0.40$--$1.03\%$ and $0.41$--$2.65\%$. Three independent Mach-2 initialisations assess repeatability. It is established that lower moments do not determine higher ones; the contribution is to quantify which moments each information regime recovers.
\end{abstract}

\section{Introduction}\label{sec:introduction}
The structure of a rarefied shock is determined by a molecular velocity distribution. Density, velocity and temperature describe its bulk behaviour, but stress, heat flux and higher moments resolve different parts of its nonequilibrium structure. These quantities matter in rarefied and high-speed gas dynamics, where local equilibrium closures can become inaccurate \citep{Cercignani1988,Sone2007,Bird1994,KarniadakisBeskok2001}. A numerical method must therefore be assessed through the moments needed by the application, rather than through agreement in a few bulk fields alone.

The Bhatnagar--Gross--Krook (BGK) model \citep{Bhatnagar1954} replaces the Boltzmann collision integral with relaxation towards a local Maxwellian. It preserves the nonlinear coupling between a distribution and its moments, but gives unit Prandtl number. The Shakhov model \citep{Shakhov1968} introduces a heat-flux correction to recover the monatomic value $2/3$. Differences between collision models are visible in shock structure and higher moments \citep{FeiLiuLiuZhang2020ShockBenchmark,ShojaSaniRoohiStefanov2021SBTGBT}. The conservative discrete-velocity method (DVM) and related kinetic schemes provide numerical references \citep{Mieussens2000,FilbetJin2010,DimarcoPareschi2014,Jin1999,XuHuang2010,Tcheremissine2006}.

Physics-informed neural networks (PINNs) approximate solutions by minimising equations and boundary residuals \citep{Raissi2019}. The PINN--BGK formulation has been applied to forward and inverse rarefied flows \citep{LouMengKarniadakis2021PINNBGK}. Separable neural representations address the cost of phase-space approximation \citep{Oh2025SPINNBGK}, and velocity-weighted losses control the errors of the conserved moments \citep{Ko2026WeightedBGK}. Asymptotic-preserving, micro--macro and structure-preserving neural methods incorporate further kinetic structure \citep{JinMaWu2021APNNTransport,JinMaWu2023APNNKinetic,LiuZhuZhu2026BiAPNN,ChenDimarcoPareschi2025SAPNN}. Related work concerns learned collision models and neural descriptions of rarefied flows \citep{BallMacArtSirignano2026MLDMS,RoohiShojaSaniStefanov2026PoF,RoohiShojaSaniAzghadi2026PoF}. These approaches motivate two related questions: which moments become identifiable when selected observations are supplied, and what transport and higher-moment accuracy can be obtained from the kinetic equations alone?

The higher moments of the distribution, the trace-free third-order moment and the fourth-order moments, have a direct role in gas dynamics. In the balance laws for stress and heat flux obtained from the kinetic equation, they form the parts of the fluxes of these quantities that the lower moments do not determine. They vanish for the distributions underlying the Grad and, for the relaxation models used here, the Navier--Stokes--Fourier descriptions, so that their values inside a shock measure transport that those descriptions omit. In Grad's and the regularised 13-moment equations they are the quantities whose closure determines the predicted stress and heat flux, and in the 26-moment equations they are variables \citep{Grad1949,StruchtrupTorrilhon2003,TorrilhonStruchtrup2004,Struchtrup2005,Torrilhon2016,GuEmerson2009R26,Levermore1996}. Their velocity kernels weight the tails of the distribution, which inside a shock are shaped by the overlap of the upstream and downstream populations \citep{MottSmith1951}, so that they amplify errors there. Agreement in lower moments does not by itself bound those projections. A finite-dimensional projection argument and velocity-space examples explain this distinction.

This work compares data-assisted and equation-trained PINNs for BGK and Shakhov shock tubes and stationary shocks. The neural calculations extend through Mach~7, and PINN-initialised characteristic BGK calculations cover Mach~7, 8 and~12. The comparison focuses on heat flux, stress and higher moments, together with the bulk fields and kinetic balances.

The data-assisted calculations determine what selected observations add to a finite kinetic representation. Moment anchors constrain transport quantities in the positive macro--micro shock-tube model, while a stationary closure head tests selective fourth-order recovery. For the base shock tube, the representation and the observations are extended to third- and fourth-order moments; the intermediate results show that these moments are limited first by the velocity kernels and then by the spatial envelope of the non-equilibrium amplitudes. These experiments are useful even when an equation-trained solution is available: they isolate which moment information improves a prediction and which quantities remain poorly determined by lower-order agreement. A constrained least-squares continuation couples the supplied moment information to separate field and physical-residual controls.

The equation-trained PINNs solve the forward problem using kinetic residuals, analytic initial and boundary states, and moments of their own predicted distributions. We consider direct kinetic-residual networks and characteristic-integral source networks. The stationary neural solutions use Shakhov collisions at Mach~2--5 and BGK collisions at Mach~6--7. For the stronger BGK shocks, characteristic correction updates macroscopic source coefficients with the neural weights fixed. This stage is a kinetic solve initialised by a PINN prediction. The calculations therefore examine both neural approximation and the improvement available from a subsequent transport solve.

The numerical inputs and reference-derived quantities are specified for each calculation. Comparisons use signed physical fields and collision-matched DVM references, with reference-grid sensitivity reported alongside neural discrepancies. All kinetic experiments are planar and one-dimensional.

The aim of this work is simultaneous accuracy of the lower- and higher-order moments\label{sec:introduction:aim}: density, velocity, temperature, heat flux and stress, and the higher moments defined in \S\ref{subsec:models}, integrated from the same predicted distribution and compared with the same reference. Because the higher moments enter the fluxes of stress and heat flux (\S\ref{subsec:models}), an approximate distribution whose lower moments are accurate but whose higher moments are not transports stress and heat flux incorrectly, so that its stress and heat flux are inconsistent with their own balance laws; for this reason the PINNs of this work are assessed on both, and the data-assisted base case is also trained on both. Neural kinetic solvers have been assessed mainly through their conserved fields or flow rates \citep{LouMengKarniadakis2021PINNBGK,DeFlorio2021PINNThermalCreep,JinMaWu2023APNNKinetic,Oh2025SPINNBGK,Ko2026WeightedBGK}. Stress and heat flux have been reported in a few cases: integrated from a distribution predicted by the network \citep{ZhangMaLouZhang2023PINNDVM,ZhangWang2024NNMicroflowBGK}, computed with a neural collision operator \citep{XiaoFrank2021NNBoltzmann,XiaoFrank2023RelaxNet,RoohiShojaSaniStefanov2026PoF}, regressed from reference data \citep{RoohiShojaSaniAzghadi2026PoF}, or obtained as the outputs of a PINN that solves a closed moment system for a Mach-10 normal shock \citep{KimJun2026FPPINN}, discussed in \S\ref{sec:discussion}. Machine-learned moment closures evolve learned or higher moments as variables of a reduced moment system rather than integrating them from one predicted distribution \citep{HanMaMaE2019PNAS,Christlieb2025HyperbolicMLClosure}. To our knowledge, whether the higher moments beyond stress and heat flux, integrated from the same predicted distribution, remain accurate alongside the lower-order fields has not been examined for neural kinetic solvers. Simultaneous accuracy is demanding for four reasons, examined in \S\ref{sec:discussion}. The higher-order kernels weight the velocity tails, where small distribution errors are amplified (\S\ref{subsec:observability}; see also \citealt{XiaoFrank2023RelaxNet,Oh2025SPINNBGK}). As is well established, a loss that controls the lower moments leaves the higher projections unconstrained (\cref{lem:observability}). Once the higher moments are targeted, fitting them alone can degrade the lower ones unless both are controlled (\S\S\ref{sup:sec:eightmoment} and~\ref{sup:sec:rphys}; Appendix~\ref{ext:app:supervised}). Finally, the higher moments of the reference carry larger discretisation errors, so that their discrepancies must be judged against reference refinement.

Section~\ref{sec:bgk} defines the kinetic models and problems; section~\ref{sec:dvm} describes the evaluation protocol. Section~\ref{subsec:observability} explains moment sensitivity. Section~\ref{sup:sec:cases} presents the data-assisted algorithms and experiments; section~\ref{sec:transient_characteristic_pinn} develops the equation-trained shock-tube algorithm. Sections~\ref{ext:sec:methodA} and~\ref{ext:sec:methodB} describe the stationary algorithms, followed by their physical comparisons in section~\ref{ext:sec:extresults}. Appendix~\ref{ext:app:reproducibility} provides the numerical
configurations, self-scaling coefficients and independent-initialisation
tests. Appendix~\ref{ext:app:reference} summarises the reference
comparisons and available refinement studies.
Appendix~\ref{ext:app:optimisation} presents the stage-resolved
optimisation histories for the data-assisted and equation-trained calculations. The data-assisted schedules are detailed in Appendix~\ref{ext:app:supervised} and tables~\ref{ext:tab:supervised_schedule} and~\ref{ext:tab:eightmoment_schedule}.\label{sec:introduction:end}

\section{Kinetic models and shock problems}
\label{sec:bgk}

\subsection{Relaxation models}
\label{subsec:models}

We consider a monatomic gas with one spatial coordinate and three molecular velocity components. All quantities are nondimensional. The gas constant, molecular mass, and reference temperature scale are absorbed into the nondimensional variables, and the translational degrees of freedom are three.

Both kinetic models used in this paper have the relaxation form
\begin{equation}
    \partial_t f + v_x\partial_x f = \frac{1}{\tau}\left(\mathcal{T}[f]-f\right),
    \qquad \bv=(v_x,v_y,v_z)\in\R^3,
    \label{eq:relaxation}
\end{equation}
where $f(x,t,\bv)$ is the molecular velocity distribution, $x$ is position and $t$ time; the relaxation target $\mathcal{T}[f]$ is constructed from moments of $f$ and $\tau$ is the relaxation time.

The BGK model takes $\mathcal{T}[f]=\M[f]$, where the local Maxwellian is defined by~\eqref{eq:maxwellian}
\begin{equation}
    \M_{\rho,u,T}(\bv)
    =\frac{\rho}{(2\pi T)^{3/2}}
      \exp\left[-\frac{|\bv-u\bm{e}_x|^2}{2T}\right],
    \label{eq:maxwellian}
\end{equation}
Here $\bm e_x$ is the unit vector along the spatial coordinate. The density, streamwise velocity and temperature in~\cref{eq:rho,eq:u,eq:T} are moments of $f$:
\begin{align}
    \rho &= \int_{\R^3} f\,\dd\bv,\label{eq:rho}\\
    u &= \frac{1}{\rho}\int_{\R^3} v_x f\,\dd\bv,\label{eq:u}\\
    T &= \frac{1}{3\rho}\int_{\R^3}|\bv-u\bm{e}_x|^2 f\,\dd\bv.\label{eq:T}
\end{align}
For a monatomic gas, the corresponding equilibrium pressure is $p=\rho T$ and the ratio of specific heats is $\gamma=5/3$. The peculiar velocity is defined by~\eqref{eq:peculiar_velocity}
\begin{equation}
    \bc = \bv-u\bm{e}_x.
\label{eq:peculiar_velocity}
\end{equation}
The nonequilibrium heat flux and normal deviatoric stress used for validation are the central moments in~\cref{eq:qx,eq:sigmaxx}
\begin{align}
    q_x &= \frac{1}{2}\int_{\R^3}|\bc|^2 c_x f\,\dd\bv,\label{eq:qx}\\
    \sigma_{xx} &= \int_{\R^3}\left(c_x^2-\frac{|\bc|^2}{3}\right)f\,\dd\bv.\label{eq:sigmaxx}
\end{align}
The sign convention in \eqref{eq:sigmaxx} follows the deviatoric second central moment used in the computations. These two quantities vanish for a local Maxwellian and therefore directly measure the nonequilibrium component of the shock layer.
The third- and fourth-order nonequilibrium moments assessed in this work are the trace-free third moment $m_{xxx}$, the trace-free fourth-order moment $R_{xx}$ and the scalar fourth-order moment $\Delta$ \citep{Struchtrup2005},
\begin{align}
    m_{xxx} &= \int_{\R^3}\left(c_x^3-\frac{3}{5}c_x|\bc|^2\right)f\,\dd\bv,\label{eq:mxxx}\\
    R_{xx} &= \int_{\R^3}|\bc|^2\left(c_x^2-\frac{|\bc|^2}{3}\right)f\,\dd\bv-7T\sigma_{xx},\label{eq:Rxx}\\
    \Delta &= \int_{\R^3}|\bc|^4f\,\dd\bv-15\rho T^2,\qquad S_{xx}=R_{xx}+\frac{\Delta}{3}.\label{eq:Delta}
\end{align}
Like $q_x$ and $\sigma_{xx}$, these moments vanish for a local Maxwellian; $S_{xx}$ combines the two fourth-order moments. All higher-moment profiles in this paper use these full-distribution definitions, evaluated with the velocity quadrature of each calculation; the closure experiments of \cref{sup:sec:closure} use the finite-grid subtractions stated there.

These moments have a direct gas-dynamic role. Multiplying \eqref{eq:relaxation} by $c_x^2-|\bc|^2/3$ or by $|\bc|^2c_x/2$ and integrating over $\bv$ gives the balance laws for $\sigma_{xx}$ and $q_x$. Besides the time derivatives, the convective fluxes $u\sigma_{xx}$ and $uq_x$, the collision terms and production terms proportional to $\partial_x u$ and to $\partial_t u+u\partial_x u$, these balance laws contain the $x$-derivatives of the fluxes relative to the mean flow,
\begin{align}
    \int_{\R^3}c_x\left(c_x^2-\frac{|\bc|^2}{3}\right)f\,\dd\bv &= m_{xxx}+\frac{8}{15}q_x,\label{eq:stress_flux}\\
    \frac12\int_{\R^3}c_x^2|\bc|^2f\,\dd\bv &= \frac12S_{xx}+\frac72T\sigma_{xx}+\frac52\rho T^2.\label{eq:heat_flux_flux}
\end{align}
In \eqref{eq:stress_flux} the term $8q_x/15$ is fixed by the heat flux, and in \eqref{eq:heat_flux_flux} the terms $\frac72T\sigma_{xx}+\frac52\rho T^2$ are fixed by the lower moments; the trace-free third moment $m_{xxx}$ and the fourth-order combination $S_{xx}$ carry the remainder, that is, the part of the transport of stress and heat flux that the lower moments do not determine. The moment $m_{xxx}$ also enters the heat-flux balance through the production term $m_{xxx}\,\partial_x u$. All three higher moments vanish for the Grad 13-moment distribution and, for the relaxation models considered here, for the first-order Chapman--Enskog distribution \citep{Struchtrup2005}, on which the Grad and Navier--Stokes--Fourier descriptions respectively rest, so that their values inside a shock measure transport of stress and heat flux beyond those descriptions. In regularised 13-moment theory $m_{xxx}$, $R_{xx}$ and $\Delta$ are the quantities whose constitutive relations close the equations for stress and heat flux \citep{StruchtrupTorrilhon2003,Struchtrup2005}, and in the 26-moment equations they are themselves variables \citep{GuEmerson2009R26}; kinetic solutions that give them accurately can therefore test such closures directly, rather than only through the stress and heat flux they predict. Physically, $m_{xxx}$ is the trace-free part of the third central moment of $c_x$, $\int_{\R^3}c_x^3f\,\dd\bv=m_{xxx}+6q_x/5$, and $\Delta$ is the departure of the fourth moment of the peculiar speed $|\bc|$ from its value for the Maxwellian with the same density and temperature, a measure of the kurtosis of the distribution. Inside a shock these moments arise from the overlap of the upstream and downstream populations \citep{MottSmith1951}; for two overlapping Maxwellians, the separation of their mean velocities lowers $\Delta$ and the difference of their temperatures raises it, so that $\Delta$ can take either sign within the layer, as it does in the base shock tube (figure~\ref{sup:fig:shocktube}). For these reasons the PINNs of this paper are assessed on the higher moments as well as on the lower ones.

Because the BGK target relaxes every nonequilibrium moment at the same rate $1/\tau$, the model has unit Prandtl number. The Shakhov model \citep{Shakhov1968} replaces the target by
\begin{equation}
    \mathcal{T}_S[f]=\M[f]\left[1+(1-\Pr)\frac{q_xc_x}{5\rho T^2}\left(\frac{|\bc|^2}{T}-5\right)\right],
    \label{eq:shakhov}
\end{equation}
whose stress vanishes and whose heat-flux moment is $(1-\Pr)q_x$. Stress therefore relaxes at the rate $1/\tau$ and heat flux at the rate $\Pr/\tau$, so that $\Pr$ is the Prandtl number of the model; the monatomic value is $\Pr=2/3$, and $\Pr=1$ recovers the BGK target. Unlike the Maxwellian, $\mathcal{T}_S$ is not positive at every velocity. The relaxation time follows from the viscosity, $\tau=\mu/p$ with $p=\rho T$. The stationary Shakhov calculations use $\mu=\mu_1T^{\omega}$ with $\omega=1$ and $\mu_1=1$, so that $\tau=1/\rho$; the BGK calculations and the Shakhov shock-tube calculation use a constant relaxation time, specified for each case below.

\subsection{Shock problems}
\label{subsec:shockproblems}

The shock-tube calculations use a constant relaxation time for both BGK and Shakhov collisions. The BGK shock tube obeys~\eqref{eq:bgk}:
\begin{equation}
    \partial_t f + v_x\partial_x f = \frac{1}{\Kn}\left(\M[f]-f\right),
    \qquad x\in[-0.5,0.5],
    \label{eq:bgk}
\end{equation}
where $\Kn$ is the nondimensional relaxation time.

The smoothed Riemann initial condition is specified in~\eqref{eq:initial}:
\begin{equation}
\begin{aligned}
    (\rho,u,T)(x,0)&=(1-H_k(x))(\rho_L,u_L,T_L)+H_k(x)(\rho_R,u_R,T_R),\\
    H_k(x)&=\frac{1}{2}\left(1+\tanh(kx)\right).
\end{aligned}
    \label{eq:initial}
\end{equation}
with $k=20$ unless otherwise stated. The distribution is initialized as a Maxwellian built from these macroscopic states. The Sod-type states in~\eqref{eq:sod_state} are
\begin{equation}
    (\rho_L,u_L,T_L)=(1,0,1),\qquad
    (\rho_R,u_R,T_R)=(0.125,0,0.8).
    \label{eq:sod_state}
\end{equation}
Inflow boundary distributions are Maxwellians corresponding to the left state for $v_x>0$ at $x=-0.5$ and to the right state for $v_x<0$ at $x=0.5$.

The stationary normal shocks are computed in a shock-fixed frame on a finite interval. The upstream state $(\rho_1,u_1,T_1)=(1,M\sqrt{5/3},1)$ is connected to the Rankine--Hugoniot (RH) state
\begin{equation}
    \rho_2=\frac{4M^2}{M^2+3},\qquad u_2=\frac{u_1}{\rho_2},\qquad T_2=\frac{5M^2-1}{4\rho_2},
    \label{eq:rh}
\end{equation}
Here $M$ is the upstream Mach number; the incoming distributions at the two ends are the corresponding Maxwellians. The supervised Mach-2 coordinate and relaxation convention is specified in \cref{sup:sec:closure}; the equation-trained stationary settings are given in \cref{ext:sec:newextension,ext:sec:methodA,ext:sec:methodB}.

\subsection{Calculations and information protocol}\label{subsec:calculations}
We use \emph{data-assisted} for calculations whose training or calibration includes reference-derived quantities, such as moment targets, output amplitudes, residual scales or reference profiles used to accept or reject updates. \emph{Equation-trained} denotes neural calculations based on governing equations, analytic initial or boundary distributions and moments of their own predicted distributions, without reference-field targets or reference-derived initialisation. The stationary \emph{equation-based characteristic correction} optimises macroscopic coefficients with fixed neural weights. Table~\ref{tab:calculations} identifies the information regime of each calculation; the detailed algorithms state the actual supplied quantities.
\begin{table}
\centering
\caption{Calculation families, collision models and information supplied
to the solver. For equation-trained PINNs and characteristic corrections,
discrete velocity method (DVM) fields are used only for assessment.}
\label{tab:calculations}
\begingroup
\small
\renewcommand{\arraystretch}{1.18}
\setlength{\tabcolsep}{5pt}
\begin{tabularx}{\textwidth}{@{}
  >{\raggedright\arraybackslash}p{0.28\textwidth}
  >{\raggedright\arraybackslash}p{0.22\textwidth}
  >{\raggedright\arraybackslash}X@{}}
\toprule
Calculation & Model and Mach number & Information supplied \\
\midrule
Data-assisted shock-tube PINNs
& BGK; prescribed initial states
& Analytic initial/boundary conditions and selected reference
  macroscopic and moment targets; in the eight-moment base case, reference profiles also guard the acceptance of updates. \\
\addlinespace[4pt]

Data-assisted closure studies
& BGK; Mach~2
& Reference macrofields, selected moments and specified
  reference-derived scales. \\
\addlinespace[4pt]

Equation-trained shock-tube PINNs
& BGK and Shakhov; prescribed initial states
& Analytic initial/incoming distributions; kinetic transport
  and moments computed by the PINN. \\
\addlinespace[4pt]

Equation-trained stationary PINNs
& Shakhov: Mach~2--5;\newline BGK: Mach~6--7
& Analytic Rankine--Hugoniot states and kinetic residuals. \\
\addlinespace[4pt]

PINN-initialised characteristic corrections
& BGK; Mach~7, 8 and~12
& PINN initial fields and analytic inflow; macroscopic coefficients
  are updated with neural weights fixed. \\
\addlinespace[4pt]

Independent-initialisation PINN study
& Shakhov; Mach~2
& Three random initialisations with the same equation-based
  training conditions and optimisation schedule. \\
\bottomrule
\end{tabularx}
\endgroup
\end{table}

\section{Discrete-velocity method (DVM) references and evaluation}\label{sec:dvm}
For the equation-trained characteristic shock tubes, the reference uses a rectangular velocity grid and first-order upwind transport. For cell index $i$, velocity index $j$ and time level $n$, the transport update is
\begin{equation}
 f_{ij}^{*}=f_{ij}^{n}-\frac{\Delta t}{\Delta x}\left[v_{x,j}^{+}(f_{ij}^{n}-f_{i-1,j}^{n})+v_{x,j}^{-}(f_{i+1,j}^{n}-f_{ij}^{n})\right],\label{eq:dvm_transport}
\end{equation}
where $v^+=\max(v,0)$, $v^-=\min(v,0)$, $\Delta x$ is the cell width and $\Delta t$ the time step. Moments of $f^*$ define the collision target $\mathcal T^*$. The constant-relaxation collision substep is
\begin{equation}
 f_{ij}^{n+1}=e^{-\Delta t/\Kn}f_{ij}^{*}+(1-e^{-\Delta t/\Kn})\mathcal T_{ij}^{*}.\label{eq:dvm_relax}
\end{equation}
Equations~\eqref{eq:dvm_transport} and~\eqref{eq:dvm_relax} apply with the BGK or Shakhov target of the corresponding calculation. The time step satisfies the Courant--Friedrichs--Lewy (CFL) condition $\Delta t\max_j|v_{x,j}|/\Delta x\leq0.4$. The two spatial grids have 1600 and 3200 cells and the same $48\times20\times20$ velocity nodes on $[-7,7]^3$. Their difference measures coupled space--time refinement, since the time step is reduced with the cell width. Shakhov targets and signed tails are not clipped. A 6400-cell BGK solution with the same scheme and velocity grid supplies the observations and the comparison of the data-assisted base case in \cref{sup:sec:eightmoment} and is compared with the equation-trained BGK prediction in Appendix~\ref{ext:app:reference}; the other data-assisted shock-tube comparisons use the coarser solutions stated in \cref{sup:sec:shocktube}.

Stationary comparisons use separately computed references matched in Mach number, collision model and relaxation convention. The model and reference are evaluated on common positions within the stated comparison window. For a field $a$, the reported relative discrepancy is
\begin{equation}
 \varepsilon_a=\left[\frac{\int_I(a_\theta(x)-a_{\rm DVM}(x))^2\,\dd x}{\int_I a_{\rm DVM}(x)^2\,\dd x}\right]^{1/2},\label{eq:field_error}
\end{equation}
where the subscript $\theta$ denotes a quantity computed from the trained model, $I$ is the declared spatial interval and the integrals are evaluated by the stated quadrature or trapezoidal rule; on uniform grids they reduce to sums over the common nodes. Unless stated otherwise, the relative $L^2$ discrepancies reported in this paper are of this form, quoted as fractions or in per cent. Equation~\eqref{eq:field_error} is a post-solve diagnostic. No model shift or amplitude fit is used to reduce it. Where a stationary reference is centred, its own prescribed density midpoint defines that frame independently of the model. A finite reference discrepancy includes both neural and reference discretisation errors.

\section{Velocity-tail sensitivity and moment observability}
\label{subsec:observability}

An exact solution determines every moment for which the defining integral exists. A finite approximation need not resolve all such moments equally well. The following projection argument explains why agreement in lower moments is an incomplete accuracy test.

Let $e(\bv)=f_\theta(\bv)-f(\bv)$ denote the local velocity-space error of a model distribution $f_\theta$ with respect to the exact distribution $f$ at a fixed physical point. Low-order moment functionals act through the finite collection of velocity test functions in~\eqref{eq:observed_kernels}:
\begin{equation}
    \psi_i(\bv)\in\{1,v_x,|\bv|^2,c_x|\bc|^2,c_x^2-|\bc|^2/3,\ldots\},
\label{eq:observed_kernels}
\end{equation}
where the first group corresponds to mass, momentum, and energy and the latter group corresponds to heat flux and stress. The higher-order kernels in~\cref{eq:phi_mxxx,eq:phi_Rxx} probe different velocity-space components:
\begin{align}
    \phi_m(\bv) &= c_x^3-\frac{3}{5}c_x|\bc|^2,\label{eq:phi_mxxx}\\
    \phi_R(\bv) &= |\bc|^2\left(c_x^2-\frac{|\bc|^2}{3}\right),\label{eq:phi_Rxx}
\end{align}
The moment $R_{xx}$ in~\eqref{eq:Rxx} also contains the subtraction $7T\sigma_{xx}$.

\begin{lemma}[Finite moment constraints do not control unobserved closure moments]
\label{lem:observability}
Let $V=L^2(w\,\mathrm d\bv)$ be a real Hilbert space, with $w>0$ almost everywhere, and assume that $\phi,\psi_1,\ldots,\psi_N\in V$. For physical unweighted moments on a bounded velocity domain, take $w=1$. On an unbounded domain, a weighted representation of those same functionals instead requires their Riesz kernels $\phi/w$ and $\psi_i/w$ to belong to $V$. Equip it with the inner product $\langle a,b\rangle_w=\int a(\bv)b(\bv)w(\bv)\,\dd\bv$, and let $\Psi=\mathrm{span}\{\psi_1,\ldots,\psi_N\}$ be the finite span of all velocity kernels directly observed by a loss. If a higher-order moment kernel $\phi\notin\Psi$, then there exists a nonzero error $e\in V$ satisfying the observed-moment constraints in~\eqref{eq:observed_orthogonality},
\begin{equation}
    \langle e,\psi_i\rangle_w=0,\qquad i=1,\ldots,N,
\label{eq:observed_orthogonality}
\end{equation}
but with the nonzero unobserved moment in~\eqref{eq:unobserved_moment}:
\begin{equation}
    \langle e,\phi\rangle_w\neq 0.
\label{eq:unobserved_moment}
\end{equation}
Consequently, exact agreement in all observed low-order moment functionals does not imply exact agreement of the unobserved higher-order moment associated with $\phi$.
\end{lemma}

\begin{proof}
Let $P_\Psi\phi$ be the orthogonal projection of $\phi$ onto $\Psi$ and set $e=\phi-P_\Psi\phi$. Because $\phi\notin\Psi$, $e\neq 0$. By construction, $e$ is orthogonal to every $\psi_i$, so all observed moment errors vanish. The projection identity in~\eqref{eq:projection_identity} nevertheless gives a nonzero unobserved moment:
\begin{equation}
    \langle e,\phi\rangle_w
    =\langle \phi-P_\Psi\phi,\phi\rangle_w
    =\langle e,e+P_\Psi\phi\rangle_w
    =\norm{e}_w^2\neq 0,
\label{eq:projection_identity}
\end{equation}
where the last equality uses $e\perp\Psi$ and $P_\Psi\phi\in\Psi$. This proves the claim. The result identifies the moment projections left unresolved by a finite set of observations; the numerical comparisons assess those projections directly.
\end{proof}

\subsection{Hermite example}
The abstract projection argument in \cref{lem:observability} can be made explicit with normalised one-dimensional peculiar velocity $\xi$ and the unit Maxwellian $M(\xi)=(2\pi)^{-1/2}\exp(-\xi^2/2)$. The probabilists' Hermite polynomials used here are listed in~\eqref{eq:hermite_polynomials}:
\begin{equation}
\begin{aligned}
&H_0=1,\qquad H_1=\xi,\qquad H_2=\xi^2-1,\\
&H_3=\xi^3-3\xi,
\qquad H_4=\xi^4-6\xi^2+3 .
\end{aligned}
\label{eq:hermite_polynomials}
\end{equation}
For $0<\varepsilon<1/6$, consider the fourth-order perturbation in~\eqref{eq:h4_perturbation}:
\begin{equation}
    f_\varepsilon(\xi)=M(\xi)\left[1+\varepsilon H_4(\xi)\right].
    \label{eq:h4_perturbation}
\end{equation}
Hermite orthogonality gives the vanishing lower-order projections in~\eqref{eq:h4_low_order_zero}:
\begin{equation}
    \int M H_4\,\dd\xi=0,
    \qquad
    \int \xi M H_4\,\dd\xi=0,
    \qquad
    \int \xi^2 M H_4\,\dd\xi=0.
    \label{eq:h4_low_order_zero}
\end{equation}
Thus the perturbation leaves the density, mean velocity, and temperature-level second moment unchanged. In contrast, the fourth-order projection in~\eqref{eq:h4_fourth_nonzero} does not vanish:
\begin{equation}
\begin{aligned}
    \int \xi^4 f_\varepsilon\,\dd\xi
    &=\int \xi^4 M\,\dd\xi
    +\varepsilon\int \xi^4 M H_4\,\dd\xi,\\
    \int \xi^4 M H_4\,\dd\xi
    &=\int H_4^2M\,\dd\xi=4!\neq0.
\end{aligned}
    \label{eq:h4_fourth_nonzero}
\end{equation}
Therefore two positive distributions can share the same low-order macroscopic moments while possessing different fourth-order information. At the level of signed perturbations, $M H_3$ changes third-order information while remaining orthogonal to the density, momentum and energy kernels. The linear form $M(1+\varepsilon H_3)$ is not globally positive for nonzero $\varepsilon$, so it is not used as a second positive-distribution example. Figure~\ref{fig:hermite_observability} shows the fourth-order perturbation and its tail-sensitive contribution.

\begin{figure}
\centering
\includegraphics[width=0.98\textwidth,keepaspectratio,height=0.72\textheight]{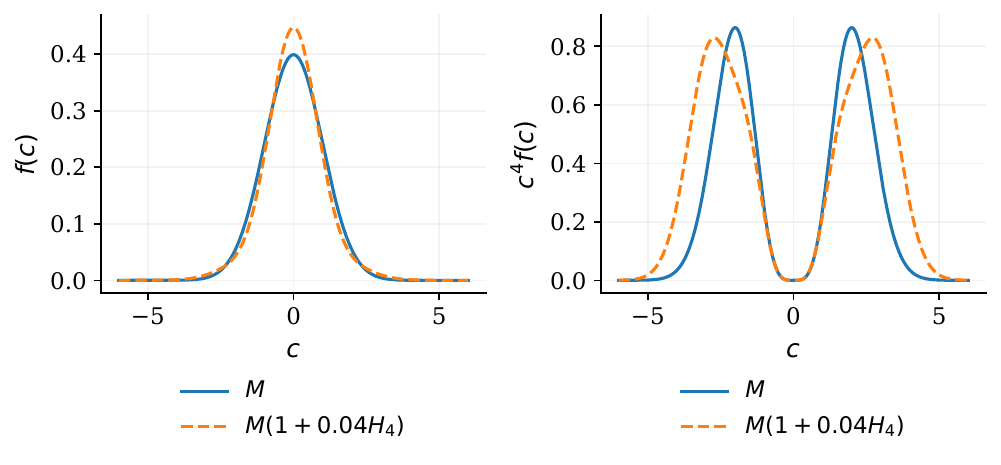}
\caption{Illustration of high-order non-observability under low-order moment constraints. A fourth-order Hermite perturbation of a Maxwellian, $M(1+\varepsilon H_4)$, leaves the density, mean velocity, and second moment nearly unchanged while altering the fourth-order moment. The right panel plots the fourth-moment integrand $c^4f(c)$ for the same two distributions. }
\label{fig:hermite_observability}
\end{figure}

\subsection{Velocity-tail sensitivity}
The same conclusion can be stated in physical scaling terms. If a distribution error $\delta f$ is concentrated near a peculiar-speed magnitude $|c|\approx C$, the induced moment errors have the scaling in~\eqref{eq:tail_scaling}:
\begin{equation}
    \delta\rho=O(\varepsilon),
    \qquad
    \delta q_x=O(C^3\varepsilon),
    \qquad
    \delta m=O(C^3\varepsilon),
    \qquad
    \delta R=O(C^4\varepsilon),
    \label{eq:tail_scaling}
\end{equation}
Here $\varepsilon$ measures the magnitude of the localised distribution perturbation, $C$ its characteristic peculiar speed, $m$ a third-order closure-like observable and $R$ a fourth-order closure-like observable; $\varepsilon$ is not the Hermite coefficient of the preceding example. Hence, a small high-peculiar-speed-tail error that is almost invisible to density can be strongly amplified in $R_{xx}$. 
\Cref{fig:velocity_tail_weights,fig:cumulative_tail_contribution} combine the polynomial moment weights and the cumulative fraction of each model moment coming from $|c|>C$. Here \(c=v-u\) denotes the peculiar velocity and \(W_\rho\), \(W_q\), \(W_m\), and \(W_R\) denote schematic velocity-space weights associated with density-like, heat-flux-like, third-order-closure-like, and fourth-order-closure-like observables, respectively.  Thus \(W_\rho=1\), while the heat-flux and third-order closure weights grow approximately as \(|c|^3\), and the fourth-order closure weight grows approximately as \(|c|^4\).  These schematic kernels illustrate the increasing velocity-tail sensitivity of higher-order observables. The tail fraction plotted in figure~\ref{fig:cumulative_tail_contribution} is $\Phi_k(C)=\int_{|c|>C}|c|^k f_{\rm toy}\,\mathrm dc/\int|c|^k f_{\rm toy}\,\mathrm dc$, with $k=0$ for $\Phi_\rho$, $k=3$ for $\Phi_q=\Phi_m$ and $k=4$ for $\Phi_R$, and $f_{\rm toy}=(2\pi)^{-1/2}e^{-c^2/2}+0.01e^{-(c-4)^2/(2\times0.35^2)}+0.01e^{-(c+4)^2/(2\times0.35^2)}$. The ordering $\Phi_R(C)>\Phi_m(C)=\Phi_q(C)>\Phi_\rho(C)$ illustrates the increasing tail sensitivity in this example.

\begin{figure}
\centering
\includegraphics[width=0.98\textwidth,
  height=0.34\textheight,
  keepaspectratio]{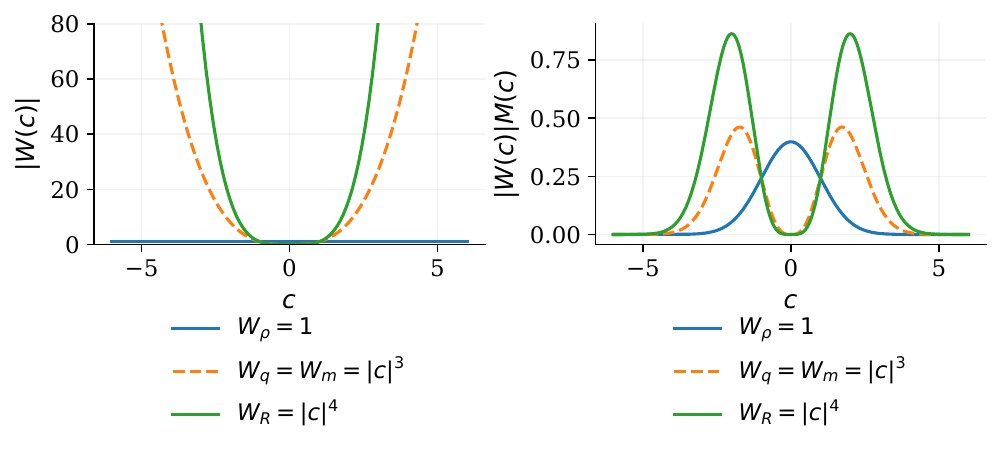}
\caption{Polynomial and Maxwellian-weighted moment kernels used to illustrate the velocity-space origin of closure sensitivity. Here \(W_\rho\), \(W_q\), \(W_m\), and \(W_R\) denote schematic velocity weights for density-like, heat-flux-like, third-order-closure-like, and fourth-order-closure-like observables, respectively. The fourth-order closure kernel grows faster in the high-peculiar-speed tail, explaining why small errors in \(f-M\) can have a weak density signature but a large contribution to \(R_{xx}\).}
\label{fig:velocity_tail_weights}
\end{figure}

\begin{figure}
\centering
\includegraphics[width=0.88\textwidth,
  height=0.36\textheight,
  keepaspectratio]{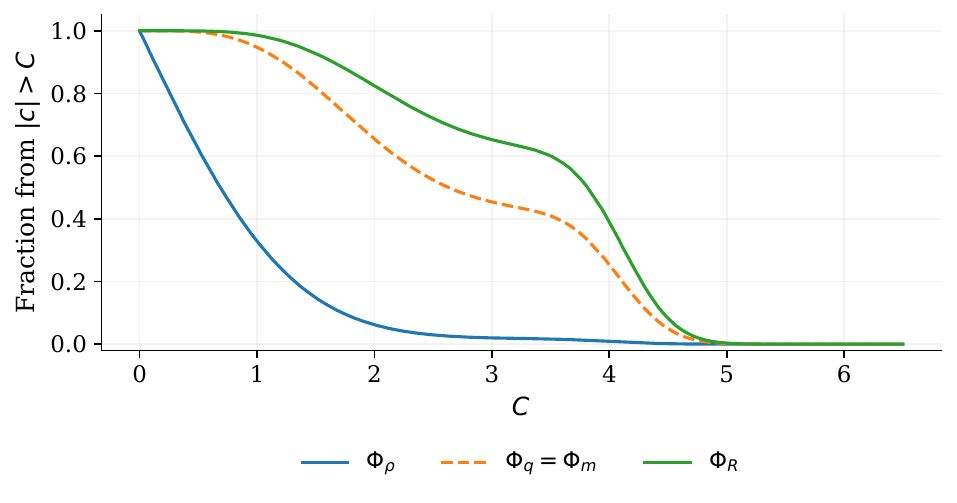}
\caption{Cumulative contribution from the high-peculiar-speed region $|c|>C$ for schematic density-like, heat-flux-like, third-order-closure-like and fourth-order-closure-like moment kernels. The fourth-order closure observable retains a larger tail contribution than density-like and lower-order moment functionals, explaining why small velocity-tail errors can be weakly visible in low-order moments but strongly amplified in $R_{xx}$.}
\label{fig:cumulative_tail_contribution}

\end{figure}

The integrals in $\Phi_k$ are evaluated on $[-8,8]$. For any moment kernel $\phi$ and error $e=f_\theta-f_{\mathrm{ref}}$ with respect to a reference distribution $f_{\rm ref}$, the corresponding weighted bound is~\eqref{eq:moment_bound}:

\begin{equation}
    \left|\int \phi(\bv)e(\bv)\,\dd\bv\right|
    \leq \norm{e}_{L^2_w}\,\norm{\phi}_{L^2_{w^{-1}}}.
    \label{eq:moment_bound}
\end{equation}
Here $w(\bv)>0$, $\|e\|_{L^2_w}^2=\int |e|^2w\,\mathrm d\bv$ and $\|\phi\|_{L^2_{w^{-1}}}^2=\int |\phi|^2/w\,\mathrm d\bv$; both integrals are assumed finite. 
Equation~\eqref{eq:moment_bound} relates moment error to distribution error and motivates the assessment of signed higher moments alongside the training objective.

\section{Data-assisted PINNs and closure experiments}\label{sup:sec:cases}
The data-assisted calculations use selected DVM macroscopic and moment observations to train a positive kinetic representation. The shock-tube cases examine transport moments and, for the base case, the third- and fourth-order moments integrated from the same distribution, while the stationary BGK cases examine fourth-order closure information. The DVM profiles supply both the training observations and the full-profile comparisons; the latter include training locations.

\subsection{Moment-supervised BGK shock tubes}\label{sup:sec:shocktube}
The collision law is BGK with constant relaxation time $\mathrm{Kn}$. Initial and incoming-boundary Maxwellians are prescribed analytically. Define $\widehat t=t/t_f$, where $t_f$ is the comparison time. The network input is
\begin{equation}
\Phi=[\widehat t,x,\sin(2\pi x),\cos(2\pi x),\sin(4\pi x),\cos(4\pi x),\sin(\pi\widehat t),\cos(\pi\widehat t)].\label{sup:eq:features}
\end{equation}
A four-hidden-layer network with 128 sigmoid linear unit (SiLU) activations per layer maps $\Phi$ to $(r_\rho,r_u,r_T,r_q,r_\sigma)$. Its latent parameters and correction amplitudes are
\begin{align}
\rho_\theta&=\rho_0\exp(a_\rho\widehat t\tanh r_\rho),\quad u_\theta=u_0+a_u\widehat t\tanh r_u,\quad T_\theta=T_0\exp(a_T\widehat t\tanh r_T),\label{sup:eq:macros}\\
A_q&=\widehat t G_w a_q\tanh r_q,\quad A_\sigma=\widehat t G_w a_\sigma\tanh r_\sigma,\quad G_w=\exp[-(x/w)^4].\label{sup:eq:amplitudes}
\end{align}
Here $(\rho_0,u_0,T_0)$ are the initial fields. The empirical output bounds are $(a_\rho,a_u,a_T)=(0.45,0.42,0.35)$, $a_q=a_\sigma=0.25$, and $w=0.13$ is the width of the spatial envelope $G_w$ of the transport amplitudes. These parameters bound the representation amplitudes and are distinct from the loss weights. With $\widehat{\boldsymbol c}=(\boldsymbol v-u_\theta\boldsymbol e_x)/\sqrt{T_\theta}$, the representation is
\begin{align}
f_\theta&=M_{\rho_\theta,u_\theta,T_\theta}\exp\psi_\theta,\quad
\psi_\theta=\psi_{\max}\tanh[(A_q\phi_q+A_\sigma\phi_\sigma)/\psi_{\max}],\label{sup:eq:f}\\
\phi_q&=(|\widehat{\boldsymbol c}|^2/2-5/2)\widehat c_x/8,\quad
\phi_\sigma=(\widehat c_x^2-|\widehat{\boldsymbol c}|^2/3)/4,\quad\psi_{\max}=0.65.\label{sup:eq:kernels}
\end{align}
The features in \eqref{sup:eq:features}, amplitudes in \eqref{sup:eq:amplitudes} and velocity kernels in \eqref{sup:eq:kernels} specify the reconstruction. The parameters in \eqref{sup:eq:macros} need not equal the moments integrated from \eqref{sup:eq:f}. The exponential preserves positivity throughout optimisation. Because its amplitudes control $\rho$, $u$, $T$, $q_x$ and $\sigma_{xx}$, this is called the five-field representation.

The objective combines kinetic and relative residuals, incoming-boundary discrepancies, a numerical guard and reference targets:
\begin{equation}
\mathcal L=\lambda_p\mathcal L_p+\lambda_r\mathcal L_r+\lambda_b\mathcal L_b+\lambda_+\mathcal L_+
+\lambda_M\mathcal L_M+\lambda_q\mathcal L_q+\lambda_\sigma\mathcal L_\sigma.\label{sup:eq:loss}
\end{equation}
The $\lambda$ values are scalar weights; $\mathcal L_p$ and $\mathcal L_r$ are the kinetic and relative residual losses, $\mathcal L_b$ the incoming-boundary loss, $\mathcal L_+$ the numerical guard, $\mathcal L_M$ the macro lock on $\rho$, $u$ and $T$, and $\mathcal L_q$, $\mathcal L_\sigma$ the heat-flux and stress targets \eqref{sup:eq:anchors}, whose observation locations are called anchors. For the implemented latent Maxwellian $M_\theta=M_{\rho_\theta,u_\theta,T_\theta}$, the residual is $\mathcal R=\partial_tf_\theta+v_x\partial_xf_\theta-(M_\theta-f_\theta)/\mathrm{Kn}$. With $B$ space--time samples, velocity weights $w_j$ and $W_j=1+0.1|\boldsymbol v_j|^4$, the kinetic loss is
\begin{equation}
\mathcal L_p=\frac{\sum_{i=1}^{B}\sum_jw_jW_j\mathcal R_{ij}^2}{B\sum_jw_jW_j}.\label{sup:eq:pde}
\end{equation}
In \eqref{sup:eq:pde}, $\mathcal R_{ij}$ is evaluated at sample $i$ and velocity node $j$. The relative term uses the same average with residual divided by $|f_\theta|+|M_\theta|+\epsilon_r$. The boundary term averages squared incoming-Maxwellian discrepancies. The guard averages $[\max(f_{\rm floor}-\exp\psi_\theta,0)]^2$. In the eight-anchor configuration, $\epsilon_r=10^{-5}$ and $f_{\rm floor}=0.02$.

At $t_f$, the macro lock compares the density, velocity and temperature obtained by integrating $f_\theta$ with $N_M$ sampled DVM values, normalised by their reference maximum amplitudes; the velocity denominator includes $10^{-4}$ and its contribution has a configuration-specific weight (12 for the eight-anchor configuration). Integrated heat-flux and stress targets have the form
\begin{equation}
\mathcal L_a=\frac1{N_A}\sum_{k=1}^{N_A}\left[\frac{a[f_\theta](t_f,x_k)-a^{\rm ref}(x_k)}{S_a+s_a}\right]^2,
\quad a\in\{q_x,\sigma_{xx}\}.\label{sup:eq:anchors}
\end{equation}
In \eqref{sup:eq:anchors}, $a[f_\theta]$ is the integrated moment, $N_A$ the location count per moment, and $S_a=\max(\max_k|a^{\rm ref}(x_k)|,10^{-3})$ unless overridden; $s_a$ is a fixed offset added to this scale, and the eight-anchor case uses $s_a=0.003$. Target amplitudes therefore also convey reference information. Eligible reference indices lie in $[-0.105,0.135]$; rounded equally spaced indices in this list are selected when subsampling is needed. The anchor locations are fixed at the comparison time.

The verified eight-anchor, 32-macro-point stage uses AdamW \citep{LoshchilovHutter2019AdamW}, the Adam optimiser \citep{KingmaBa2015Adam} with decoupled weight decay, for 9000 updates from its recorded warm start, seed 606, weight decay $10^{-7}$, gradient-norm clipping 0.35, eight kinetic-residual and four boundary samples per update. Its cosine learning rate falls from $1.2\times10^{-5}$ to $8\times10^{-7}$. In the order $(p,r,b,+,M,q,\sigma)$, the weights in \eqref{sup:eq:loss} are $(1,0.01,1,10,20,0.7,0.9)$. Case-specific observation counts and continuation budgets are given below.

All BGK shock-tube cases use this positive reconstruction and these observation losses; the base case is subsequently extended to third- and fourth-order moments (\cref{sup:sec:eightmoment}). AdamW supplies the initial optimisation stage; a constrained least-squares continuation is used for the base Sod case to refine its moments while controlling the kinetic residual and transport balances. The parameter-variation cases in table~\ref{sup:tab:shockcases} use their stated AdamW budgets. The base case uses 42 locations per transport moment and 160 macroscopic-target locations, with $s_q=s_\sigma=0.003$ in \eqref{sup:eq:anchors}. Its 5000-update AdamW stage has a learning rate decreasing from $2\times10^{-6}$ to $2\times10^{-7}$, gradient-norm clipping at 0.18 and velocity-target weight 15. In the order $(p,r,b,+,M,q,\sigma)$, the weights in \eqref{sup:eq:loss} are $(0.6,0.006,1,12,35,0.65,1.15)$.

The continuation forms the separate field losses $\ell_\rho,\ell_u,\ell_T,\ell_q,\ell_\sigma$ from the same observations. Physical controls comprise the raw and relative kinetic losses, the boundary loss, a kinetic loss using the Maxwellian of the integrated moments, and the squared transport balances
\begin{equation}
\boldsymbol b(t,x)=\sum_jw_j
\begin{bmatrix}1\\v_{x,j}\\|\boldsymbol v_j|^2/2\end{bmatrix}
\left(\partial_t f_\theta+v_{x,j}\partial_xf_\theta\right).
\label{sup:eq:balances}
\end{equation}
Here $\boldsymbol b$ contains the mass, momentum and energy transport residuals, $\boldsymbol v_j$ is velocity node $j$, and $w_j$ is its quadrature weight. The mean squares of the three entries are controlled individually. The latent-Maxwellian and integrated-moment-Maxwellian residuals are evaluated separately on the finite velocity grid. Three additional observations specify the reference maximum of $q_x$ and the reference maximum and minimum of $\sigma_{xx}$; the corresponding model extrema are extracted on the comparison grid.

Let $\boldsymbol r(\boldsymbol\theta)$ collect the pointwise field residuals, normalised by the field scales and the square root of their count and initial component loss. The joint-fit block uses these observation residuals in the Jacobian and imposes the physical loss through the direction constraints. The balance-controlled block appends $\sqrt{\ell_k/\ell_k^0}$ for the separate physical components, where $\ell_k$ is a component loss and $\ell_k^0$ its value at the start of the current continuation block, and the extremum-controlled block additionally appends these scaled terms for the three extremum losses. A damped Gauss--Newton proposal \citep{Levenberg1944,Marquardt1963} is
\begin{equation}
\boldsymbol d_0=-J^T(JJ^T+\lambda_d I)^{-1}\boldsymbol r,
\qquad J=\frac{\partial\boldsymbol r}{\partial\boldsymbol\theta}.
\label{sup:eq:gn}
\end{equation}
Here $\boldsymbol\theta$ contains the neural weights, $J$ is the residual Jacobian, $I$ is the identity matrix and $\lambda_d>0$ is a damping parameter. The solve in \eqref{sup:eq:gn} acts in residual space. To protect the individual components, its direction is corrected by
\begin{equation}
\min_{\boldsymbol d}\ \tfrac12(\boldsymbol d-\boldsymbol d_0)^T
(J^TJ+\lambda_d I)(\boldsymbol d-\boldsymbol d_0),
\qquad U\boldsymbol d\leq-\eta\boldsymbol 1.
\label{sup:eq:direction}
\end{equation}
Each row of $U$ is an active component gradient divided by its Euclidean norm; $\boldsymbol 1$ is the vector of ones and $\eta=10^{-4}\max(\|\boldsymbol d_0\|_2,10^{-8})$ sets the descent margin before step-length capping. Thus the direction in \eqref{sup:eq:direction} accounts for field sensitivity and kinetic balance together.

The damping is $\lambda_d=\mu\,\operatorname{mean}\{\operatorname{diag}(JJ^T)\}$ with $\mu$ initially $10^{-3}$. Acceptance divides $\mu$ by three down to $10^{-8}$; rejection multiplies it by ten. The direction norm is capped at 0.1, and backtracking tests $\alpha=2^{-k}$ for $k=0,\ldots,15$. The update is $\boldsymbol\theta\leftarrow\boldsymbol\theta+\alpha\boldsymbol d$. Field losses and extremum errors must remain below their block-start values. In the balance-controlled blocks, each physical component must be nonincreasing, with strict decrease required for the raw residuals and the three balance losses. A rejected trial restores the preceding weights.

The base-case schedule has 45 accepted joint-fit updates, 11 updates with separate residual and balance control, and 44 updates with the extremum controls active. These blocks use the same constrained least-squares framework with the stated active components. The balance-controlled blocks use 32 fixed space--time training points, double precision, and the original quadrature-node values. Figure~\ref{sup:fig:continuation_losses} in Appendix~\ref{ext:app:optimisation} reports the objectives separately by block. Residual and balance assessments use 32 comparison points and a fresh 64-point set, neither used for gradient updates. On these sets, the selected endpoint has raw-residual mean-square discrepancies 64\% and 60\% below the AdamW endpoint, respectively; momentum-balance root-mean-square discrepancies are 35\% and 36\% lower. These are sampled assessments rather than global convergence bounds. This five-field endpoint is the starting state of the eight-moment extension in \cref{sup:sec:eightmoment}.\label{sup:sec:shocktube:end}

Table~\ref{sup:tab:shockcases} reports time, relaxation-time and jump-strength variations for the five-field representation; the mild jump replaces the right state of \eqref{eq:sod_state} by $(\rho_R,u_R,T_R)=(0.375,0,0.8)$, with the left state, $\mathrm{Kn}=1$, $t_f=0.02$ and $k=20$ of the base case. Each row is compared with the saved DVM solution that supplied its observations. These solutions use a $32\times12\times12$ velocity grid in single precision, with 400 cells for the base case, and are coarser than the references of \S\ref{sec:dvm}. The base case is re-assessed against the refined reference in \cref{sup:sec:eightmoment}.
\begin{table}\centering\small
\caption{Moment-supervised BGK shock-tube discrepancies of the five-field representation, as relative $L^2$ fractions against the DVM solution that supplied the observations of each row (400 cells for the base case). The rows specify comparison time, relaxation scale and jump strength.}\label{sup:tab:shockcases}
\begin{adjustbox}{max width=\textwidth}
\begin{tabular}{lrrrrr}\toprule
Case & $\rho$ & $u$ & $T$ & $q_x$ & $\sigma_{xx}$\\\midrule
Sod, $Kn=1$, $t=0.020$ & $6.0610\times10^{-7}$ & $6.8606\times10^{-5}$ & $1.7007\times10^{-7}$ & 0.011811 & 0.013351\\
Sod, $Kn=1$, $t=0.015$ & 0.0011369 & 0.002832 & 0.0010848 & 0.033741 & 0.021072\\
Sod, $Kn=1$, $t=0.025$ & 0.0064093 & 0.0070481 & 0.0033904 & 0.076672 & 0.076487\\
Sod, $Kn=0.5$, $t=0.020$ & 0.0019268 & 0.0036471 & 0.0012872 & 0.032416 & 0.033021\\
Sod, $Kn=2$, $t=0.020$ & 0.0020469 & 0.0036484 & 0.0013706 & 0.034824 & 0.03343\\
Mild jump, $Kn=1$, $t=0.020$ & 0.0011411 & 0.0018488 & 0.0014552 & 0.031617 & 0.024638\\
\bottomrule\end{tabular}
\end{adjustbox}\end{table}
Figure~\ref{sup:fig:shocktube} displays the base-case signed profiles, including the third- and fourth-order moments discussed next.

\subsection{Third- and fourth-order moments of the base case}\label{sup:sec:eightmoment}
Beyond the macroscopic fields, the five-field representation supplies and controls only $q_x$ and $\sigma_{xx}$. To first order in the amplitudes, \eqref{sup:eq:f} gives $\delta q_x=\tfrac{5}{16}\rho_\theta T_\theta^{3/2}A_q$ and $\delta\sigma_{xx}=\tfrac13\rho_\theta T_\theta A_\sigma$, whereas $m_{xxx}$, $R_{xx}$ and $\Delta$, defined in \cref{eq:mxxx,eq:Rxx,eq:Delta}, respond only at second order. The third- and fourth-order moments of the five-field state are therefore uncontrolled by-products of its heat-flux and stress corrections. When this state is integrated on the $48\times20\times20$ double-precision velocity grid of \S\ref{sec:dvm} and compared with the refined reference introduced below, $m_{xxx}$, $R_{xx}$ and $\Delta$ differ from the reference by more than $100\%$ (table~\ref{sup:tab:eightmoment}), and $R_{xx}$ and $\Delta$ have the opposite sign in the core of the layer (figure~\ref{sup:fig:shocktube}). This is a direct instance of \S\ref{subsec:observability}: agreement in the supplied moments does not constrain the projections that are not supplied. The base case is therefore extended so that the third- and fourth-order moments are represented and observed. The extension is reported in three steps, because the intermediate states identify what limits these moments in a data-assisted PINN.
\begin{figure}\centering
\includegraphics[width=\textwidth,height=.72\textheight,keepaspectratio]{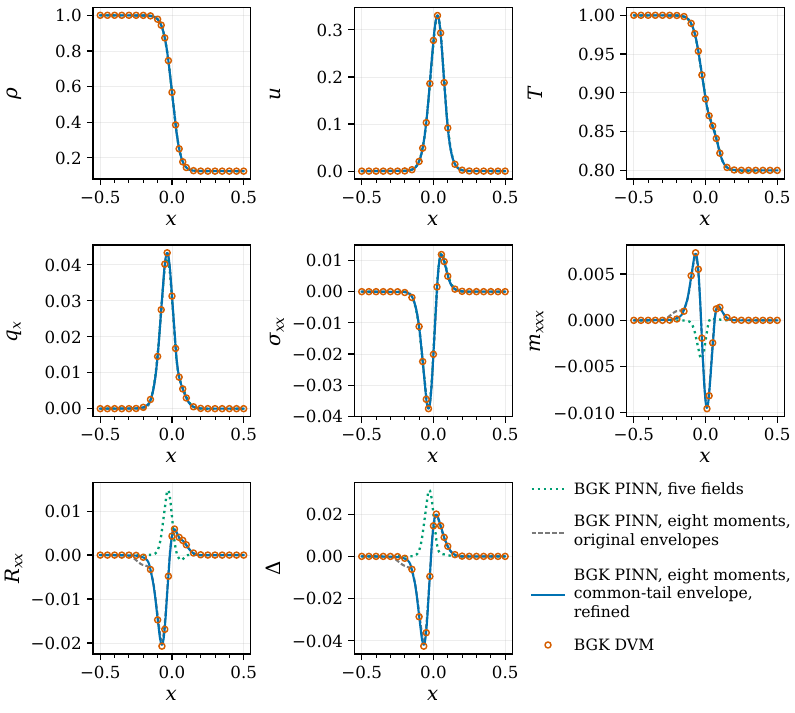}
\caption{Moment-supervised BGK shock tube at $\mathrm{Kn}=1$ and $t_f=0.02$: eight signed fields of the five-field state, of the eight-moment state with the original envelopes, and of the eight-moment state with the common-tail envelope after the output-layer refinement, compared with the 6400-cell DVM reference. All displayed fields are integrated from the corresponding positive kinetic distribution. The eight-moment states use macroscopic, transport-moment and third- and fourth-order moment observations. In $\rho$, $u$, $T$, $q_x$ and $\sigma_{xx}$ the three states nearly coincide at the scale of the plot (table~\ref{sup:tab:eightmoment} gives their discrepancies); the two eight-moment states differ most visibly in the upstream shoulder of $m_{xxx}$, $R_{xx}$ and $\Delta$, which is enlarged in figure~\ref{sup:fig:uptail}. The original envelopes are $G_w$ in \eqref{sup:eq:amplitudes} and $G_h$ in \eqref{sup:eq:amplitudes8}; the common-tail envelope is \eqref{sup:eq:commontail}, and the output-layer refinement is described in \S\ref{sup:sec:eight:refine}.}\label{sup:fig:shocktube}
\end{figure}

\subsubsection{Reference and observations.}\label{sup:sec:eight:obs}
The reference for this subsection is a 6400-cell BGK solution computed with the scheme, velocity grid and CFL number of \S\ref{sec:dvm} in double precision (2240 time steps). It uses a transverse-velocity reduction of the same discrete scheme, which reproduces the 3200-cell fields to $6\times10^{-14}$. Its difference from the 3200-cell solution is given in the last row of table~\ref{sup:tab:eightmoment}. Successive refinement differences halve (for $q_x$, $0.41\%$ between 1600 and 3200 cells and $0.20\%$ between 3200 and 6400 cells), as expected for a first-order scheme, so the remaining error of the 6400-cell reference is of the order of that last row. The 400-cell solution that supplied the five-field observations differs from the 6400-cell reference by $2.99\%$ in $q_x$ and $2.53\%$ in $\sigma_{xx}$. This exceeds the base-case discrepancies in table~\ref{sup:tab:shockcases}, which therefore measure the fit to the coarse solution; the five-field state is re-assessed against the refined reference in table~\ref{sup:tab:eightmoment}.
\begin{table}\centering\small
\caption{Base case, $\mathrm{Kn}=1$ and $t_f=0.02$: full-domain relative $L^2$ discrepancies (\%) of the four data-assisted states against the 6400-cell DVM reference, and the difference between the 3200- and 6400-cell DVM solutions, with the fine solution interpolated linearly onto the coarse nodes. All moments are integrated from the distribution on the $48\times20\times20$ velocity grid; $S_{xx}=R_{xx}+\Delta/3$. `Same weights': the network of the preceding row with only the envelopes replaced (\S\ref{sup:sec:eight:tail}); `output-layer refinement': \S\ref{sup:sec:eight:refine}.}\label{sup:tab:eightmoment}
\begin{adjustbox}{max width=\textwidth}
\begin{tabular}{lrrrrrrrrr}\toprule
State & $\rho$ & $u$ & $T$ & $q_x$ & $\sigma_{xx}$ & $m_{xxx}$ & $R_{xx}$ & $\Delta$ & $S_{xx}$\\\midrule
Five fields & 0.0710 & 1.4926 & 0.0754 & 2.8759 & 3.0891 & 105.88 & 141.09 & 138.64 & 140.43\\
Eight moments, original envelopes & 0.0036 & 0.1942 & 0.0034 & 0.8324 & 1.0055 & 10.170 & 11.428 & 10.761 & 11.178\\
Common-tail envelope, same weights & 0.0036 & 0.1942 & 0.0034 & 0.3513 & 0.6802 & 5.3742 & 7.1039 & 6.8825 & 7.0239\\
Common-tail envelope, output-layer refinement & 0.0008 & 0.0418 & 0.0007 & 0.3306 & 0.6654 & 1.8321 & 1.9656 & 1.7932 & 1.8823\\\midrule
DVM, 3200 to 6400 cells & 0.0048 & 0.1023 & 0.0046 & 0.2047 & 0.1798 & 0.4122 & 0.8558 & 0.6978 & 0.7867\\
\bottomrule\end{tabular}
\end{adjustbox}
\end{table}

The observation locations are those of the five-field case: 160 locations for each of $\rho$, $u$ and $T$, and 42 locations for each of $q_x$, $\sigma_{xx}$, $m_{xxx}$, $R_{xx}$ and $\Delta$, all at $t_f$. The 42 moment locations lie in $-0.104\leq x\leq0.134$, with spacings between $0.005$ and $0.0075$. The 690 observation values are obtained by linear interpolation of the 6400-cell reference. With $N_k$ locations and a scale $S_k$ equal to the largest absolute observation of field $k$, the observation losses are
\begin{equation}
\ell_k=\frac1{N_k}\sum_{i=1}^{N_k}\left[\frac{a_k[f_\theta](t_f,x_i)-a_k^{\rm ref}(x_i)}{S_k}\right]^2,
\qquad k\in\{\rho,u,T,q,\sigma,m,R,\Delta\}.\label{sup:eq:obs8}
\end{equation}
All moments in \eqref{sup:eq:obs8} are integrated from $f_\theta$ on the $48\times20\times20$ grid with the definitions \cref{eq:rho,eq:u,eq:T,eq:qx,eq:sigmaxx,eq:mxxx,eq:Rxx,eq:Delta}. The extremum summaries of the five-field continuation are no longer objective terms.

\subsubsection{Eight-moment representation.}\label{sup:sec:eight:rep}
The exponent in \eqref{sup:eq:f} is extended to
\begin{equation}
\psi_\theta=\psi_{\max}\tanh\!\left[\frac{A_q\phi_q+A_\sigma\phi_\sigma}{\psi_{\max}}\right]
+\psi_h\tanh\!\left[\frac{A_m\phi_m+A_R\phi_R+A_\Delta\phi_\Delta}{\psi_h}\right],\label{sup:eq:f8}
\end{equation}
with $\psi_h=0.35$ and the additional velocity kernels
\begin{equation}
\phi_m=\frac{\widehat c_x^{\,3}-\tfrac35\widehat c_x|\widehat{\boldsymbol c}|^2}{\sqrt{12/5}},\quad
\phi_R=\frac{(|\widehat{\boldsymbol c}|^2-7)(\widehat c_x^{\,2}-|\widehat{\boldsymbol c}|^2/3)}{\sqrt{56/3}},\quad
\phi_\Delta=\frac{|\widehat{\boldsymbol c}|^4-10|\widehat{\boldsymbol c}|^2+15}{\sqrt{120}}.\label{sup:eq:kernels8}
\end{equation}
Under the unit Maxwellian, the kernels in \eqref{sup:eq:kernels8} have unit norm, are mutually orthogonal, and are orthogonal to the collision invariants and to $\phi_q$ and $\phi_\sigma$. To first order, each new amplitude therefore changes one moment only: $\delta m_{xxx}=\sqrt{12/5}\,\rho_\theta T_\theta^{3/2}A_m$, $\delta R_{xx}=\sqrt{56/3}\,\rho_\theta T_\theta^{2}A_R$ and $\delta\Delta=\sqrt{120}\,\rho_\theta T_\theta^{2}A_\Delta$. The amplitudes are
\begin{equation}
\begin{aligned}
A_m&=\widehat t^{\,3}G_ha_h\tanh r_m,\quad
A_R=\widehat t^{\,2}G_ha_h\tanh r_R,\quad
A_\Delta=\widehat t^{\,2}G_ha_h\tanh r_\Delta,\\
G_h&=\exp[-(x/w_h)^4](1-4x^2)^2,
\end{aligned}
\label{sup:eq:amplitudes8}
\end{equation}
with $a_h=0.20$ and $w_h=0.22$; like $\psi_h$, these bounds are empirical choices. The outputs $(r_m,r_R,r_\Delta)$ are produced by a second linear layer acting on the last hidden layer. A third linear layer supplies corrections $(1-\widehat t)(1-4x^2)^2r_k^{c}$ that are added to the eight outputs $r_k$; they vanish at $t_f$ and at the boundaries. The heat-flux amplitude $A_q$ is multiplied by $1+0.1\tanh\vartheta$, with $\vartheta$ a scalar parameter. The two added layers and $\vartheta$ are initialised to zero, so that the extended representation initially reproduces the five-field distribution. It has 52\,753 parameters.

Seven physical components accompany the eight observation losses. Let $M[f_\theta]$ denote the Maxwellian of the moments integrated from $f_\theta$, and $\mathcal R=\partial_tf_\theta+v_x\partial_xf_\theta-(M[f_\theta]-f_\theta)/\mathrm{Kn}$. The components are the weighted mean-square residual \eqref{sup:eq:pde}; its relative form, with $\mathcal R$ divided by $|f_\theta|+|M[f_\theta]|+10^{-5}$; the mean squares of the three balances \eqref{sup:eq:balances}; the residual \eqref{sup:eq:pde} at the initial time, at $x=\pm0.025$ and $\pm0.075$; and the mean-square incoming-boundary discrepancy at $x=\pm\tfrac12$ for $t=t_f/4$, $t_f/2$ and $t_f$. The first five are evaluated on 64 fixed space--time points, the training points. A separate set of 64 points enters no objective term. In the first continuation below it supplies acceptance guards and, from the 6th and 23rd updates, linearised inequality constraints on the step; the output-layer refinement uses it only for assessment. A third set of 64 points is used by no stage. All quantities are evaluated in double precision. The fifteen components are denoted $\chi_k$, $k=1,\ldots,15$, and their values in the five-field state $\bar\chi_k$.

\subsubsection{Continuation with the original envelopes.}\label{sup:sec:eight:orig}
The five-field weights are retained and the added layers start from zero, so the first state of the extension is the five-field distribution. All 52\,753 parameters are then updated by a constrained least-squares continuation of the kind described in \cref{sup:sec:shocktube}. Each step is a column-scaled Levenberg--Marquardt proposal on the 690 observation residuals and the square roots of the seven physical components. It is projected so that no protected component increases to first order, limited by a trust radius on the change of the represented functions, and followed by backtracking. Later updates add steps of the temporal correction layer alone and, in the last nine updates, constraints on the time derivatives of the eight coefficients (Appendix~\ref{ext:app:supervised}). An update is accepted when $\sum_{k=1}^{15}\chi_k/\bar\chi_k$ and $\sum_{k\in\{m,R,\Delta\}}\ell_k/\bar\chi_k$ both decrease while the five lower-field losses and the seven physical components remain below $(1+10^{-6})\bar\chi_k$. Further guards compare the model with the reference at 399 positions, with its heat-flux and stress extrema and, in the last nine updates, on all 6400 cells. These reference guards prevent regression of any field and enter no gradient, but for this stage they make the reference an aid to selection rather than an independent test. The schedule was not fixed in advance: the step control was changed several times when progress stalled. The continuation comprises 170 accepted updates, during which $\sum_k\chi_k/\bar\chi_k$ falls from 15 to 0.89. It ended when the remaining admissible steps lowered the full-domain errors only by raising the upstream errors: between the 168th and 170th updates the squared errors of $m_{xxx}$, $R_{xx}$ and $\Delta$ in $x<-0.15$ rose by $1.6\%$, $2.5\%$ and $2.9\%$, and subsequent attempts were rejected.

\subsubsection{Diagnosis of the upstream tails.}\label{sup:sec:eight:diag}
The state after these updates reduces the third- and fourth-order discrepancies from above $100\%$ to $10$--$11\%$ (table~\ref{sup:tab:eightmoment}), but the remaining error is not distributed evenly. The fraction of the squared error lying upstream of the observed interval, $x<-0.104$, is $94\%$ for $q_x$, $91\%$ for $\sigma_{xx}$, $74\%$ for $m_{xxx}$, $64\%$ for $R_{xx}$ and $61\%$ for $\Delta$. There the third- and fourth-order moments develop a shoulder that is absent from the reference (figure~\ref{sup:fig:uptail}).
\begin{figure}\centering
\includegraphics[width=\textwidth]{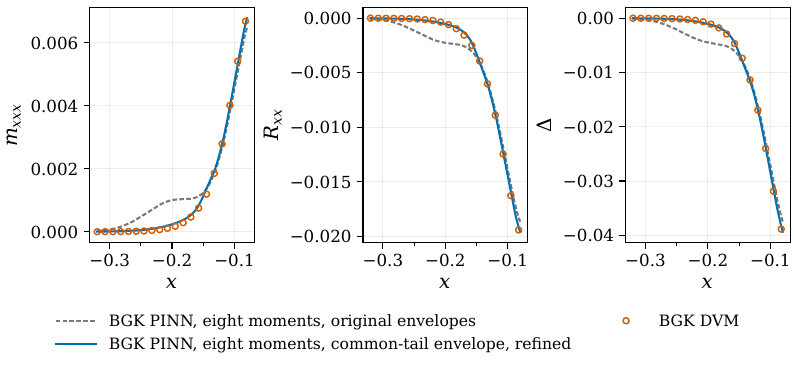}
\caption{Upstream detail of the third- and fourth-order moments of the two eight-moment states of figure~\ref{sup:fig:shocktube}. The moment observations lie in $x\geq-0.104$. With the original envelopes the amplitudes persist upstream and form a shoulder; the common-tail envelope \eqref{sup:eq:commontail}, shown after the output-layer refinement, approximately restores the exponential decay of the initial condition.}\label{sup:fig:uptail}
\end{figure}

The cause lies in the representation rather than in the optimisation. For $t_f\ll\mathrm{Kn}$ the solution is close to free transport of the initial Maxwellian, $f\simeq f(0,x-v_xt,\boldsymbol v)$. Far upstream, \eqref{eq:initial} gives $H_k(x)\simeq e^{2kx}$, so that the departure of $f$ from the left equilibrium, and with it every non-equilibrium moment, decays as $e^{-\kappa|x|}$ with $\kappa=2k=40$. The reference confirms this: on $-0.25\leq x\leq-0.20$ the logarithmic decay rates of $q_x$, $\sigma_{xx}$, $m_{xxx}$, $R_{xx}$ and $\Delta$ lie between $39.7$ and $39.9$, a fall by a factor of $7.3$ across the interval. In the model, however, the amplitudes are multiplied by the envelopes $G_w$ and $G_h$, neither of which has this decay. The super-Gaussian $G_w$ with $w=0.13$ truncates the heat flux and stress too rapidly: across the same interval the model values fall by factors of about $160$ and $50$ and change sign. The envelope $G_h$ with $w_h=0.22$ instead lets the third- and fourth-order amplitudes persist: their rates lie between $13.5$ and $15.1$, and at $x=-0.25$ the model values of $m_{xxx}$, $R_{xx}$ and $\Delta$ are $24$, $17$ and $17$ times those of the reference. No moment observation lies upstream of $x=-0.104$, and the sampled mean-square kinetic residual is almost insensitive to these small tail moments: correcting the tails, as described next, changes it by less than $0.06\%$. Neither the observations nor the residual can therefore correct the tail, which is set by the envelope.

\subsubsection{Common-tail envelope.}\label{sup:sec:eight:tail}
For $|x|>x_1$, both envelopes are replaced by envelopes that share the decay rate $\kappa$ of the initial condition. Let $G$ denote either $G_w$ or $G_h$, $y=|x|$ and $L(y)=-\mathrm d\ln G/\mathrm dy$. The corresponding common-tail envelope $E$ is defined by
\begin{equation}
E=G\quad(y\leq x_1),\qquad
-\frac{\mathrm d\ln E}{\mathrm dy}=
\begin{cases}
[1-s(\xi)]L(y)+\kappa\,s(\xi),&x_1<y<x_2,\\
\kappa,&y\geq x_2,
\end{cases}
\label{sup:eq:commontail}
\end{equation}
where $\xi=(y-x_1)/(x_2-x_1)$, $s(\xi)=\xi^3(10-15\xi+6\xi^2)$, $x_1=0.135$ and $x_2=x_1+1/\kappa=0.16$. Continuity at $x_1$ and integration of \eqref{sup:eq:commontail} determine the two envelopes $E_w$ and $E_h$ completely; for $y\geq x_2$, $E_w=\exp(b_w-\kappa y)$ and $E_h=\exp(b_h-\kappa y)$ with $b_w=4.2325$ and $b_h=5.5144$. Each envelope has three continuous derivatives, adds no parameter and uses no reference value: it depends only on $k$ and on $x_1$, which is placed at the downstream edge of the observed interval. Upstream, the unobserved interval $-0.135<x<-0.104$ therefore keeps the original envelopes. The need for the replacement was identified by comparison with the reference. Since every moment observation lies in $|x|<x_1$, the model values at those locations and the five moment losses are unchanged; the macroscopic observations extend over the whole domain, and their losses change by $-0.3\%$, $-0.01\%$ and $-0.6\%$ for $\rho$, $u$ and $T$. At the boundaries $E_w=1.4\times10^{-7}$ and $E_h=5.1\times10^{-7}$, whereas $G_h$ vanishes there; the incoming-boundary component is monitored accordingly.

With the weights unchanged, this replacement alone reduces the full-domain discrepancies of $m_{xxx}$, $R_{xx}$ and $\Delta$ from $10.2\%$, $11.4\%$ and $10.8\%$ to $5.4\%$, $7.1\%$ and $6.9\%$, and those of $q_x$ and $\sigma_{xx}$ from $0.83\%$ and $1.01\%$ to $0.35\%$ and $0.68\%$ (table~\ref{sup:tab:eightmoment}). Upstream of the observed interval, the squared errors of the three higher moments fall by $97\%$, $96\%$ and $97\%$. On $-0.25\leq x\leq-0.20$ the decay rate of $q_x$ becomes $33$ and those of the higher moments lie between $29$ and $31$; at $x=-0.25$ the model values of $q_x$, $m_{xxx}$, $R_{xx}$ and $\Delta$ are $1.5$, $2.6$, $1.8$ and $1.8$ times those of the reference. The stress decays at a rate of about $53$ on $-0.20\leq x\leq-0.16$ and changes sign at $x\simeq-0.23$, so that it has the wrong sign farther upstream. The envelope thus reduces the upstream errors of the higher moments by an order of magnitude without restoring the exact decay; the remaining difference reflects the variation of the network outputs across a region in which no observation constrains them. On the 64 separate points the mean-square residual changes by $-0.02\%$, the relative residual by $-0.2\%$ and the balance root-mean-square values by less than $0.02\%$.

\subsubsection{Output-layer refinement.}\label{sup:sec:eight:refine}
The envelope correction leaves the observed core unchanged, where the discrepancies of the higher moments are still $5$--$7\%$. A final refinement updates only the two linear output layers, $\boldsymbol\theta_o\in\mathbb R^{1032}$; the hidden layers, the temporal correction layer and $\vartheta$ are fixed. It reads the reference only through the 690 observations. Its step design followed six short trial runs that read only training quantities; its acceptance rule and attempt budget, and the assessment criteria on the separate points, were then fixed before the reported run. Let $\chi_k^0$ denote the component values at the start of the refinement. Twelve components are protected: the losses of $\rho$, $u$, $T$, $q_x$ and $\sigma_{xx}$ and the seven physical components, each of which must satisfy $\chi_k\leq(1+10^{-6})\chi_k^0$ throughout. With $\Lambda_0=\ell_m^0+\ell_R^0+\ell_\Delta^0$, the least-squares system $A\boldsymbol d\simeq-\boldsymbol z$ for the step $\boldsymbol d$ of $\boldsymbol\theta_o$, with $A=\partial\boldsymbol z/\partial\boldsymbol\theta_o$, stacks in $\boldsymbol z$ four blocks of rows:
\begin{enumerate}
\item[(i)] the 126 residuals $[a_k(x_i)-a_k^{\rm ref}(x_i)]/(S_k\sqrt{N_k})$ of $m_{xxx}$, $R_{xx}$ and $\Delta$, whose squares sum to $\ell_m+\ell_R+\ell_\Delta$;
\item[(ii)] the corresponding 564 residuals of the five lower fields, those of field $k$ multiplied by $(\Lambda_0/\ell_k)^{1/2}$ with the current $\ell_k$;
\item[(iii)] changes of the bounded outputs $\tanh r_k$, $k\in\{q,\sigma,m,R,\Delta\}$, from their initial values at $|x|=0.16,0.17,\ldots,0.50$ and $t\in\{t_f/2,t_f\}$;
\item[(iv)] changes of all eight bounded outputs at $x=-0.45,-0.44,\ldots,0.45$ and $t/t_f\in\{\tfrac14,\tfrac12,\tfrac34,1\}$.
\end{enumerate}
In blocks (iii) and (iv) each output is divided by its root-mean-square value over the 42 moment locations, and the blocks are scaled so that root-mean-square relative changes of $0.02$ and $0.5$, respectively, cost $\Lambda_0$. Block (ii) is required because protecting a least-squares loss through its gradient alone does not prevent pointwise changes orthogonal to the current residual. Block (iii) preserves the unobserved tails set by the envelope, and block (iv) limits changes at earlier times, where no observation exists. The proposal and its constrained correction are
\begin{align}
\boldsymbol d_0&=-B^{-1}A^T\boldsymbol z,\qquad B=A^TA+\lambda_dI,\qquad
\lambda_d=\mu\operatorname{mean}\{\operatorname{diag}(A^TA)\},\label{sup:eq:olgn}\\
\boldsymbol d&=\operatorname*{arg\,min}_{\boldsymbol d'}\ \tfrac12(\boldsymbol d'-\boldsymbol d_0)^TB(\boldsymbol d'-\boldsymbol d_0)
\quad\text{subject to}\quad\boldsymbol g_k\cdot\boldsymbol d'\leq\beta_k\ \ (k\ \text{protected}),\label{sup:eq:olqp}
\end{align}
with $\boldsymbol g_k=\partial\chi_k/\partial\boldsymbol\theta_o$, $\tau_k=\tfrac14\max[(1+10^{-6})\chi_k^0-\chi_k,0]$ and $\beta_k=\tau_k-0.1\max(\boldsymbol g_k\cdot\boldsymbol d_0-\tau_k,0)$. Each protected component may thus use, to first order, a quarter of its remaining margin; where the unconstrained proposal would use more, a tenth of the excess is reversed. Problem \eqref{sup:eq:olqp} has the form of \eqref{sup:eq:direction} and is solved through its dual, $\boldsymbol d=\boldsymbol d_0-B^{-1}\Gamma^T\boldsymbol\lambda$ with $\boldsymbol\lambda\geq0$, by projected Gauss--Seidel iteration; the rows of $\Gamma$ are the $\boldsymbol g_k$.

Step lengths $\alpha=2^{-j}$, $j=0,\ldots,13$, are tried in turn. The update $\boldsymbol\theta_o\leftarrow\boldsymbol\theta_o+\alpha\boldsymbol d$ is accepted if all twelve protected components remain within their limits, $\sum_{k\in\{m,R,\Delta\}}\ell_k/\bar\chi_k$ and $\sum_{k=1}^{15}\chi_k/\bar\chi_k$ both decrease, and the root-mean-square relative change in block (iii) does not exceed $0.02$. The damping factor $\mu$ starts at $10^{-3}$, is divided by three after an accepted update, down to $10^{-10}$, and is multiplied by ten after a rejected attempt. The gradients $\boldsymbol g_k$, obtained by automatic differentiation, agree with a central difference along a test direction to within a relative difference of $5\times10^{-5}$. Of the stopping rules (40 attempts, 60 accepted updates, a decrease of the higher-moment sum below $0.1\%$ over eight updates, three hours, $\mu>10^{6}$), only the attempt cap was reached: 35 of the 40 attempts were accepted, and the last accepted updates still reduced the losses slowly. The refinement is therefore not converged. Each accepted update reduces all eight observation losses by nearly the same factor, so that further updates would continue towards interpolation of the observation values, including the discretisation error of the reference. The refinement required 941~s on four threads. Its history, settings and two variants are given in Appendix~\ref{ext:app:supervised} and figure~\ref{sup:fig:outputlayer_losses}.

\subsubsection{Results.}\label{sup:sec:eight:results}
Table~\ref{sup:tab:eightmoment} lists the four states; figure~\ref{sup:fig:shocktube} displays the first, the second and the final one. The refinement reduces all eight observation losses by $94$--$95\%$. The full-domain discrepancies of $m_{xxx}$, $R_{xx}$ and $\Delta$ become $1.83\%$, $1.97\%$ and $1.79\%$, and those of $q_x$ and $\sigma_{xx}$ $0.33\%$ and $0.67\%$. Within the observed interval the discrepancies are $0.03\%$ and $0.04\%$ for $q_x$ and $\sigma_{xx}$, and $1.3\%$, $1.7\%$ and $1.7\%$ for the higher moments. The heat-flux maximum is $0.043285$, compared with $0.043303$ in the reference, and the stress minimum is $-0.037531$, compared with $-0.037541$. The minima of $m_{xxx}$, $R_{xx}$ and $\Delta$ are $-0.009459$, $-0.020405$ and $-0.041933$, compared with $-0.009549$, $-0.020711$ and $-0.042586$.

Four qualifications apply to these values. First, the reference is not independent of the result. The observation values are taken from the reference used for the comparison, the full-domain norms include the observation locations, and the acceptance guards of the first continuation used the reference beyond those locations. Excluding the reference nodes that bracket an observation location gives $0.34\%$, $0.68\%$, $1.86\%$, $1.98\%$ and $1.80\%$ for $q_x$, $\sigma_{xx}$, $m_{xxx}$, $R_{xx}$ and $\Delta$; because the fields are smooth and the observations are only 32 to 48 nodes apart, this is not an independent test. The final discrepancies of $\rho$, $u$ and $T$ are smaller than the refinement difference of the reference itself; they measure the fit to this reference and not accuracy beyond its resolution. Against the first-order extrapolate of the 3200- and 6400-cell solutions, the discrepancies of $q_x$, $\sigma_{xx}$, $m_{xxx}$, $R_{xx}$ and $\Delta$ are $0.38\%$, $0.67\%$, $1.88\%$, $2.65\%$ and $2.38\%$ (Appendix~\ref{ext:app:reference}).

Second, the discrepancies of the lower and the higher moments have different origins. For $m_{xxx}$, $R_{xx}$ and $\Delta$, $48\%$, $69\%$ and $78\%$ of the remaining squared error lies inside the observed interval; the largest pointwise errors are $1.7\%$, $1.6\%$ and $1.6\%$ of the field extrema, at $x\simeq-0.15$, $-0.06$ and $-0.05$. For $q_x$ and $\sigma_{xx}$, by contrast, $99\%$ of the remaining squared error lies upstream of the observed interval, and $72\%$ and $41\%$ of the total in the blending zone $0.135<|x|<0.16$. There the model lies below the reference, by $10\%$ and $16\%$ at $x=-0.15$; the largest pointwise errors, $0.58\%$ and $0.81\%$ of the field extrema, occur at the same position. The values $0.33\%$ and $0.67\%$ are therefore errors of the envelope region, whereas $0.03\%$ and $0.04\%$ measure the fit to the supplied values.

Third, the tails remain the least accurate part of the solution. For $x<-0.15$ the regional discrepancies of the final state are $10\%$, $27\%$, $21\%$, $12\%$ and $8\%$ for $q_x$, $\sigma_{xx}$, $m_{xxx}$, $R_{xx}$ and $\Delta$ (they were $32\%$, $42\%$, $198\%$, $124\%$ and $132\%$ with the original envelopes), and for $x>0.15$ they are $4\%$, $3\%$, $41\%$, $51\%$ and $49\%$ (previously $29\%$, $27\%$, $12\%$, $56\%$ and $48\%$). In the final state $\sigma_{xx}$ has the wrong sign upstream of $x\simeq-0.23$, and $R_{xx}$ and $\Delta$ have the wrong sign downstream of $x\simeq0.19$, as they already had with the original envelopes; there the reference moments are below $0.5\%$ of their extrema. The common-tail envelope thus reduces the upstream errors of the third- and fourth-order moments by an order of magnitude and those of the heat flux and stress by factors of about $3$ and $1.6$. It corrects the downstream heat flux and stress, but not the downstream third- and fourth-order tails; the downstream $m_{xxx}$ tail is degraded. Outside the observed interval the output-layer refinement is not driven by data; there the squared error of $q_x$ increases by $18\%$ downstream and by $3\%$ upstream, which is included in the full-domain values.

Fourth, the result is a single training history whose first continuation followed no schedule fixed in advance, and whose refinement ended at its attempt cap. A tolerance of $0.1\%$ on the increase of the five sampled physical measures on the separate points had been set before the run. It was exceeded there and, when the same check was applied afterwards, on the unused points, as detailed next.

The physical components are controlled but not small. On the training points the refinement decreases the mean-square residual by $0.2\%$, the relative residual by $2.5\%$, the mean squares of the mass, momentum and energy balances by $6.0\%$, $2.1\%$ and $1.0\%$, and the incoming-boundary discrepancy by $9.5\%$; the initial-time residual remains at its limit. On the 64 separate points the mean-square residual increases by $0.6\%$ and the relative residual decreases by $2.3\%$, while the root-mean-square values of the mass, momentum and energy balances change by $-2.0\%$, $-0.6\%$ and $+0.8\%$. On the 64 points used by no stage the corresponding changes are $-0.6\%$, $+0.4\%$, $-2.1\%$, $-0.9\%$ and $+0.2\%$. The pre-set tolerance was thus exceeded by the residual and the energy balance on the separate points; applied afterwards to the unused points, it was exceeded by the relative residual and the energy balance. In absolute terms, the root-mean-square mass, momentum and energy balance residuals of the final state are $0.16$, $0.13$ and $0.14$ times the root-mean-square time-derivative terms on the training points, $0.18$, $0.21$ and $0.18$ times on the separate points, and $0.19$, $0.23$ and $0.23$ times on the unused points. For the five-field state the corresponding ratios are $0.40$, $0.35$ and $0.37$; $0.38$, $0.34$ and $0.36$; and $0.35$, $0.40$ and $0.31$. The weighted mean-square residual \eqref{sup:eq:pde} of the five-field state is larger than that of the final state by factors of $8.7$, $2.7$ and $3.4$ on the three sets. With the weights of \eqref{sup:eq:pde}, the root-mean-square residual of the final state exceeds the root-mean-square collision term $(M[f_\theta]-f_\theta)/\mathrm{Kn}$ by factors of $62$, $95$ and $121$. At this short time the sampled residual therefore does not resolve the collision contribution, and the dependence on $\mathrm{Kn}$ enters the data-assisted state through the observations. The data-assisted state is a positive kinetic representation fitted to the observed moments under residual and balance controls; it is not a converged solution of the kinetic equation.\label{sup:sec:eightmoment:end}

\subsection{Moment selection and target count}\label{sup:sec:ablation}
Table~\ref{sup:tab:ablation} separates macro locks, individual moment anchors and joint anchors. The ablations share the supervised warm start; zero anchor count removes the corresponding terms from that stage. Macro targets alone leave large heat-flux and stress errors. A heat-flux target improves that observable without resolving stress, and conversely for stress-only anchoring. 
\begin{table}
\centering
\caption{Data-assisted BGK shock-tube ablations. Errors are relative $L^2$ fractions. $N_A$ counts locations per moment and $N_M$ macro-target locations. In the case labels, the anchored moments ($q$, $\sigma$ or both, `qsig') are followed by $N_A$ and `m' by $N_M$, where `m' denotes macroscopic targets; `macro32' has macroscopic targets only and `full42/m160' has the counts of the base case. The cases share the supervised warm start.}
\label{sup:tab:ablation}
\small
\begin{adjustbox}{max width=\textwidth}
\begin{tabular}{lccccccc}
\toprule
Case & $N_A$ & $N_M$ & $\rho$ & $u_x$ & $T$ & $q_x$ & $\sigma_{xx}$\\
\midrule
none & 0 & 0 & $2.352\times10^{-2}$ & $6.026\times10^{-1}$ & $1.001\times10^{-2}$ & $4.252\times10^{-1}$ & $5.504\times10^{-1}$\\
macro32 & 0 & 32 & $3.622\times10^{-3}$ & $1.558\times10^{-2}$ & $2.333\times10^{-3}$ & $8.789\times10^{-1}$ & $1.508$\\
$q$-only8/m32 & 8 & 32 & $6.375\times10^{-3}$ & $1.696\times10^{-2}$ & $4.432\times10^{-3}$ & $6.227\times10^{-2}$ & $1.192$\\
$\sigma$-only8/m32 & 8 & 32 & $3.610\times10^{-3}$ & $1.714\times10^{-2}$ & $2.349\times10^{-3}$ & $9.593\times10^{-1}$ & $8.815\times10^{-2}$\\
qsig8/no-macro & 8 & 0 & $1.369\times10^{-1}$ & $5.415\times10^{-1}$ & $2.304\times10^{-2}$ & $7.952\times10^{-2}$ & $8.870\times10^{-2}$\\
qsig4/m32 & 4 & 32 & $5.443\times10^{-3}$ & $1.530\times10^{-2}$ & $4.193\times10^{-3}$ & $9.634\times10^{-2}$ & $2.924\times10^{-1}$\\
qsig8/m32 & 8 & 32 & $4.952\times10^{-3}$ & $1.865\times10^{-2}$ & $3.053\times10^{-3}$ & $9.520\times10^{-2}$ & $1.078\times10^{-1}$\\
qsig16/m32 & 16 & 32 & $5.076\times10^{-3}$ & $1.650\times10^{-2}$ & $3.588\times10^{-3}$ & $8.867\times10^{-2}$ & $6.438\times10^{-2}$\\
qsig8/m64 & 8 & 64 & $4.370\times10^{-3}$ & $1.472\times10^{-2}$ & $2.470\times10^{-3}$ & $8.621\times10^{-2}$ & $7.602\times10^{-2}$\\
full42/m160 & 42 & 160 & $3.680\times10^{-3}$ & $1.281\times10^{-2}$ & $2.568\times10^{-3}$ & $8.293\times10^{-2}$ & $5.038\times10^{-2}$\\
\bottomrule
\end{tabular}
\end{adjustbox}
\end{table}
Figure~\ref{sup:fig:efficiency} compares 4, 8 and 16 anchors at $N_M=32$: stress discrepancy falls from 0.2924 to 0.06438, whereas the heat-flux change is smaller. The target-count ablations use the AdamW protocol; the base-case result in table~\ref{sup:tab:shockcases} additionally includes the constrained continuation. All ablations use the five-field representation and are compared with the coarse solution that supplied their observations; the eight-moment extension of \cref{sup:sec:eightmoment} is applied to the base case only. Each entry of tables~\ref{sup:tab:shockcases} and~\ref{sup:tab:ablation} results from a single training history; no independent-seed ensemble is available for these configurations.\label{sup:sec:ablation:single}
\begin{figure}\centering
\includegraphics[width=.9\textwidth]{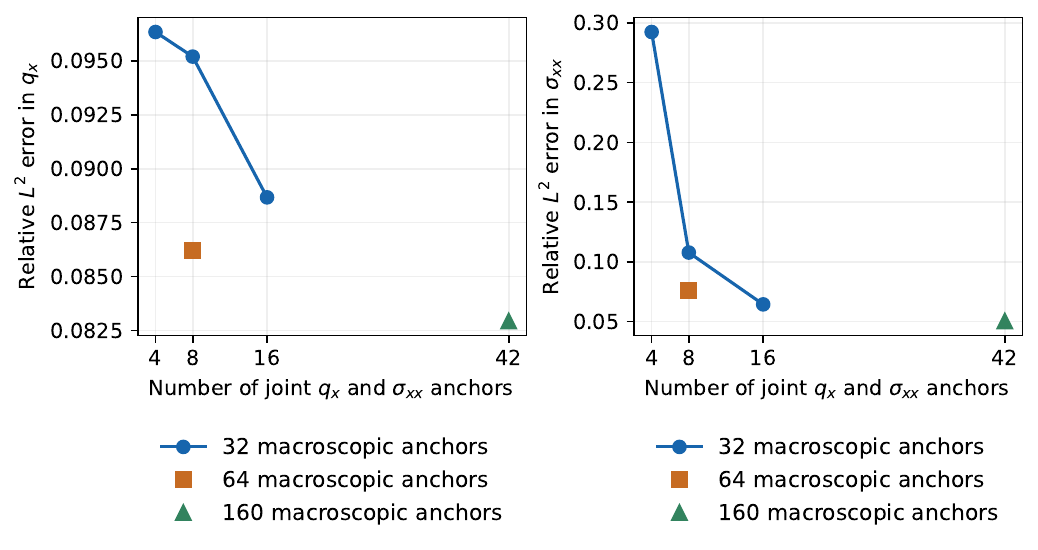}
\caption{Data-assisted BGK shock-tube anchor study. DVM macroscopic values and heat-flux/stress targets, including reference amplitudes, enter training. Lines join tested target counts; no independent-seed uncertainty is implied.}\label{sup:fig:efficiency}
\end{figure}

\subsection{Supervised fourth-order output in a stationary BGK shock}\label{sup:sec:closure}
The Mach-2 base model is positive and flux locked, with DVM targets for density, velocity, temperature, heat flux and stress. Mass, momentum and energy fluxes are constrained to their analytic upstream values. Its kinetic residual uses the Maxwellian of the moments integrated from $f_\theta$ and the relaxation scale $\mathrm{Kn}_{\rm eff}=0.0125$ of the network coordinate $x\in[-1/2,1/2]$. In this coordinate the upstream mean free path is $\lambda_1=\mathrm{Kn}_{\rm eff}$, so that the plotted coordinate is $x/\lambda_1=80x$ and the relaxation time is unity in units of $\lambda_1$. The reference uses 1600 spatial cells, $97\times19\times19$ velocity nodes on $[-12,12]^3$, conservative discrete-Maxwellian inflow and collision targets, and density-midpoint shock fixing. In its upstream-mean-free-path coordinate, the reference collision time is unity; this is not the same coordinate convention as the neural $\mathrm{Kn}_{\rm eff}$. The Mach-2 base state and the closure fits reported below are single training histories.\label{sup:sec:closure:single}

Write $e=|\boldsymbol c|^2$ with $\boldsymbol c=\boldsymbol v-u\boldsymbol e_x$. The reported closures are
\begin{align}
m_{xxx}^{\rm cl}&=\int(c_x^3-3c_xe/5)(f-M_*)\,\mathrm d\boldsymbol v,\label{sup:eq:mcl}\\
R_{xx}^{\rm cl}&=\int e(c_x^2-e/3)(f-M_*)\,\mathrm d\boldsymbol v-7T\sigma_{xx}.\label{sup:eq:Rcl}
\end{align}
In \eqref{sup:eq:mcl} and \eqref{sup:eq:Rcl}, $u,T,\sigma_{xx}$ are integrated moments. The DVM uses its conservative discrete Maxwellian $M_*=M_d$; neural post-processing uses an analytic Maxwellian sampled on its quadrature. These finite-grid subtractions differ from each other and from the full-distribution convention of \cref{sup:sec:eightmoment} and of the equation-trained comparisons.

The fourth-order output augments the distribution-derived value $R_{xx}^f$:
\begin{equation}
\widehat R_{xx}^{\rm cl}=R_{xx}^f+3G_R(x)\tanh\left[\sum_{\ell=1}^{17}a_\ell\exp\left(-\frac{(x-x_\ell)^2}{w_R^2}\right)\right].\label{sup:eq:head}
\end{equation}
The coefficients $a_\ell$ and logarithmic width $\zeta_R$ are trainable. The width $w_R=\operatorname{clip}(e^{\zeta_R},0.004,0.080)$ starts at 0.018, and $G_R=\exp[-((x-x_0)/\max(3w_0,10^{-6}))^2]$ uses the centre $x_0$ and width $w_0$ of the reference density profile: $x_0$ is the reference node at which $\rho$ is closest to $(\rho_1+\rho_2)/2$, and $w_0=\max(|x_{90}-x_{10}|/4,\,0.025)$, where $x_{10}$ and $x_{90}$ are the nodes at which $\rho$ is closest to $\rho_1+0.1(\rho_2-\rho_1)$ and $\rho_1+0.9(\rho_2-\rho_1)$; here $x_0=-3.1\times10^{-4}$ and $w_0=0.031$. The prescribed centres $x_\ell/\lambda_1$ are $\{-8,-6,-4,-3,-2,-1.5,-1,-.5,0,.5,1,1.5,2,3,4,6,8\}$. The Gaussian gate suppresses the correction in the plateaus. This calibrated output does not change the fourth moment of $f_\theta$ or establish a corrected kinetic distribution.

The closure trainer uses AdamW parameter groups for the base network and its moment heads, with stage-dependent trainable blocks and loss ramps. The auxiliary coefficients and learned width form a separate parameter group with zero weight decay. This is a supervised continuation of an existing state, not a cold start or the self-scaled Broyden (SSB) algorithm used by the equation-trained models. The fourth-order contribution contains weighted point, normalised-moment, slope, feature and tail discrepancies:
\begin{equation}
\mathcal L_R=w_{R,p}\mathcal L_{R,p}+w_{R,n}\mathcal L_{R,n}+w_{R,s}\mathcal L_{R,s}+w_{R,f}\mathcal L_{R,f}+w_{R,t}\mathcal L_{R,t}.\label{sup:eq:headloss}
\end{equation}
The $w_{R,\cdot}$ in \eqref{sup:eq:headloss} are stage weights. The point term averages $(1+8|R^{\rm ref}|/(S_R+10^{-12}))[(\widehat R-R^{\rm ref})/S_R]^2$, with $S_R=\max(\|R^{\rm ref}\|_\infty,10^{-6})$ and DVM target $R^{\rm ref}$. The normalised-moment term applies Huber loss with transition 0.025 to $R/d_R$ and its reference counterpart. Here $d_R=\rho\max(T,\epsilon_R)^2+\epsilon_R$, with $\epsilon_R=10^{-8}$. Slopes use Huber transition 0.05 and scale $S_R$. Features compare extrema, softmax-based extremum positions and integrated area; the tail term compares far-field fourth-order values. Other configured field/residual losses remain active; no coefficient-squared penalty is assumed.

Figure~\ref{sup:fig:m2} and table~\ref{sup:tab:m2} report the selected endpoint. They include the BGK collision productions $Q_{xx}^{\rm BGK}=-\sigma_{xx}$ and $Q_x^{\rm BGK}=-2q_x$, evaluated with unit relaxation time in $x/\lambda_1$; for BGK these are fixed multiples of $\sigma_{xx}$ and $q_x$, and their discrepancies equal those of $\sigma_{xx}$ and $q_x$. The fourth-order value is the calibrated output of \eqref{sup:eq:head}; lower moments come from the kinetic representation.
\begin{figure}\centering
\includegraphics[width=\textwidth,height=.82\textheight,keepaspectratio]{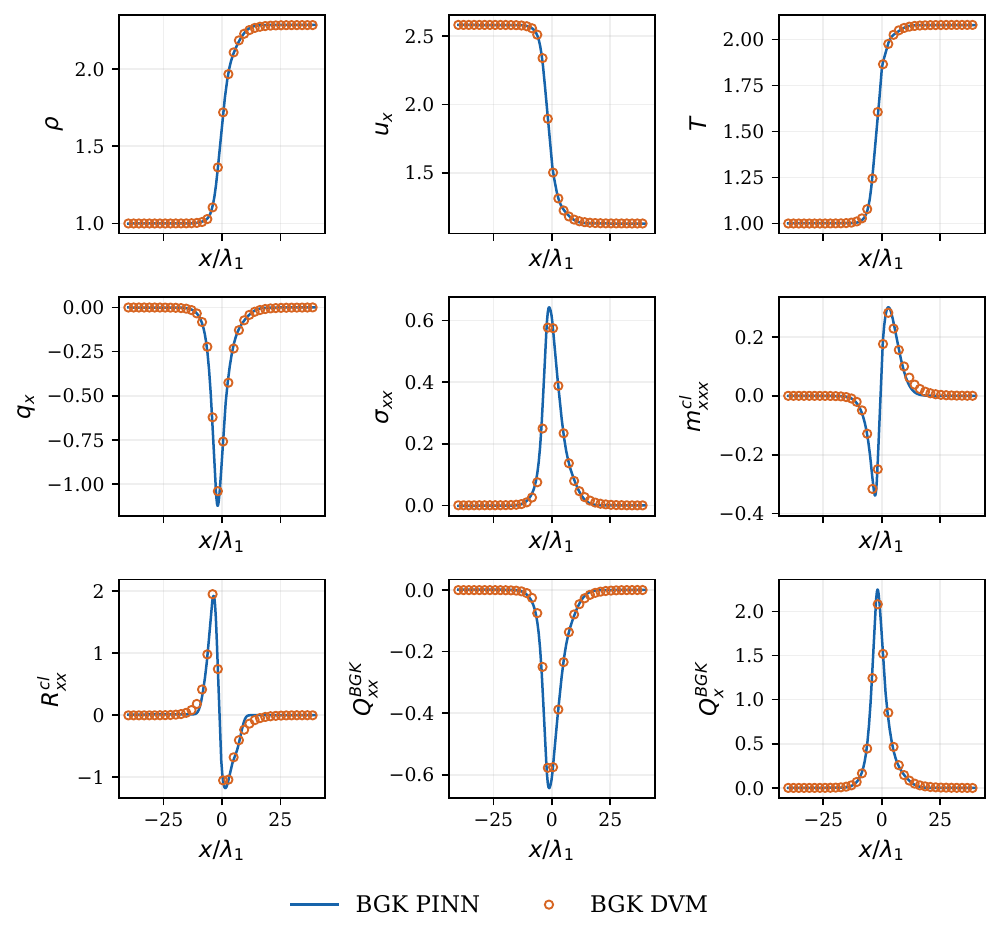}
\caption{Data-assisted Mach-2 BGK stationary shock. DVM macro, heat-flux and stress targets train the kinetic base; DVM fourth-order targets calibrate the auxiliary $\widehat R_{xx}^{\rm cl}$ output. Other fields are integrated from the positive distribution. BGK production quantities, $Q_{xx}^{\rm BGK}=-\sigma_{xx}$ and $Q_x^{\rm BGK}=-2q_x$, are dependent moment identities, not independent validation tests.}\label{sup:fig:m2}
\end{figure}
\begin{table}
\centering
\caption{Data-assisted BGK Mach-2 selected endpoint: relative $L^2$ fractions. The fourth-order entry is the supervised auxiliary output; all entries belong to the selected state.}
\label{sup:tab:m2}
\normalsize
\renewcommand{\arraystretch}{1.25}
\setlength{\tabcolsep}{9pt}
\begin{tabular}{@{}lclc@{}}
\toprule
\textbf{Quantity} & \textbf{rel. $L^2$ error} &
\textbf{Quantity} & \textbf{rel. $L^2$ error} \\
\midrule
$\rho$              & $2.53\times10^{-3}$ &
$\sigma_{xx}$       & $9.40\times10^{-2}$ \\

$u_x$               & $2.98\times10^{-3}$ &
$m_{xxx}^{cl}$      & $1.17\times10^{-1}$ \\

$T$                 & $3.88\times10^{-3}$ &
$R_{xx}^{cl}$       & $1.12\times10^{-1}$ \\

$q_x$               & $7.14\times10^{-2}$ &
$Q_{xx}^{BGK}$      & $9.40\times10^{-2}$ \\

                   &                       &
$Q_x^{BGK}$         & $7.14\times10^{-2}$ \\
\bottomrule
\end{tabular}
\end{table}
The common-parent direct-$f$ ablation in table~\ref{sup:tab:nof} removes the optional distribution-probe term while retaining moment and closure supervision. It tests the incremental role of that term, not removal of all reference information.
\begin{table}
\centering
\caption{Common-parent supervised closure-stage ablation. $N_f$ counts optional velocity-space probes; a station is one spatial position of the reference grid, at which the stated number of probes is selected. Removing direct $f$ loss retains DVM macro, moment and closure supervision.}
\label{sup:tab:nof}
\footnotesize
\begin{adjustbox}{max width=\textwidth}
\begin{tabular}{lccc}
\toprule
Case & $N_f$ & rel. $L^2$ error in $R_{xx}^{cl}$ & rel. $L^2$ error in $m_{xxx}^{cl}$\\
\midrule
No direct $f$ loss & 0 & $1.108\times10^{-1}$ & $1.168\times10^{-1}$\\
$R$-sensitive probes, 32 per station & 544 & $1.109\times10^{-1}$ & $1.168\times10^{-1}$\\
$R$-sensitive probes, 64 per station & 1088 & $1.109\times10^{-1}$ & $1.168\times10^{-1}$\\
$R$-sensitive probes, 128 per station & 2176 & $1.109\times10^{-1}$ & $1.168\times10^{-1}$\\
Random probes, 128 per station & 2176 & $1.109\times10^{-1}$ & $1.168\times10^{-1}$\\
\bottomrule
\end{tabular}
\end{adjustbox}
\end{table}
Post-hoc pruning changes the retained-centre count from 17 to 5 without retraining. Its $R$ discrepancy rises from 0.11965 to 0.57075, while $m_{xxx}^{\rm cl}$ stays at 0.11795 and lower moments remain unchanged. This supports an auxiliary $R$-specific effect, not recovery of the distribution. The pruning baseline is a separately evaluated state, not the selected endpoint in table~\ref{sup:tab:m2}.

\subsection{Velocity-space cancellation and reference-informed residual weighting}\label{sup:sec:rphys}
Figure~\ref{sup:fig:tails} examines $H_R(\boldsymbol c)(f-M_d)$, where $H_R=(e-7T)(c_x^2-e/3)$. At finite quadrature, \eqref{sup:eq:Rcl} obeys $R_{xx}^{\rm cl}=\int H_R(f-M_*)\,\mathrm d\boldsymbol v-7T\sigma_{xx}[M_*]$, with $\sigma_{xx}[M_*]=\int(c_x^2-e/3)M_*\,\mathrm d\boldsymbol v$. Matching discrete mass, momentum and energy does not force this equilibrium stress to vanish.
\begin{figure}\centering
\includegraphics[width=.8\textwidth,height=.32\textheight,keepaspectratio]{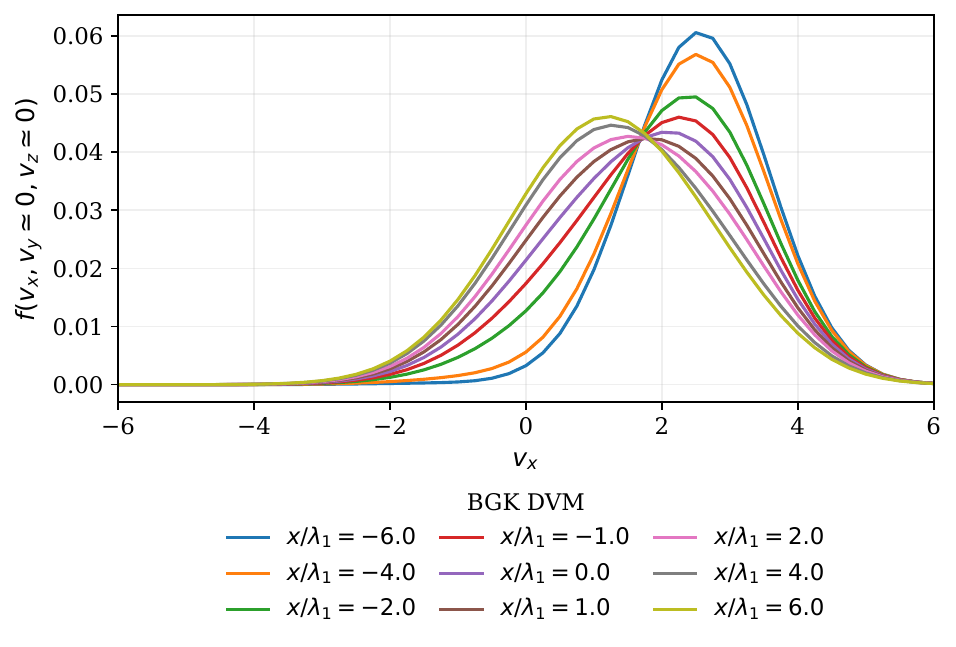}\par
\includegraphics[width=.8\textwidth,height=.32\textheight,keepaspectratio]{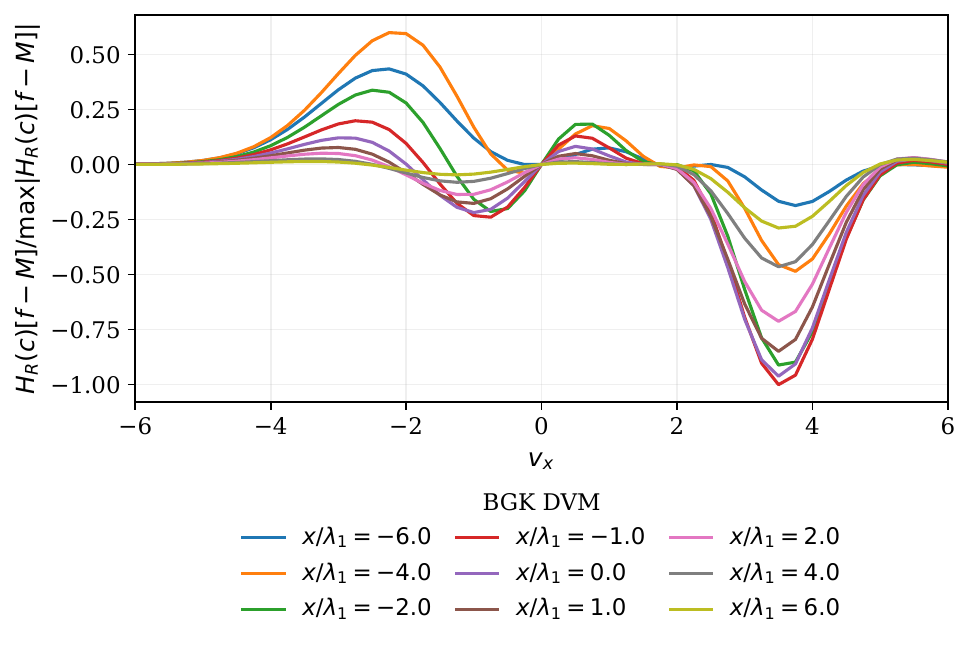}
\caption{BGK Mach-2 reference velocity slices and signed fourth-order contributions supporting the data-assisted closure study. Straight segments join native saved nodes without smoothing. These reference diagnostics explain sensitivity of the calibrated moment; they do not show a neural distribution corrected by the auxiliary head.}\label{sup:fig:tails}
\end{figure}
Another ablation disables direct $R$-target losses and adds
\begin{equation}
\mathcal L_{R,\rm phys}=\frac1{N_x}\sum_i\left[\frac{\sum_jw_jH_R(\boldsymbol c_{ij})\mathcal R_{ij}}{D_R}\right]^2,
\quad D_R=\max\left(\frac{S_R}{\max(\mathrm{Kn}_{\rm eff},10^{-12})},10^{-8}\right).\label{sup:eq:rphys}
\end{equation}
In \eqref{sup:eq:rphys}, $N_x$ counts sampled positions, $w_j$ are velocity weights, and $\mathcal R_{ij}=v_{x,j}\partial_xf_\theta-(M_\theta-f_\theta)/\mathrm{Kn}_{\rm eff}$, with the latent Maxwellian $M_\theta$ of the network's macroscopic outputs. The kernel uses integrated velocity and temperature. The recorded reference scale is $S_R=2.022136926651001$, giving $D_R\simeq161.77095$ in implemented arithmetic. Reference $R$ also contributes to the residual sampling pool. These are therefore reference-informed runs, not equation-only experiments. The four matched runs differ only in the weight $w_{R,\rm phys}$ of $\mathcal L_{R,\rm phys}$, held at 0 in the control run, where the term is not evaluated, and at 5, 10 and 20 in the other runs; all four have zero weight on the direct $R$ targets and use the same seed, 700 updates and stage boundaries at updates 1 and 40; the final residual-active stage follows update 40. Table~\ref{sup:tab:rphys} shows a modest $R$ improvement accompanied by worse heat-flux/stress discrepancies; it does not establish inevitable failure of residual weighting in other representations.
\begin{table}\centering\small
\caption{Reference-informed BGK residual ablation: relative $L^2$ fractions. All rows retain the supervised lineage and reference-derived residual scale/sampling. Control: $w_{R,\rm phys}=0$; weight $n$: $w_{R,\rm phys}=n$.}\label{sup:tab:rphys}
\begin{tabular}{lrrrr}\toprule
Case & $q_x$ & $\sigma_{xx}$ & $m_{xxx}^{\rm cl}$ & $R_{xx}^{\rm cl}$\\\midrule
Control & 0.019606 & 0.049599 & 0.10204 & 0.91644\\
Weight 5 & 0.022953 & 0.052333 & 0.10346 & 0.90672\\
Weight 10 & 0.025273 & 0.053374 & 0.10465 & 0.90585\\
Weight 20 & 0.027639 & 0.054073 & 0.10643 & 0.90528\\
\bottomrule\end{tabular}\end{table}
\subsection{Supervised scalar fourth-order calibration}\label{sup:sec:delta}
The scalar excess $\Delta=\int|\boldsymbol c|^4f\,\mathrm d\boldsymbol v-15\rho T^2$ has a separate supervised output fit. Figure~\ref{sup:fig:delta} compares its saved profiles. Uniform 17-, 33- and 65-target fits give relative discrepancies 0.101, 0.00150 and 0.000315; the full scalar fit gives 0.000307. These are calibration results, not fourth moments recovered by an independently trained kinetic solution. The recorded 65-target peak-only fit gives 0.655, demonstrating sensitivity to placement within this model. The hybrid placement, recorded as a combination of amplitude-, gradient- and uniformly placed targets, gives 0.897, 0.414, 0.0985 and 0.0140 with 9, 17, 33 and 65 targets; the rule that distributed the targets among these three kinds was not retained with the saved results. In figure~\ref{sup:fig:delta} the targets are called scalar anchors.
\begin{figure}\centering
\includegraphics[width=.95\textwidth]{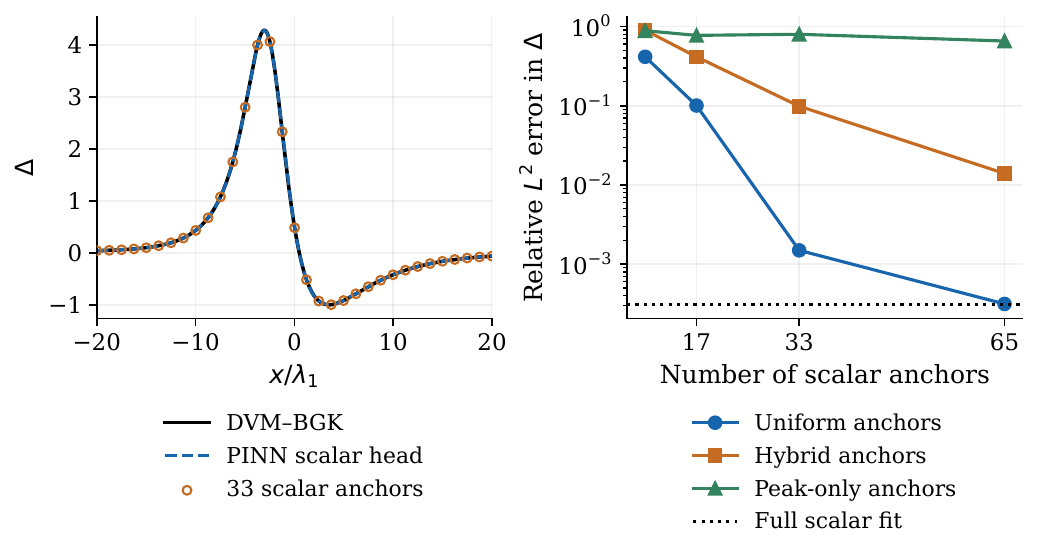}
\caption{Supervised scalar fourth-order output in the BGK Mach-2 shock. The auxiliary $\Delta$ output is fitted to DVM scalar targets; the dense-fit comparator is a calibration limit, not an equation-only prediction. This output is separate from the anisotropic $R$ head and from the integrated moments of equation-trained models. Panel (a) shows $|x/\lambda_1|\leq20$, the window on which the uniform 33-target discrepancy of 0.00150 is evaluated; over the full domain $|x/\lambda_1|\leq40$ it is 0.0226. Panel (b) compares uniform, hybrid (amplitude-, gradient- and uniformly placed) and peak-only target placements, described in \S\ref{sup:sec:delta}.}\label{sup:fig:delta}
\end{figure}
The closure studies distinguish calibration of an auxiliary observable from recovery of that moment by the kinetic distribution. The equation-trained calculations examine how transport structure and moment self-consistency recover the same physical quantities.

% Characteristic-integral shock-tube PINN.
\section{Equation-trained characteristic PINNs for shock tubes}
\label{sec:transient_characteristic_pinn}
The characteristic-integral PINN is independently initialised for each collision model. Both collision models use the smoothed initial state \eqref{eq:initial} with $k=20$, prescribed incoming Maxwellians, $x\in[-1/2,1/2]$, $\mathrm{Kn}=1$ and final time $t_f=0.02$.

For a point $(t,x,\boldsymbol v)$, let $b(x,v_x)$ denote the backward travel time to an inflow boundary, as specified in~\eqref{eq:backward_travel_time}:
\begin{equation}
 b(x,v_x)=\begin{cases}
 (x+1/2)/v_x,&v_x>0,\\
 (x-1/2)/v_x,&v_x<0,\\
 +\infty,&v_x=0,
 \end{cases}
 \qquad \ell=\min(t,b).
\label{eq:backward_travel_time}
\end{equation}
The characteristic starts from the prescribed initial Maxwellian when $\ell=t$, or from the appropriate incoming reservoir when $\ell=b<t$. Writing that value as $f_{\rm in}$ gives the characteristic solution in~\eqref{eq:transient_characteristic_integral}:
\begin{equation}
 f_\theta(t,x,\boldsymbol v)=e^{-\ell/\mathrm{Kn}}f_{\rm in}
 +\frac{1}{\mathrm{Kn}}\int_{t-\ell}^{t}
 e^{-(t-s)/\mathrm{Kn}}
 \mathcal T_\theta\bigl(s,x-v_x(t-s),\boldsymbol v\bigr)\,\mathrm ds.
 \label{eq:transient_characteristic_integral}
\end{equation}
Here $\mathcal T_\theta$ is the collision target built from the neural fields of \cref{eq:transient_macro_parameterisation,eq:transient_macro_parameterisation-2}: the Maxwellian $M_\theta$ for BGK and \eqref{eq:transient_shakhov_source} for Shakhov. Only incoming distributions are prescribed at the boundaries; the macroscopic fields there also contain outgoing populations and are not clamped to the reservoir values.

The neural unknowns are $\rho_\theta,u_\theta,T_\theta$ for BGK, and additionally $q_\theta$ for Shakhov. With $a=t/t_f$ and initial fields $\rho_0,u_0=0,T_0$, the neural parametrisation is given in~\cref{eq:transient_macro_parameterisation,eq:transient_macro_parameterisation-2}:
\begin{align}
 \rho_\theta&=\rho_0\exp(a r_\rho/2),&
 u_\theta&=a r_u,&T_\theta&=T_0\exp(a r_T/2),\label{eq:transient_macro_parameterisation}
\\
 q_\theta&=a\rho_0T_0^{3/2}r_q\qquad\text{(Shakhov only)}.\label{eq:transient_macro_parameterisation-2}
\end{align}
The outputs $r$ are produced by a two-hidden-layer network with 32 hyperbolic-tangent units per layer, evaluated at $(t/t_f,2x)$. The initial macroscopic state is therefore imposed exactly and density and temperature remain positive. The BGK source is the analytic Maxwellian $M_\theta$ of these fields. For Shakhov, the source includes the heat-flux correction in~\eqref{eq:transient_shakhov_source}:
\begin{equation}
 \mathcal T_\theta=M_\theta\left[1+
 \frac{(1-\Pr)q_\theta c_x}{5\rho_\theta T_\theta^2}
 \left(\frac{|\boldsymbol c|^2}{T_\theta}-5\right)\right],
 \qquad \boldsymbol c=\boldsymbol v-u_\theta\boldsymbol e_x,
 \qquad \Pr=2/3.
\label{eq:transient_shakhov_source}
\end{equation}
Here both shock-tube models retain the constant relaxation time $\mathrm{Kn}$; this differs from the temperature--density-dependent relaxation convention of the stationary Shakhov cases. Positive integration weights preserve positivity for the BGK construction. The Shakhov target can be negative, so neither target nor distribution is clipped, and positivity must be assessed rather than assumed.

Let $\langle a\rangle_h=\sum_j w_j a(\boldsymbol v_j)$ denote the implemented velocity quadrature, where $\boldsymbol v_j$ and $w_j$ are its nodes and weights. Define $\boldsymbol c_f=\boldsymbol v-u_f\boldsymbol e_x$ and $c_f^2=|\boldsymbol c_f|^2$. Here $j_f$ denotes streamwise momentum density. The actual macroscopic and heat-flux moments are calculated from $f_\theta$, independently of the neural parameters used in $M_\theta$, through~\cref{eq:transient_integrated_moments,eq:transient_integrated_moments-2}:
\begin{align}
 \rho_f&=\langle f_\theta\rangle_h,&j_f&=\langle v_xf_\theta\rangle_h,&
 E_f&=\tfrac12\langle |\boldsymbol v|^2f_\theta\rangle_h,\label{eq:transient_integrated_moments}
\\
 u_f&=j_f/\rho_f,&T_f&=(2E_f-\rho_fu_f^2)/(3\rho_f),&
 q_f&=\tfrac12\langle c_f^2c_{f,x}f_\theta\rangle_h.\label{eq:transient_integrated_moments-2}
\end{align}
The self-consistency residuals are $\rho_f-\rho_\theta$, $j_f-\rho_\theta u_\theta$ and $E_f-\rho_\theta(u_\theta^2+3T_\theta)/2$, with $q_f-q_\theta$ added only for Shakhov. Their mean squared dimensionless values, using fixed upstream density and thermal-speed units, define the training objective $J_{\rm sc}$. The characteristic construction satisfies transport and relaxation for its supplied source; it represents self-consistent BGK or Shakhov dynamics only to the extent that these moment residuals and quadrature errors are small.

Each collision model uses 1000 Adam updates followed by 300 self-scaled Broyden updates. Adam has learning rate $10^{-3}$ and batches of 64 randomly sampled space--time points. The Broyden stage uses 256 fixed points, an initially identity inverse metric and strong-Wolfe line search. The strong-Wolfe constants are $c_1=10^{-4}$ and $c_2=0.9$; the initial and maximum trial steps are both one, with at most 40 objective/gradient evaluations per search and gradient-norm stopping tolerance $10^{-12}$. Its metric and collocation set are retained throughout the 300 updates. The two collision models use separate random initialisations, with one run per model. This formulation trains the network supplying the space--time source fields; the stationary characteristic correction instead optimises separate macroscopic coefficients.

The training quadrature uses 16 Gauss--Legendre nodes along each characteristic and a $48\times20\times20$ rectangular velocity grid on $[-7,7]^3$, in double precision. Transverse Gaussian sums and their polynomial moments are evaluated on this grid. Factorising these sums avoids repeated network evaluations at identical $(t,x,v_x)$ points. Validation uses 32 characteristic nodes, a $64\times24\times24$ velocity grid and held-out space--time points.

All higher moments are integrated from the same resulting distribution. With $\boldsymbol c_f=\boldsymbol v-u_f\boldsymbol e_x$ and $e_f=|\boldsymbol c_f|^2$, the stress and the higher moments of \cref{eq:sigmaxx,eq:mxxx,eq:Rxx,eq:Delta} take the quadrature form in~\cref{eq:transient_higher_moments,eq:transient_higher_moments-2,eq:transient_higher_moments-3,eq:transient_higher_moments-4}:
\begin{align}
 \sigma_{xx}&=\langle(c_{f,x}^2-e_f/3)f_\theta\rangle_h,\label{eq:transient_higher_moments}
\\
 m_{xxx}&=\langle(c_{f,x}^3-3c_{f,x}e_f/5)f_\theta\rangle_h,\label{eq:transient_higher_moments-2}
\\
 R_{xx}&=\langle e_f(c_{f,x}^2-e_f/3)f_\theta\rangle_h-7T_f\sigma_{xx},\label{eq:transient_higher_moments-3}
\\
 \Delta&=\langle e_f^2f_\theta\rangle_h-15\rho_fT_f^2,
 \qquad S_{xx}=R_{xx}+\Delta/3.\label{eq:transient_higher_moments-4}
\end{align}
These moments are integrated directly from $f_\theta$, using the same finite-grid convention for the prediction and its matched reference.

The comparison time is $t_f/\mathrm{Kn}=0.02$, for which the collision-integral weight is at most $1-e^{-0.02}\simeq0.0198$. The test therefore primarily probes transport over a short fraction of a relaxation time.

\subsection{Post-solve comparisons}
Table~\ref{tab:transient_physics} reports full-domain signed-profile discrepancies against the collision-matched 3200-cell DVM and the changes between its 1600- and 3200-cell calculations. Both DVM grids use the same $48\times20\times20$ velocity quadrature; first-order upwind transport is followed by exponential relaxation to the frozen post-transport source, with CFL number 0.4. The time step is refined with the spatial grid. The reference changes are therefore coupled space--time refinement differences, not certified continuum errors. No spatial shift or amplitude fit is applied. The BGK prediction is also compared with a 6400-cell solution in Appendix~\ref{ext:app:reference}.

\begin{table}
\centering
\caption{Equation-trained shock-tube calculations at $t=0.02$, $\mathrm{Kn}=1$. Relative $L^2$ discrepancies (\%) of the frozen neural prediction against the 3200-cell DVM, and DVM 1600-to-3200-cell refinement differences. The latter use linear interpolation of the fine reference onto the coarse nodes. Full-distribution higher-moment definitions are used for both predictions and references.}
\label{tab:transient_physics}
\begin{tabular}{lrrrr}
\toprule
 & \multicolumn{2}{c}{BGK} & \multicolumn{2}{c}{Shakhov}\\
Moment & Neural--DVM & DVM refinement & Neural--DVM & DVM refinement\\
\midrule
$\rho$ & 0.0096 & 0.0095 & 0.0096 & 0.0095\\
$u_x$ & 0.2042 & 0.2042 & 0.2029 & 0.2042\\
$T$ & 0.0092 & 0.0091 & 0.0092 & 0.0091\\
$q_x$ & 0.4100 & 0.4065 & 0.4101 & 0.4060\\
$\sigma_{xx}$ & 0.3604 & 0.3578 & 0.3606 & 0.3578\\
$m_{xxx}$ & 0.8263 & 0.8180 & 0.8284 & 0.8178\\
$R_{xx}$ & 1.7145 & 1.6977 & 1.7125 & 1.6957\\
$\Delta$ & 1.3976 & 1.3843 & 1.3979 & 1.3828\\
$S_{xx}$ & 1.5759 & 1.5605 & 1.5748 & 1.5588\\
\bottomrule
\end{tabular}
\end{table}

The heat-flux and stress discrepancies are $0.41\%$ and $0.36\%$ for both models. The $R_{xx}$ and $\Delta$ discrepancies are $1.714\%$ and $1.398\%$ for BGK and $1.712\%$ and $1.398\%$ for Shakhov. Their proximity to the DVM refinement differences in table~\ref{tab:transient_physics} limits the precision with which neural and reference errors can be separated.

Optimisation sensitivity is assessed by comparing the states after 100 and 300 Broyden updates against the same 3200-cell reference. Heat-flux discrepancies change from $0.410617\%$ to $0.410026\%$ for BGK and from $0.412103\%$ to $0.410056\%$ for Shakhov. Stress discrepancies increase slightly, from $0.356950\%$ to $0.360440\%$ and from $0.356831\%$ to $0.360561\%$, respectively. Fourth-order discrepancies decrease slightly. A lower self-consistency objective therefore does not uniformly reduce every field error.

Figures~\ref{fig:transient_bgk_physics} and~\ref{fig:transient_shakhov_physics} compare the signed profiles. The BGK heat-flux peak is $0.0434050$, compared with $0.0432021$ for DVM, at $x=-0.03359375$. The corresponding Shakhov peaks are $0.0435445$ and $0.0433390$.

Doubling the characteristic quadrature from 16 to 32 nodes changes the retained fields by at most $5.3\times10^{-14}$ in relative $L^2$ norm. Increasing the velocity grid to $64\times24\times24$ changes them by at most $8.41\times10^{-7}$. These are fixed-endpoint sensitivity checks. Held-out self-consistency root-mean-square residuals for density, momentum and energy are $(6.47,5.35,6.69)\times10^{-5}$ for BGK and $(10.64,7.48,6.99)\times10^{-5}$ for Shakhov; the latter has a heat-flux residual of $1.15\times10^{-4}$.

The sampled BGK distribution is nonnegative. Over 21 spatial points and the training velocity grid, the Shakhov distribution reaches $-1.87\times10^{-19}$; the sum of the velocity-integrated negative masses at these points is $2.31\times10^{-18}$. The Shakhov DVM minimum is $-1.55\times10^{-19}$. These small signed tails are retained without clipping.

\begin{figure}
\centering
\includegraphics[width=\textwidth,height=.78\textheight,keepaspectratio]{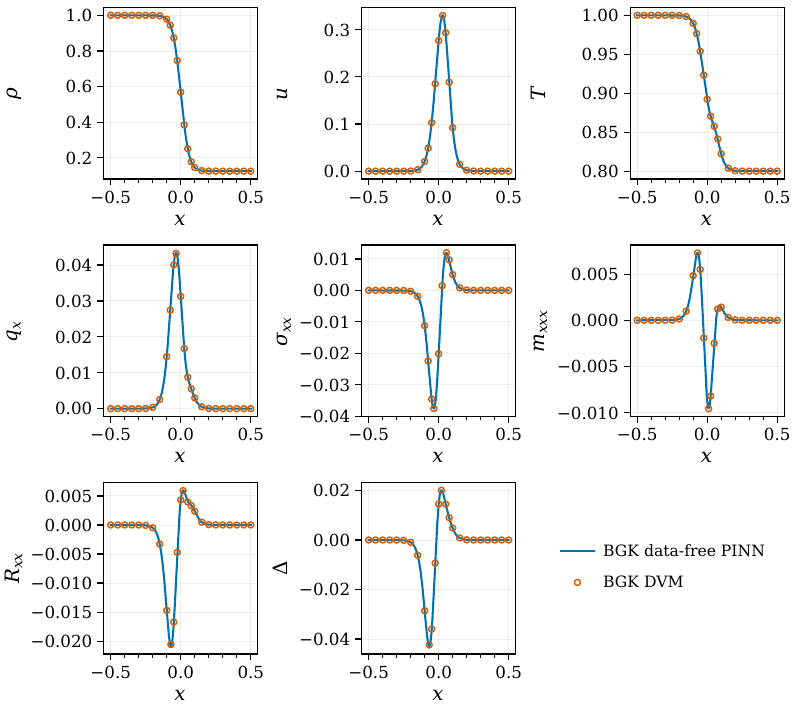}
\caption{Equation-trained BGK characteristic-PINN prediction and collision-matched DVM at $t=0.02$. All plotted moments are computed from the resulting distribution. The DVM refinement comparison and full-domain errors are given in table~\ref{tab:transient_physics}.}
\label{fig:transient_bgk_physics}
\end{figure}
\begin{figure}
\centering
\includegraphics[width=\textwidth,height=.78\textheight,keepaspectratio]{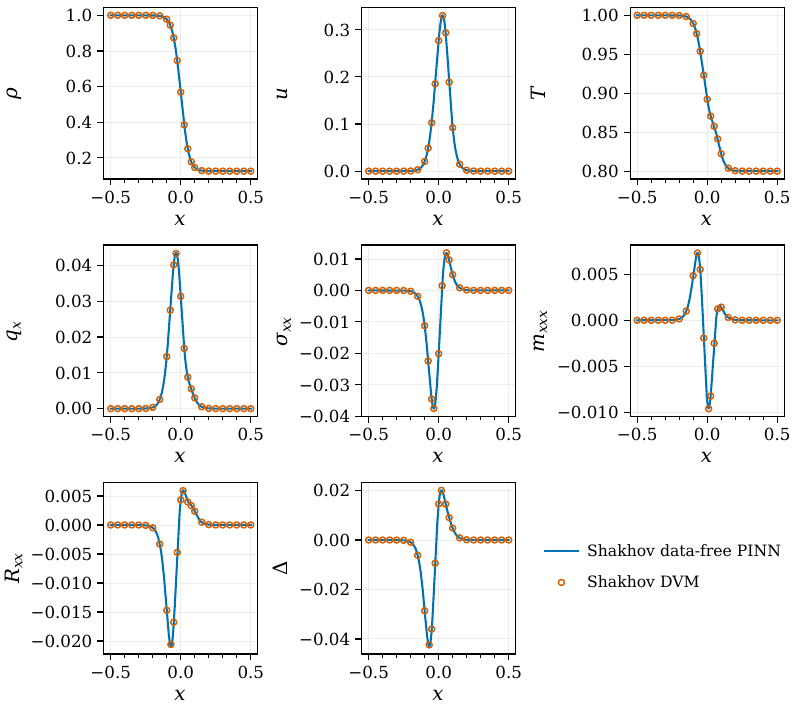}
\caption{Equation-trained Shakhov characteristic-PINN prediction ($\Pr=2/3$, constant $\mathrm{Kn}=1$) and collision-matched DVM at $t=0.02$. Signed higher moments use the same full-distribution convention as the BGK comparison. The Shakhov target and distribution are not clipped; very small negative tails remain.}
\label{fig:transient_shakhov_physics}
\end{figure}

\section{Equation-trained stationary-shock calculations}\label{ext:sec:newextension}

We consider neural kinetic continuation and positive characteristic correction for Shakhov shocks at Mach 2--5 and BGK shocks at Mach 6, 7, 8 and 12. The collision models and numerical unknowns are summarised in \cref{tab:calculations}.

\subsection{Information and evaluation protocol}
Warm-start calculations specify the construction of the initial state separately from the continuation objective.

The stationary continuation chain starts from a randomly initialised Mach-2 Shakhov network. Subsequent states transfer its own trained parameters between Mach numbers and collision models, with analytic end-state changes and kinetic-residual optimisation. The Mach-7, Mach-8 and Mach-12 characteristic initialisers use macroscopic fields integrated from those neural states and analytic Rankine--Hugoniot conditions. These warm-start trajectories are distinct from the three independent Mach-2 initialisations.

Independent repetitions use different random initialisations and stochastic samples. States along one optimisation trajectory are treated as correlated observations.

The Mach-2--5 Shakhov and higher-Mach BGK results each represent one continuation trajectory. The three independently initialised Mach-2 Shakhov runs in \cref{ext:sec:independentseeds} assess variability at a common training budget.

\subsection{Shakhov shocks at Mach 2--5}
The Mach-2--5 neural calculations solve the steady form of \eqref{eq:relaxation}, $v_x\partial_x f=(\mathcal{T}_S[f]-f)/\tau$, with the Shakhov target \eqref{eq:shakhov}, $\Pr=2/3$, $\omega=1$ and $\mu_1=1$, so that $\tau=\mu_1T^{\omega-1}/\rho=1/\rho$. The BGK calculations at Mach 6, 7, 8 and 12 instead use $\Pr=1$ and the constant relaxation time $\tau=1$. The Shakhov target is not clipped to positive values; its negative mass and collision-invariant defects are monitored separately.
The upstream state is $(\rho_1,u_1,T_1)=(1,M\sqrt{5/3},1)$ and the downstream state is the Rankine--Hugoniot state \eqref{eq:rh}.
Write $[-L,L]$ for the spatial interval, $C$ for the velocity cutoff, and $r=(v_y^2+v_z^2)^{1/2}$ for transverse molecular speed. The axial velocity quadrature is split into Gauss--Legendre rules on $[-C,0]$ and $[0,C]$, with 192 total axial nodes; we refer to this as the split-half quadrature. A 64-node Gauss--Legendre rule on $r\in[0,C]$ carries the cylindrical measure $2\pi r\,dr\,dv_x$. The recorded $(L,C)$ pairs for Mach 2, 3, 4 and 5 are $(30,16)$, $(45,20)$, $(60,28)$ and $(75,32)$, respectively.
These calculations use a signed distribution representation. A three-hidden-layer, 64-neuron Tanh network takes $(x/L,\tanh(x/w_0),v_x/\sqrt{T_2},r^2/T_2,s(x))$ and supplies a bounded correction. Its mixture $B$, mixture fraction $s$ and incoming gate $g$ have the form of \cref{eq:mixture_reference,eq:incoming_gate}, and the loss terms named below are defined in \S\ref{ext:sec:methodA}, \crefrange{eq:frozen_tail_multiplier}{eq:phase_condition}. With $B=(1-s)M_1+sM_2$ and a characteristic incoming gate $g$, the implemented branch is $f=B(1+g\psi)$, where $\psi=\psi_{\max}\tanh(N_\theta/\psi_{\max})$. Here $M_1$ and $M_2$ are the analytic upstream and downstream Maxwellians, $N_\theta$ is the scalar network output, and the incoming distance $d_{\rm in}$ defining the gate is specified in~\eqref{eq:incoming_gate}. The trainable sigmoid width defines $s$ and the gate enforces the incoming kinetic boundary values.
Here $w_0$ is the fixed input-feature width, distinct from the learned mixture width. The gate is $g=1-\exp[-d_{\rm in}/(|v_x|\tau_1)]$, with $\tau_1=\mu_1T_1^{\omega-1}/\rho_1$; it reduces to \eqref{eq:incoming_gate} for $\tau_1=1$. The Mach-2 split-half configuration uses $w_0=3$ and $\psi_{\max}=6$. The gate uses the upstream relaxation scale; the residual retains the local $\tau=1/\rho$.
The scalar weights in these recorded configurations are 1 for kinetic, tail and raw residual terms; 30 for flux and phase; 0.1 for stress, heat, $R$ and $\Delta$ residual projections; and zero for invariant and outgoing-boundary penalties. Frozen-objective continuation uses these declared weights rather than an annealed schedule. The signed ansatz does not guarantee positivity.

\section{Method A: stationary kinetic PINNs}\label{ext:sec:methodA}
\subsection{Representation and the fixed objective}
The Mach-6 implementation represents the distribution directly. Let \(M_1,M_2\) denote the analytic upstream and downstream Maxwellians. The analytic reference mixture is specified in~\eqref{eq:mixture_reference}:
\begin{equation}
 B_\theta(x,\bv)=[1-s_\theta(x)]M_1(\bv)+s_\theta(x)M_2(\bv),\qquad
 s_\theta(x)=\frac{\ell(x/w_\theta)-\ell(-L/w_\theta)}{\ell(L/w_\theta)-\ell(-L/w_\theta)},
\label{eq:mixture_reference}
\end{equation}
where \(\ell(z)=(1+e^{-z})^{-1}\) and \(w_\theta>0\). A bounded neural correction \(\psi_\theta=\psi_{\max}\tanh(N_\theta/\psi_{\max})\) acts through the incoming-boundary gate in~\eqref{eq:incoming_gate}:
\begin{equation}
 g(x,v_x)=1-\exp[-d_{\rm in}(x,v_x)/|v_x|],\qquad
 d_{\rm in}=\begin{cases}x+L,&v_x\geq0,\\L-x,&v_x<0.\end{cases}
\label{eq:incoming_gate}
\end{equation}
The representation admits the positive branch \(f_\theta=B_\theta\exp(g\psi_\theta)\) and the signed branch \(f_\theta=B_\theta(1+g\psi_\theta)\). The Mach-6 solution uses the signed branch; positivity is therefore not guaranteed by its ansatz. At zero axial speed the interior gate is one and its sign-selected incoming endpoint value is zero. The inputs to \(N_\theta\) are \(x/L,\tanh(x/w_0),v_x/\sqrt{T_2},v_\perp^2/T_2,s_\theta\). Here $v_\perp=(v_y^2+v_z^2)^{1/2}=r$ is the transverse molecular speed. Density, velocity and temperature used in the collision target are integrated from \(f_\theta\).

The Mach-6 model uses three hidden layers of width 64 with hyperbolic-tangent activation and the signed representation. Its training velocity grid has 256 axial and 80 radial nodes with cutoff 36, on a spatial domain of half-length 90. These settings are case-specific; the Mach-7 construction uses half-length 105 and a 512-by-144 grid with cutoff 42.

For the fixed Mach-6 quasi-Newton stage, write \(r_{ij}=v_{x,j}\partial_x f_\theta(x_i,\bv_j)-(M[f_\theta]-f_\theta)_{ij}\), with \(\tau=1\). The Shakhov stages use residuals of the same form with the target \eqref{eq:shakhov} and the local relaxation time \(\tau=1/\rho\). Let \(\omega_j\) be velocity quadrature weights and \(p_i>0\) the prescribed spatial weights, normalised to \(\sum_i p_i=1\). The training grid has 187 prescribed spatial points. An analytic RH mixture \(E_{ij}\), evaluated with the fixed initial width, defines \(D_{ij}=E_{ij}+10^{-7}\max_j E_{ij}\). The tail multiplier in~\eqref{eq:frozen_tail_multiplier} is fixed by the initial physical state:
\begin{equation}
 t_{ij}=1+\left[\frac{(v_{x,j}-u_i^{(0)})^2+v_{\perp,j}^2}{5T_2}\right]^4.
\label{eq:frozen_tail_multiplier}
\end{equation}
Here $u_i^{(0)}$ is the velocity obtained by integrating the initial neural distribution at $x_i$. The three kinetic terms in~\cref{eq:kinetic_loss_terms,eq:kinetic_loss_terms-2,eq:kinetic_loss_terms-3} use these fixed scales:
\begin{align}
 \mathcal L_{\rm kin}&=\sum_i p_i\frac{\sum_j E_{ij}\omega_j(r_{ij}/D_{ij})^2}{\sum_j E_{ij}\omega_j},\label{eq:kinetic_loss_terms}
\\
 \mathcal L_{\rm tail}&=\sum_i p_i\frac{\sum_j E_{ij}\omega_j t_{ij}(r_{ij}/D_{ij})^2}{\sum_j E_{ij}\omega_j t_{ij}},\label{eq:kinetic_loss_terms-2}
\\
 \mathcal L_{\rm raw}&=\sum_i p_i\sum_j\omega_j r_{ij}^2/\rho_2^2.\label{eq:kinetic_loss_terms-3}
\end{align}
The envelope and tail weights are therefore fixed during the line search and restored on restart. This avoids changing the scalar objective between trial steps. It does not freeze the velocity-centred physical moment kernels.

With \(c_x=v_x-u(x_i)\), \(e=c_x^2+v_\perp^2\), the projected residuals are \(P_{i\ell}=\rho_2^{-1}\sum_j\omega_j r_{ij}K_{i\ell j}\), with the velocity kernels in~\eqref{eq:projection_kernels}:
\begin{equation}
 (K_1,\ldots,K_7)=\left(1,\frac{v_x}{\sqrt{T_2}},\frac{v_x^2+v_\perp^2}{T_2},
 \frac{c_x^2-e/3}{T_2},\frac{c_xe}{2T_2^{3/2}},
 \frac{(e-7T)(c_x^2-e/3)}{T_2^2},\frac{e^2}{T_2^2}\right).
\label{eq:projection_kernels}
\end{equation}
The invariant contribution averages \(P_1^2,P_2^2,P_3^2\); four separate contributions use \(P_4^2,\ldots,P_7^2\). All are spatially averaged using \(p_i\). The implemented gradient differentiates through the current \(u,T\) in these kernels. These terms penalise velocity-weighted projections of the kinetic residual.

The remaining terms penalise the relative defects of the three fluxes in the moment budgets, averaged over their three components, and the phase condition in~\eqref{eq:phase_condition}:
\begin{equation}
 \mathcal L_{\rm phase}=\left[\frac{\rho(0)-(1+\rho_2)/2}{(1+\rho_2)/2}\right]^2.
\label{eq:phase_condition}
\end{equation}
Upstream density and temperature are unity in this implementation. The flux normalisers are consequently \(u_1\), \(u_1^2+1\), and \(u_1^3/2+5u_1/2\). The complete fixed objective is the configuration-weighted sum of these ten contributions: kinetic, tail, raw, flux, invariant, stress, heat, fourth-order deviatoric projection, fourth-order scalar projection and phase. The outgoing finite-domain boundary penalty is disabled. The phase penalty is evaluated once per objective.

\subsection{Dense self-scaled quasi-Newton continuation}
Let $\theta_k$ be the trainable parameter vector after $k$ accepted updates, $J_{\rm NN}$ the fixed weighted objective just defined, and $H_k$ its inverse-Hessian approximation. For the Mach-6 stage, the search direction is \(d_k=-H_k\nabla J_{\rm NN}(\theta_k)\). The bounded line search requires both strong-Wolfe inequalities in~\cref{eq:wolfe_conditions,eq:wolfe_conditions-2}:
\begin{align}
 J_{\rm NN}(\theta_k+\alpha d_k)&\leq J_{\rm NN}(\theta_k)+10^{-4}\alpha\nabla J_{\rm NN}(\theta_k)^T d_k,\label{eq:wolfe_conditions}
\\
 |\nabla J_{\rm NN}(\theta_k+\alpha d_k)^T d_k|&\leq0.9|\nabla J_{\rm NN}(\theta_k)^Td_k|.\label{eq:wolfe_conditions-2}
\end{align}
The continuation uses initial step one, maximum step ten and a budget of 40 objective/gradient evaluations. On failure the base model is restored; the inverse metric is not reset. For \(s_k=\theta_{k+1}-\theta_k\), \(y_k=\nabla J_{k+1}-\nabla J_k\), \(z_k=H_ky_k\), \(q_k=y_k^Tz_k\), and \(\varrho_k=(y_k^Ts_k)^{-1}\), the dense self-scaled Broyden implementation forms
\begin{equation}
 H_{k+1}=\frac{H_k-z_kz_k^T/q_k+\phi_kq_k v_kv_k^T}{\tau_k}+\varrho_k s_ks_k^T,
 \qquad v_k=\varrho_ks_k-z_k/q_k.
 \label{ext:eq:ssbupdate}
\end{equation}
Here \(\phi_k,\tau_k\) are the implemented self-scaling coefficients; the scaling factor \(\tau_k\) is unrelated to the relaxation time \(\tau\). The computed matrix is symmetrised and its relative secant defect is checked against \(10^{-6}\); nonfinite or invalid curvature updates are rejected. The coefficient prescription is given in Appendix~\ref{ext:app:ssb}; case-specific settings are given in the configuration table.

The Mach-6 BGK calculation is initialised from a Mach-6 Shakhov solution after 500 updates. This collision-model transition uses a fresh quasi-Newton metric and fixed objective. Subsequent continuation retains both. The result uses seed 101 and is a warm-start BGK calculation, distinct from the cold-start Shakhov ensemble in \cref{ext:sec:independentseeds}.

\subsection{Algorithm A: frozen-objective neural continuation}

\begin{enumerate}
\item Specify the collision model, RH data, representation and source lineage. Restore the complete model and, for an exact continuation, its metric, frozen objective and random states; construct a fresh metric only where the recorded stage explicitly does so.
\item Bind the spatial/velocity quadrature and fixed residual/tail scales. Evaluate all configured objective terms and the single phase term without changing the samples during line search.
\item Differentiate the full objective, including the state-dependent central kernels. Form $d=-H\nabla J$ and reject a nonfinite or non-descent direction.
\item Test steps by the recorded line search. For the Mach-6 SSB stage, require both strong-Wolfe inequalities; restore the base parameters on rejection without changing the objective or metric.
\item For an accepted step, form the curvature pair and checked SSB update. Save the model, metric, objective state and actual accepted-step counter. Record physical diagnostics separately from objective decrease.
\item Stop on the prescribed budget, numerical failure or recorded termination condition. Freeze the selected endpoint before the declared post-solve reference comparison. 
\end{enumerate}

\section{Method B: characteristic kinetic correction}\label{ext:sec:methodB}
\subsection{Positive transport and macro self-consistency}
The characteristic correction is formulated for the BGK target, which is fixed by $(\rho,u,T)$; a Shakhov version would also have to carry the heat flux in the source. For BGK with prescribed macroscopic fields, let distance \(s\) increase from the inflow boundary along the sign of \(v_x\), and set \(\nu=1/\tau\) and \(a=\nu/|v_x|\), where \(\nu\) is the collision frequency. For nonzero axial velocity, the characteristic integral is~\eqref{ext:eq:characteristic}:
\begin{equation}
 f(s,\bv)=f_{\rm in}(\bv)e^{-\int_0^s a(r,\bv)\,\dd r}
 +\int_0^s a(t,\bv)M(t,\bv)e^{-\int_t^s a(r,\bv)\,\dd r}\dd t.
\label{ext:eq:characteristic}
\end{equation}
The constant-\(\tau=1\) implementation integrates a piecewise-linear Maxwellian source exactly on each spatial interval. If \(h=\Delta x/|v_x|\), \(A=e^{-h}\) and \(Z=(1-e^{-h})/h\), its recurrence in the transport direction is
\begin{equation}
 f_i=A f_{i-1}+(Z-A)M_{i-1}+(1-Z)M_i.
\label{ext:eq:transportrecurrence}
\end{equation}
All three coefficients are nonnegative and sum to one for \(h>0\). Positive analytic RH inflow and positive Maxwellian source therefore give a positive transported distribution under this recurrence. This property does not certify macro self-consistency, velocity truncation accuracy or continuum convergence. The implementation evaluates \(1-e^{-h}\) using \texttt{expm1} to avoid subtraction loss.

The adjustable fields are \(\boldsymbol U_\beta=(\rho_\beta,u_\beta,T_\beta)\), with \(\rho_\beta=J_m/u_\beta\) and fixed physical mass flux \(J_m\). Compact cubic spline increments adjust velocity and temperature; the velocity coefficient block is restricted to a centre-phase nullspace. End states are fixed, and positivity of velocity and temperature is checked at trial evaluation.
For the Mach-7 correction, cubic splines are evaluated at $\tanh(x/J_m)$ and multiplied by the fixed endpoint envelope. If $B(x)$ is their row and $N$ spans the nullspace of the spline row at zero, then $\delta u=J_mB(x)N\beta_u$ and $\delta T=T_2B(x)\beta_T$. The free coefficient vector is $\beta=(\beta_u,\beta_T)$, and $\delta u(0)=0$ preserves $u(0)=2J_m/(1+\rho_2)$, equivalently $\rho(0)=(1+\rho_2)/2$. The transported moments \(\mathcal M[f_\beta]=(\rho_f,u_f,T_f)\) need not equal the prescribed source fields before correction. The objective is
\begin{equation}
 J_{\rm char}(\beta)=\sum_i p_i\|G_i(\beta)\|_2^2,\qquad
 G_i=\left(\frac{\rho_f-\rho_\beta}{\rho_2},\frac{u_f-u_\beta}{J_m},\frac{T_f-T_\beta}{T_2}\right)_{x_i}.
\label{ext:eq:charobjective}
\end{equation}
The normalisers are fixed RH quantities in the implementation; the use of \(J_m\) in the nondimensional velocity defect follows that implementation.
The weights $p_i$ are normalised spatial trapezoidal weights: for nodes $x_0<\cdots<x_N$, take endpoint weights $(x_1-x_0)/2$ and $(x_N-x_{N-1})/2$, interior weights $(x_{i+1}-x_{i-1})/2$, and divide by $x_N-x_0$. When the input fields originate from a neural checkpoint, this is a neural-network-initialised kinetic correction with zero neural-weight updates, not a further neural training stage.

\subsection{Bounded macroparameter update}
The transport tangent is propagated together with the characteristic solution. Let \(r\) stack \(\sqrt{p_i}G_i\) and \(A_\beta\) be its coefficient Jacobian. The correction solves \(\min_\delta\|A_\beta\delta+r\|_2\) by a rank-revealing least-squares solve with relative singular-value cutoff \(10^{-12}\). Full column rank is required; the implementation does not add a ridge when that test fails. The direction is limited by the componentwise bound in~\eqref{ext:eq:stepbound}:
\begin{equation}
 \max\left(\|\delta u\|_\infty/J_m,\|\delta T\|_\infty/T_2\right)\leq0.005.
 \label{ext:eq:stepbound}
\end{equation}
Up to eight candidates \(\alpha=2^{-k}\), \(k=0,\ldots,7\), are tested for positive trial fields and the recorded Armijo decrease. The source uses \(b=(A_\beta^Tr)^T\delta\) in the condition \(J_{\rm char}(\beta+\alpha\delta)\leq J_{\rm char}(\beta)+10^{-4}\alpha b\). Since \(J_{\rm char}=r^Tr\), \(b\) is half its directional derivative; this convention is stated explicitly as used in the line search. A scaled-direction threshold or exhausted line search is recorded as a numerical stop, not convergence.

The Mach-8 correction uses 24 cubic basis functions and 47 coefficients in its coarse stages, followed by 119 basis functions and 237 coefficients for refinement. The Mach-12 chord iteration reuses a Jacobian for four iterations, and its refined approximation uses 235 basis functions and 469 coefficients. Objectives for these different approximation spaces are plotted separately.

The neural weights remain fixed during characteristic correction. Flux defects, projected residuals and quadrature sensitivity are evaluated separately from the self-consistency objective.

\(J_{\rm char}\) and \(J_{\rm NN}\) are different objectives on different unknowns. Their histories are shown separately; matched fields and grids define the common physical comparison.

\subsection{Algorithm B: positive characteristic correction}

\begin{enumerate}
\item Freeze the source macrofields and collision law. Construct the initial state from the neural macroscopic fields and RH data, fix the inflow Maxwellians and phase gauge, and define the spline basis and mass relation $\rho=J_m/u$.
\item For current macroparameters, sweep the positive characteristic recurrence in each velocity direction. Integrate the resulting density, velocity and temperature and compute the scaled self-consistency residual.
\item Propagate the transport tangent or reuse it only according to the stated chord schedule. Solve the full-rank least-squares system and bound the velocity/temperature increment.
\item Evaluate the prescribed halved trial steps. Require positive source fields and the recorded Armijo decrease; an unsuccessful search terminates the iteration.
\item Accept the macroparameter update when these conditions hold; neural weights remain fixed. Evaluate physical diagnostics on the stated diagnostic grid.
\item Evaluate the final solution using signed fields, residuals and the collision-matched reference on common positions.
\end{enumerate}

\subsection{Numerical workflow}
\label{ext:sec:workflow}
For the equation-based studies, figure~\ref{ext:fig:workflow} separates the numerical reference calculation, the transient PINNs and the two stationary algorithms. The card-and-loop presentation follows the organisational approach of \citet{RoohiMcDevitt2026DeepCavity}; the numerical steps here are those of the kinetic calculations.

\subsubsection{Reference time stepping and characteristic neural training.}
For the shock tube, initialise the velocity-node populations from the stated left/right initial states and impose Maxwellian values only on incoming characteristics at the boundaries. Each time step first advances the upwind finite-volume transport in \eqref{eq:dvm_transport}, integrates its moments on the fixed velocity grid, and then performs the exponential collision update in \eqref{eq:dvm_relax}. Repeat until the prescribed final time, shortening the last step when necessary. This time loop is distinct from an optimiser iteration. Stationary comparisons use separately computed, collision-matched reference profiles. Their coordinates, reference centring and available refinement checks are specified with the corresponding results.

The transient PINN uses the characteristic integral in \cref{sec:transient_characteristic_pinn} to transport analytic initial and incoming-boundary distributions. Neural density, velocity and temperature define its BGK source; the Shakhov source additionally uses neural heat flux. Integrating the transported distribution gives the moments against which those source fields are checked. Adam uses stochastic space--time samples, followed by self-scaled Broyden optimisation on a fixed collocation set. Endpoint evaluation integrates moments from the transported distribution, rather than substituting the neural source parameters.

The stochastic stages use Adam \citep{KingmaBa2015Adam}. Let $g_k$ denote the parameter gradient after any configured norm clipping. Here $k$ counts active updates separately for each parameter; a parameter that was frozen need not share another parameter's counter. Let $\beta_1,\beta_2\in[0,1)$ be the exponential averaging coefficients, and let $\epsilon>0$ be the numerical denominator floor. Their numerical values are those of the recorded optimiser configuration. With componentwise products and division, the gradient averages and their bias corrections are specified in~\cref{ext:eq:adam_moments,ext:eq:adam_moments-2}:
\begin{align}
 m_k&=\beta_1m_{k-1}+(1-\beta_1)g_k,\qquad
 v_k=\beta_2v_{k-1}+(1-\beta_2)(g_k\odot g_k),\label{ext:eq:adam_moments}\\
 \widehat m_k&=m_k/(1-\beta_1^k),\qquad
 \widehat v_k=v_k/(1-\beta_2^k).\label{ext:eq:adam_moments-2}
\end{align}
With learning rate $\eta_k$, the Adam parameter update is given by~\eqref{ext:eq:adamw_update}:
\begin{equation}
 \theta_{k+1}=\theta_k-\eta_k\frac{\widehat m_k}{\sqrt{\widehat v_k}+\epsilon}.
 \label{ext:eq:adamw_update}
\end{equation}
These moment estimates describe gradients, not physical velocity moments. Their iteration counter and averages belong to the optimiser state: an exact restart restores them, whereas a documented fresh optimisation phase initialises them anew. Stochastic Adam updates have no strong-Wolfe acceptance test; the fixed-objective quasi-Newton phase below is a separate algorithm.

\subsubsection{Neural stationary continuation.}
Algorithm A updates the neural parameters using the self-scaled Broyden construction of \citet{UrbanStefanouPons2025}, with coefficients given in \cref{ext:app:ssb}. At the transition from stochastic optimisation, the objective scales, samples and initial inverse metric are fixed. Each line-search evaluation differentiates the current distribution and central-moment kernels on this same design. A rejected trial restores the base parameters; an accepted trial supplies the curvature pair for the metric update. Invalid curvature or an exhausted search terminates the stage.

\subsubsection{Characteristic correction and endpoint assessment.}
Algorithm B holds neural weights fixed and adjusts spline coefficients for the macroscopic source. It sweeps \eqref{ext:eq:transportrecurrence}, integrates the transported distribution and evaluates \eqref{ext:eq:charobjective}. The tangent system supplies a coefficient increment, which is accepted when the source fields remain positive and the Armijo condition holds.

Final assessment uses signed fields, higher moments, conservation defects and residuals. Comparisons use collision-matched references and the stated spatial intervals. Optimisation curves are shown separately for each stage.

\begin{figure}
\centering
% Execution and information flow of the kinetic calculations
% (figure ext:fig:workflow).  Vector schematic typeset in the document fonts;
% equation numbers are live cross-references.
% Requires in the preamble: \usepackage{tikz}
%                           \usetikzlibrary{arrows.meta,fit}
% Canvas: 135 mm wide (JFM text width 135.5 mm); three lanes of 39 mm.
\begingroup
\definecolor{wfNavy}{HTML}{183B56}%
\definecolor{wfBlue}{HTML}{EDF4FA}%
\definecolor{wfTeal}{HTML}{EAF5F1}%
\definecolor{wfGold}{HTML}{FFF5DE}%
\definecolor{wfAmber}{HTML}{A8681A}%
\hypersetup{linkcolor=wfNavy}%
\newcommand\wfstyle{\fontsize{8}{9.6}\selectfont\raggedright}%
\newcommand\wfhead[1]{\textcolor{wfNavy}{\textbf{#1}}}%
\newlength\wftmp
% \wfmeasure{<length>}{<width>}{<content>}: raise <length> to the content height
\newcommand\wfmeasure[3]{\settototalheight{\wftmp}{\parbox[t]{#2}{\wfstyle #3}}%
  \ifdim\wftmp>#1\relax #1=\wftmp\fi}%
\newcommand\wfbox[3]{\parbox[t][#1][t]{#2}{\wfstyle #3}}%
\pgfdeclarelayer{wfpanel}\pgfdeclarelayer{wfcard}\pgfsetlayers{wfpanel,wfcard,main}%
% ------------------------------------------------------------------ card text
\def\cOne{\wfhead{1\quad Specify the kinetic problem}\\[0.6mm]
  Collision law (BGK or Shakhov), initial or Rankine--Hugoniot states,
  velocity quadrature, spatial domain and update budget.}%
\def\cTwo{\wfhead{2\quad DVM reference}\\[0.6mm]
  Upwind transport \eqref{eq:dvm_transport}\\
  Quadrature moments\\
  Exponential collision \eqref{eq:dvm_relax}}%
\def\sA{\wfhead{3a\quad Algorithm A}\\\wfhead{Neural continuation}\\[0.6mm]
  Shakhov or BGK target; own source or cold start; grids fixed}%
\def\sB{\wfhead{3b\quad Algorithm B}\\\wfhead{Characteristic correction}\\[0.6mm]
  BGK target; own macrofields; inflow, phase and $\rho=J_m/u$ fixed}%
\def\sC{\wfhead{3c\quad Transient PINN}\\\wfhead{BGK or Shakhov}\\[0.6mm]
  Fresh network; analytic initial and incoming-boundary distributions}%
\def\uA{\wfhead{4\quad Evaluate and update}\\[0.6mm]
  Frozen $J_{\rm NN}$ and its gradient\\
  Direction $d=-H\nabla J_{\rm NN}$\\
  Strong-Wolfe line search}%
\def\uB{\wfhead{4\quad Evaluate and update}\\[0.6mm]
  Positive sweep \eqref{ext:eq:transportrecurrence}\\
  Residual $J_{\rm char}$ \eqref{ext:eq:charobjective}\\
  Transport tangent\\
  Bounded step \eqref{ext:eq:stepbound}}%
\def\uC{\wfhead{4\quad Evaluate and update}\\[0.6mm]
  Characteristic transport\\
  Integrate own moments\\
  Source--moment consistency\\
  Adam, then fixed-grid SSB}%
\def\dA{\wfhead{5\quad Accept or reject}\\[0.6mm]
  Reject: restore the base\\
  Accept: SSB update \eqref{ext:eq:ssbupdate}\\
  Save the continuation state}%
\def\dB{\wfhead{5\quad Accept or reject}\\[0.6mm]
  Reject: halve the step\\
  Accept: save coefficients\\
  Neural weights unchanged}%
\def\dC{\wfhead{5\quad Repeat or stop}\\[0.6mm]
  Fixed update budgets\\
  SSB: strong-Wolfe test\\
  Save model, metric and random state}%
\def\eBody{\wfhead{6\quad Assess the frozen endpoint}\\[0.6mm]
  Evaluate signed fields and higher moments on the declared grids; check
  conservation, kinetic residuals and refinement; compare collision-matched
  DVM fields and report the discrepancies.}%
\def\eNote{Field accuracy, kinetic residuals and conservation balances}%
% ------------------------------------------------------------------ row heights
\def\TW{34.6mm}\def\TWide{78.6mm}%
\newlength\wfHi\newlength\wfHii\newlength\wfHiii\newlength\wfHiv
\wfmeasure\wfHi{\TWide}{\cOne}\wfmeasure\wfHi{\TW}{\cTwo}%
\wfmeasure\wfHii{\TW}{\sA}\wfmeasure\wfHii{\TW}{\sB}\wfmeasure\wfHii{\TW}{\sC}%
\wfmeasure\wfHiii{\TW}{\uA}\wfmeasure\wfHiii{\TW}{\uB}\wfmeasure\wfHiii{\TW}{\uC}%
\wfmeasure\wfHiv{\TW}{\dA}\wfmeasure\wfHiv{\TW}{\dB}\wfmeasure\wfHiv{\TW}{\dC}%
\begin{tikzpicture}[x=1mm,y=1mm,
  card/.style={draw=wfNavy!50,line width=.6pt,fill=white,rounded corners=1.6mm,
    inner xsep=2.2mm,inner ysep=1.9mm,anchor=north west,align=left,text=black},
  inset/.style={fill=wfGold,draw=wfAmber!40,line width=.4pt,rounded corners=1.2mm,
    inner xsep=1.8mm,inner ysep=1.5mm,align=center,text=black,
    font=\fontsize{8}{9.6}\selectfont},
  flow/.style={-{Latex[length=2.2mm,width=1.6mm]},draw=wfNavy,line width=.75pt},
  target/.style={flow,draw=wfAmber,dash pattern=on 2.4pt off 1.5pt},
  iterate/.style={flow,rounded corners=1.6mm},
  note/.style={font=\fontsize{7}{8}\selectfont\itshape,text=wfAmber!85!black,inner sep=0.6mm}]
% lane origins and canvas edges
\coordinate (xa) at (2.5,0);\coordinate (xb) at (46.5,0);\coordinate (xc) at (90.5,0);
\coordinate (xl) at (0,0);\coordinate (xr) at (135,0);\coordinate (xe) at (129.5,0);
% ------------------------------------------------------------------ row 1
\node[card] (c1) at (xa) {\wfbox{\wfHi}{\TWide}{\cOne}};
\node[card] (c2) at (xc) {\wfbox{\wfHi}{\TW}{\cTwo}};
% ------------------------------------------------------------------ row 2: sources
\coordinate (r2) at ([yshift=-8mm]c1.south);
\node[card] (s1) at (r2 -| xa) {\wfbox{\wfHii}{\TW}{\sA}};
\node[card] (s2) at (r2 -| xb) {\wfbox{\wfHii}{\TW}{\sB}};
\node[card] (s3) at (r2 -| xc) {\wfbox{\wfHii}{\TW}{\sC}};
% ------------------------------------------------------------------ row 3: updates
\coordinate (r3) at ([yshift=-9mm]s1.south);
\node[card] (u1) at (r3 -| xa) {\wfbox{\wfHiii}{\TW}{\uA}};
\node[card] (u2) at (r3 -| xb) {\wfbox{\wfHiii}{\TW}{\uB}};
\node[card] (u3) at (r3 -| xc) {\wfbox{\wfHiii}{\TW}{\uC}};
% ------------------------------------------------------------------ row 4: decisions
\coordinate (r4) at ([yshift=-6mm]u1.south);
\node[card] (d1) at (r4 -| xa) {\wfbox{\wfHiv}{\TW}{\dA}};
\node[card] (d2) at (r4 -| xb) {\wfbox{\wfHiv}{\TW}{\dB}};
\node[card] (d3) at (r4 -| xc) {\wfbox{\wfHiv}{\TW}{\dC}};
% ------------------------------------------------------------------ row 5: assessment
\coordinate (r5) at ([yshift=-9mm]d1.south);
\node[anchor=north west,inner sep=0pt] (et) at ([xshift=2.2mm,yshift=-1.9mm]r5 -| xa)
  {\parbox[t]{86mm}{\wfstyle\eBody}};
\node[inset,anchor=east,text width=31mm] (ein) at ([xshift=-2.2mm]et.east -| xe) {\eNote};
\begin{pgfonlayer}{wfcard}
  \node[card,anchor=center,inner sep=0pt,
        fit={([xshift=-2.2mm,yshift=1.9mm]et.north west) ([yshift=-1.9mm]et.south west)
             ([yshift=1.9mm]ein.north -| xe) ([yshift=-1.9mm]ein.south -| xe)}] (e) {};
\end{pgfonlayer}
\begin{pgfonlayer}{wfpanel}
  \fill[wfBlue,rounded corners=3mm] ([yshift=3mm]c1.north -| xl) rectangle ([yshift=-3mm]s1.south -| xr);
  \fill[wfTeal,rounded corners=3mm] ([yshift=3mm]u1.north -| xl) rectangle ([yshift=-3mm]d1.south -| xr);
  \fill[wfGold,rounded corners=3mm] ([yshift=3mm]e.north -| xl) rectangle ([yshift=-3mm]e.south -| xr);
\end{pgfonlayer}
% ------------------------------------------------------------------ arrows
\draw[flow] (c1.east) -- (c2.west);
\draw[flow] (c1.south -| s1.north) -- (s1.north);
\draw[flow] (c1.south -| s2.north) -- (s2.north);
\draw[flow,rounded corners=1mm] (c1.south) -- ++(0,-4) -| (s3.north);
\coordinate (refedge) at (134.5,0);
\draw[target,rounded corners=1mm] (c2.east) -- (refedge |- c2.east) |- (e.east);
\foreach \s/\u/\d in {s1/u1/d1,s2/u2/d2,s3/u3/d3}{
  \draw[flow] (\s.south) -- (\u.north);
  \draw[flow] (\u.south) -- (\d.north);
  \draw[flow] (\d.south) -- (\d.south |- e.north);
  \draw[iterate] ([yshift=1.6mm]\d.east) -- ++(2.6,0) |- ([yshift=-1.6mm]\u.east);
}
\draw[iterate] ([yshift=-3.0mm]c2.east) -- ++(2.6,0) |- ([yshift=3.0mm]c2.east);
\end{tikzpicture}%
\endgroup
\caption{Execution and information flow of the equation-trained neural and equation-based characteristic calculations. Blue cards specify the physical problem, reference calculation and solver inputs; teal cards contain iterative updates; the gold card assesses the frozen endpoint. Solid arrows show execution, and return arrows indicate time stepping or optimisation. The dashed arrow connects the reference solution to the final assessment. The transient branch transports analytic initial and incoming-boundary data and trains source--moment consistency; Algorithms A and B describe stationary neural continuation and characteristic correction. The BGK characteristic sweep preserves positivity for positive inputs; no unconditional positivity property is assigned to the Shakhov target. Optimisation phases and reference refinement are assessed separately.}
\label{ext:fig:workflow}
\end{figure}

\section{Results: neural states and characteristic correction}\label{ext:sec:extresults}
The neural profiles at Mach 2--6 are shown separately in \cref{ext:fig:neuralseries,ext:fig:neuralM3,ext:fig:neuralM4,ext:fig:neuralM5,ext:fig:neuralM6}.
\subsection{Neural solutions}
The neural calculations use Shakhov collisions at Mach 2--5 and BGK collisions at Mach 6 and 7. Consequently, differences across this series reflect both Mach number and collision model. The Mach-6 solution is evaluated at update 7450 of the collision-model continuation described in \cref{ext:sec:methodA}; the Mach-7 comparison uses its final three-update continuation state. Each result represents one trajectory.

Figures~\ref{ext:fig:neuralseries}--\ref{ext:fig:neuralserieslast} show the neural profiles. Heat-flux/stress discrepancies (\%) are $1.1623/1.2800$, $0.9725/1.0578$, $0.7838/0.8658$, $1.0752/0.7189$, $0.6898/0.7740$ and $3.3684/1.3413$ for Mach 2--7. The comparison half-widths are $9.9,\allowbreak14.85,\allowbreak19.8,\allowbreak24.75,\allowbreak29.7,\allowbreak35$, respectively, with 133 positions for Mach 2--6 and 161 for Mach 7. All profiles use unchanged model coordinates and collision-matched DVM references. Complete space--velocity--time convergence of these references has not been established.

\begin{figure}
\centering\includegraphics[width=\textwidth]{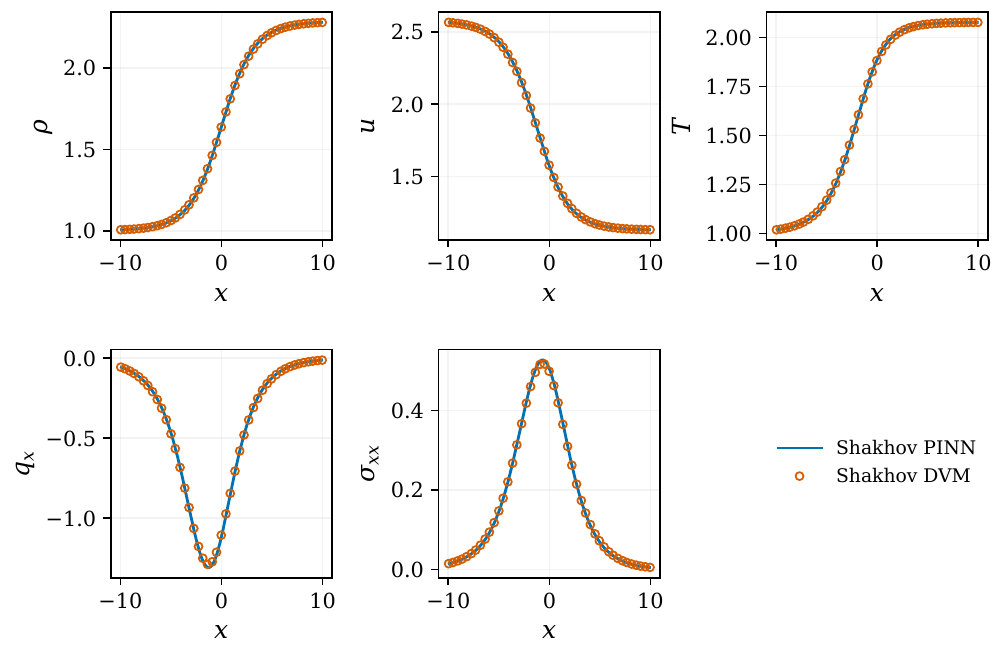}
\caption{Equation-trained Mach-2 Shakhov neural profiles ($\Pr=2/3$), compared with the matched DVM reference. Panels show $\rho,u,T,q_x,\sigma_{xx}$; blue lines show the PINN and orange open symbols show the DVM reference. The sampled core is $|x|\leq 9.9$ with 133 model positions.}
\label{ext:fig:neuralseries}
\end{figure}

\begin{figure}
\centering\includegraphics[width=\textwidth]{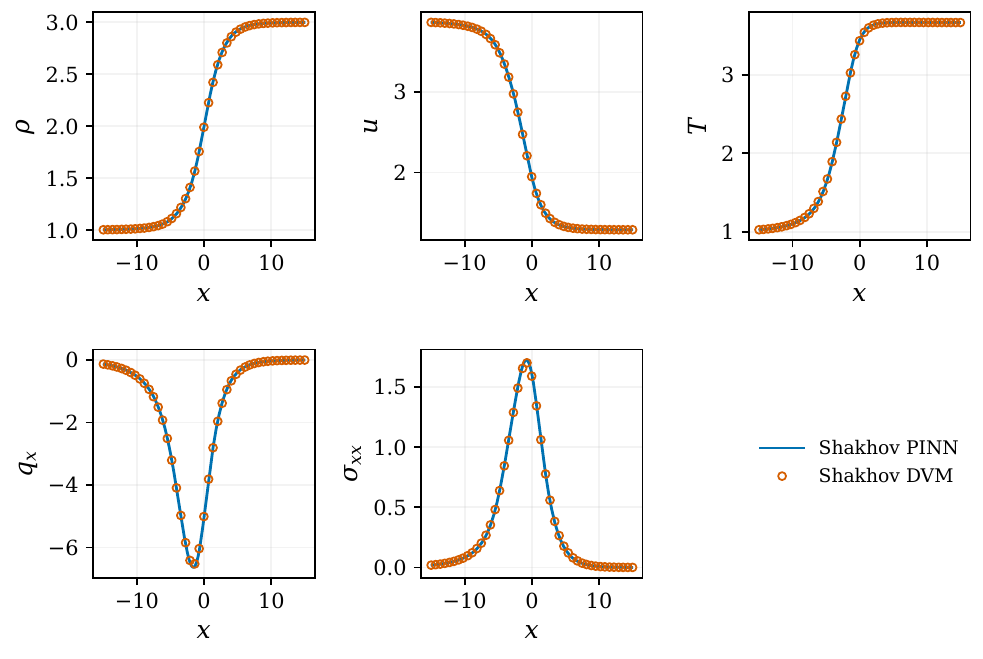}
\caption{Equation-trained Mach-3 Shakhov neural profiles ($\Pr=2/3$), compared with the matched DVM reference. Panels show $\rho,u,T,q_x,\sigma_{xx}$; blue lines show the PINN and orange open symbols show the DVM reference. The sampled core is $|x|\leq 14.85$ with 133 model positions.}
\label{ext:fig:neuralM3}
\end{figure}

\begin{figure}
\centering\includegraphics[width=\textwidth]{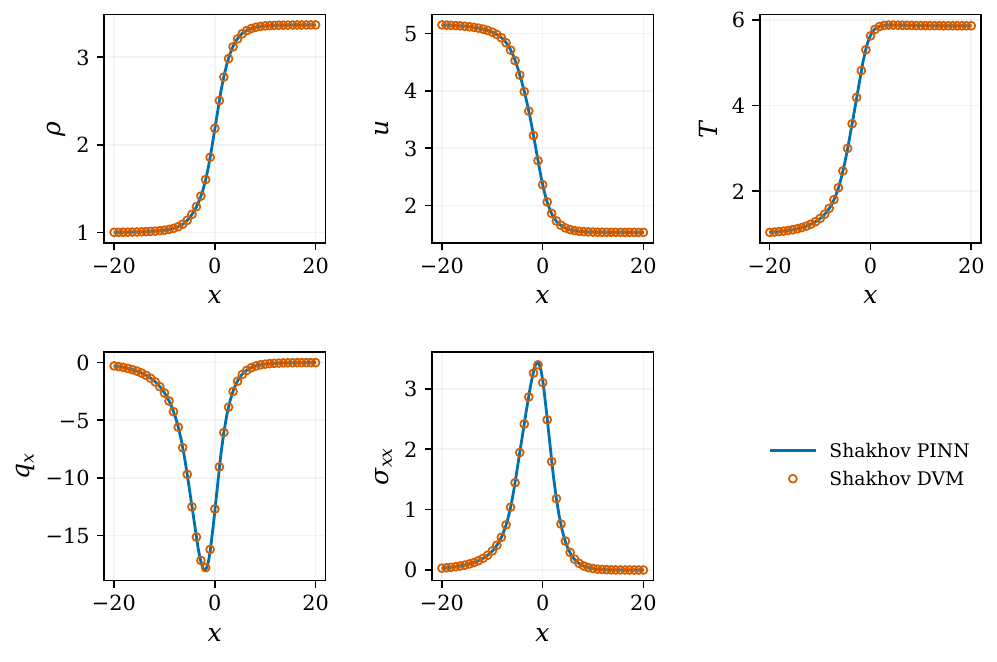}
\caption{Equation-trained Mach-4 Shakhov neural profiles ($\Pr=2/3$), compared with the matched DVM reference. Panels show $\rho,u,T,q_x,\sigma_{xx}$; blue lines show the PINN and orange open symbols show the DVM reference. The sampled core is $|x|\leq 19.8$ with 133 model positions.}
\label{ext:fig:neuralM4}
\end{figure}

\begin{figure}
\centering\includegraphics[width=\textwidth]{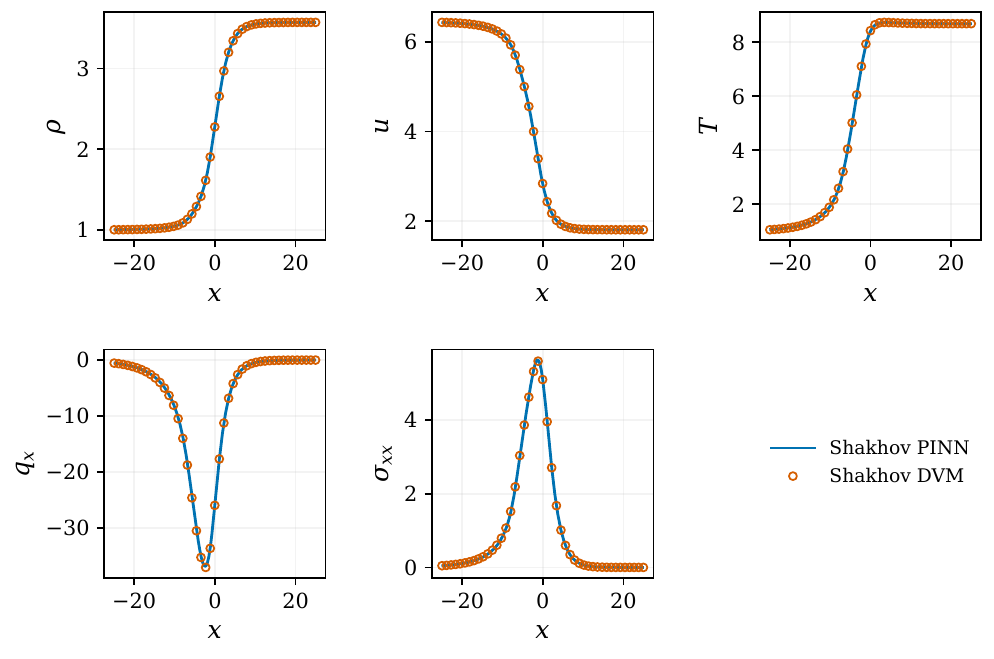}
\caption{Equation-trained Mach-5 Shakhov neural profiles ($\Pr=2/3$), compared with the matched DVM reference. Panels show $\rho,u,T,q_x,\sigma_{xx}$; blue lines show the PINN and orange open symbols show the DVM reference. The sampled core is $|x|\leq 24.75$ with 133 model positions.}
\label{ext:fig:neuralM5}
\end{figure}

\begin{figure}
\centering\includegraphics[width=\textwidth]{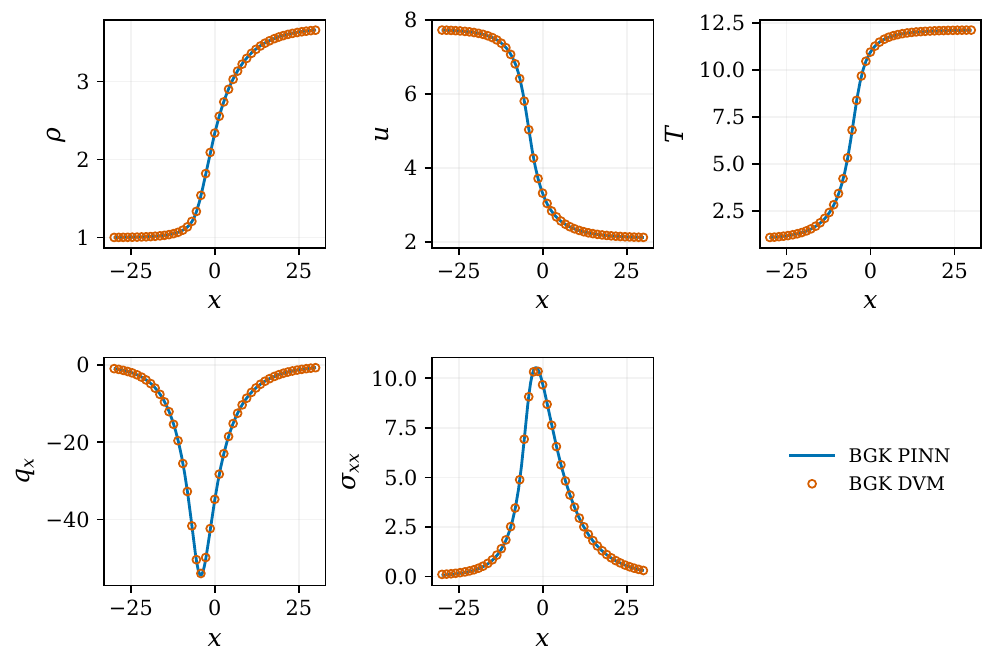}
\caption{Equation-trained Mach-6 BGK neural profiles ($\Pr=1$), compared with the matched DVM reference. Panels show $\rho,u,T,q_x,\sigma_{xx}$; blue lines show the PINN and orange open symbols show the DVM reference. The sampled core is $|x|\leq 29.7$ with 133 model positions.}
\label{ext:fig:neuralM6}
\end{figure}

\begin{figure}
\centering\includegraphics[width=\textwidth]{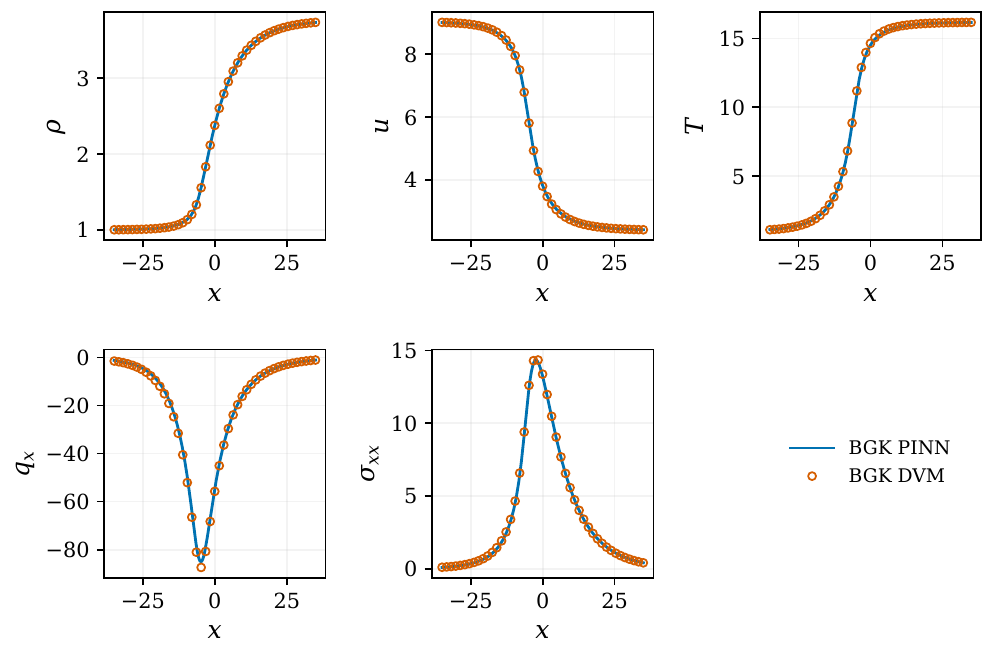}
\caption{Equation-trained Mach-7 BGK neural profiles ($\Pr=1$), compared with the matched DVM reference. Panels show $\rho,u,T,q_x,\sigma_{xx}$; blue lines show the PINN and orange open symbols show the DVM reference. The sampled core is $|x|\leq 35$ with 161 model positions.}
\label{ext:fig:neuralserieslast}
\end{figure}

\subsection{Higher central moments of the neural solutions}
\label{ext:sec:higher}
Figures~\ref{ext:fig:higherM4}, \ref{ext:fig:higherM5} and~\ref{ext:fig:higherM6} extend the Mach-4--6 comparisons to higher moments.

The five lower fields do not determine the higher central moments. We therefore supplement the corresponding profile comparisons with the saved third- and fourth-order observables in figures~\ref{ext:fig:higherM3}--\ref{ext:fig:higherM7}. For these calculations the plotted quantities are those of \cref{eq:mxxx,eq:Rxx,eq:Delta}; with $e=|\boldsymbol c|^2$ and $\langle a\rangle_f=\int a f\,\mathrm d\boldsymbol v$, they read
\begin{align}
 m_{xxx}&=\langle c_x^3-3c_x e/5\rangle_f, &
 R_{xx}&=\langle e(c_x^2-e/3)\rangle_f-7T\sigma_{xx},\label{ext:eq:rawhigher}\\
 \Delta&=\langle e^2\rangle_f-15\rho T^2, &
 S_{xx}&=R_{xx}+\Delta/3.\label{ext:eq:combinedhigher}
\end{align}
The additional contracted excess is $C_{xx}=\langle c_x^2e\rangle_f-5\rho T^2=S_{xx}+7T\sigma_{xx}$. These definitions integrate the full distribution. The BGK production quantities, $Q_{xx}^{\rm BGK}=-\sigma_{xx}$ and $Q_x^{\rm BGK}=-2q_x$ at $\tau=1$, are dependent collision moments rather than additional independent tail tests; their BGK expressions are not applied to the Shakhov cases.

Table~\ref{ext:tab:higherneural} reports all convention-matched saved comparisons. For Mach 6, the reference's separate fourth-order outputs use a different discrete-equilibrium subtraction. Writing $M_{abc}=\int c_x^a c_y^b c_z^c f\,\mathrm d\boldsymbol v$ for the un-subtracted central monomial moments, the available raw components nevertheless give the exact common combination $S_{xx}=M_{400}+M_{220}+M_{202}-7T\sigma_{xx}-5\rho T^2$, so only that matched observable is compared. The model's third-order projection was not retained for Mach 3--6. Missing entries are not zero errors and are not reconstructed from $q_x$ or the macroscopic fields. The Mach-7 neural state has appreciably larger fourth-order discrepancies than Mach 3--5, despite its close lower-field profiles.

\begin{table}
\centering
\caption{Higher-moment discrepancies of the equation-trained neural solutions (\%). The norm is $100[\sum_i(a_i-a_i^{\rm DVM})^2/\sum_i(a_i^{\rm DVM})^2]^{1/2}$ at the original model positions, using the existing reference coordinate convention. The sampled half-widths are $14.85,19.8,24.75,29.7,35$ for Mach 3--7, with 133 positions for Mach 3--6 and 161 for Mach 7. A dash denotes unavailable convention-matched evidence.}
\label{ext:tab:higherneural}
\begin{tabular}{llrrrr}\toprule
Mach & Collision model & $m_{xxx}$ & $R_{xx}$ & $\Delta$ & $S_{xx}$\\\midrule
3 & Shakhov & -- & 1.53429 & 1.40810 & 1.47913\\
4 & Shakhov & -- & 1.18587 & 1.09026 & 1.14057\\
5 & Shakhov & -- & 1.09517 & 1.20968 & 1.12356\\
6 & BGK & -- & -- & -- & 1.17553\\
7 & BGK & 3.29781 & 6.19431 & 6.45125 & 6.45648\\\bottomrule
\end{tabular}
\end{table}

\begin{figure}
\centering\includegraphics[width=\textwidth]{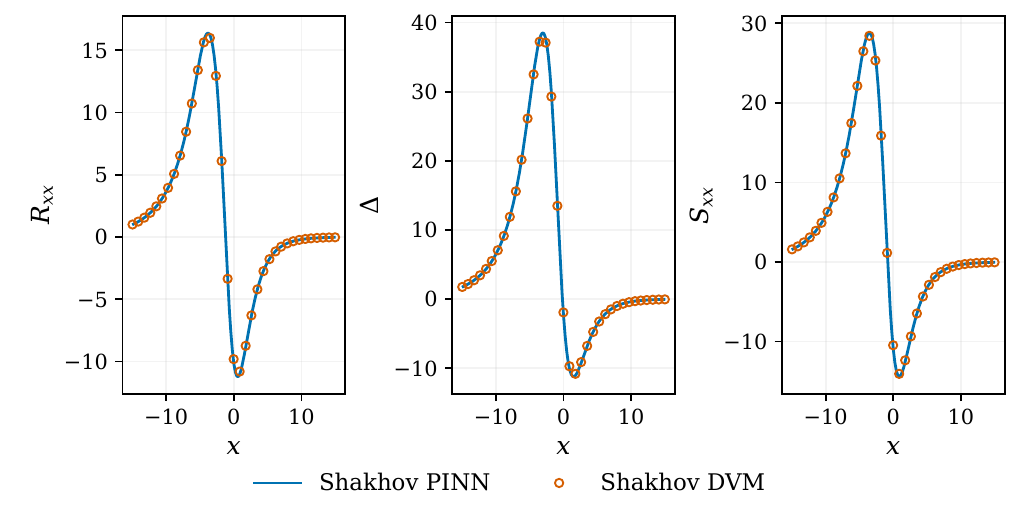}
\caption{Equation-trained Mach-3 Shakhov higher moments accompanying the lower-field comparison: $R_{xx}$, $\Delta$ and $S_{xx}$ from left to right, with the full-distribution definitions in equations~\eqref{ext:eq:rawhigher}--\eqref{ext:eq:combinedhigher}. Blue lines show the same neural endpoint; orange open symbols show its matched DVM reference. Errors use all 133 model positions on $|x|\leq 14.85$, not the display subsample.}
\label{ext:fig:higherM3}
\end{figure}

\begin{figure}
\centering\includegraphics[width=\textwidth]{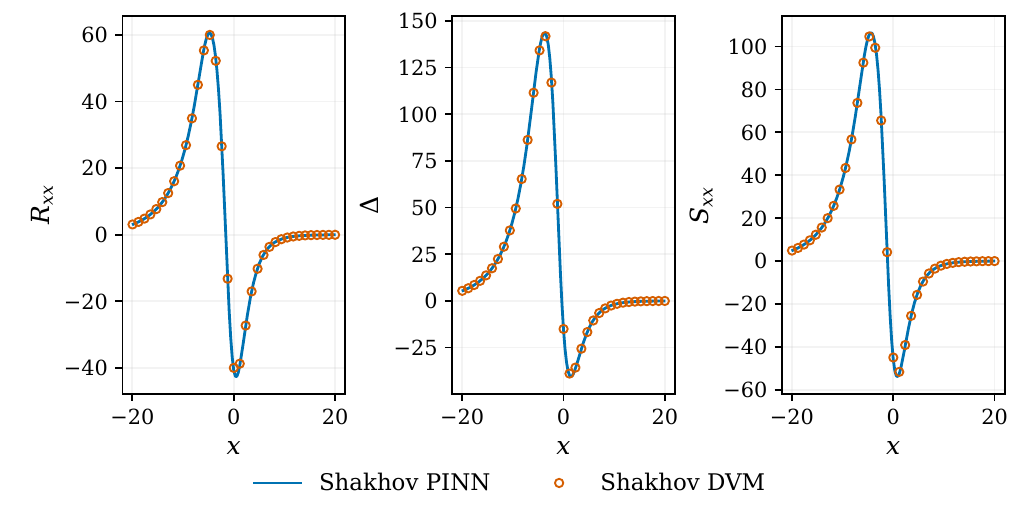}
\caption{Equation-trained Mach-4 Shakhov higher moments accompanying the lower-field comparison: $R_{xx}$, $\Delta$ and $S_{xx}$ from left to right, with the full-distribution definitions in equations~\eqref{ext:eq:rawhigher}--\eqref{ext:eq:combinedhigher}. Blue lines show the same neural endpoint; orange open symbols show its matched DVM reference. Errors use all 133 model positions on $|x|\leq 19.8$, not the display subsample.}
\label{ext:fig:higherM4}
\end{figure}

\begin{figure}
\centering\includegraphics[width=\textwidth]{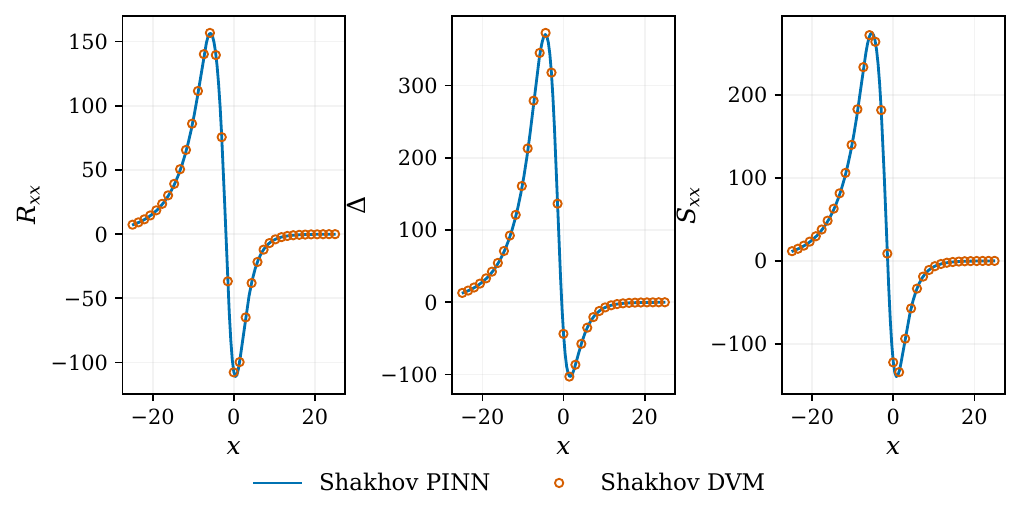}
\caption{Equation-trained Mach-5 Shakhov higher moments accompanying the lower-field comparison: $R_{xx}$, $\Delta$ and $S_{xx}$ from left to right, with the full-distribution definitions in equations~\eqref{ext:eq:rawhigher}--\eqref{ext:eq:combinedhigher}. Blue lines show the same neural endpoint; orange open symbols show its matched DVM reference. Errors use all 133 model positions on $|x|\leq 24.75$, not the display subsample.}
\label{ext:fig:higherM5}
\end{figure}

\begin{figure}
\centering\includegraphics[width=0.80\textwidth]{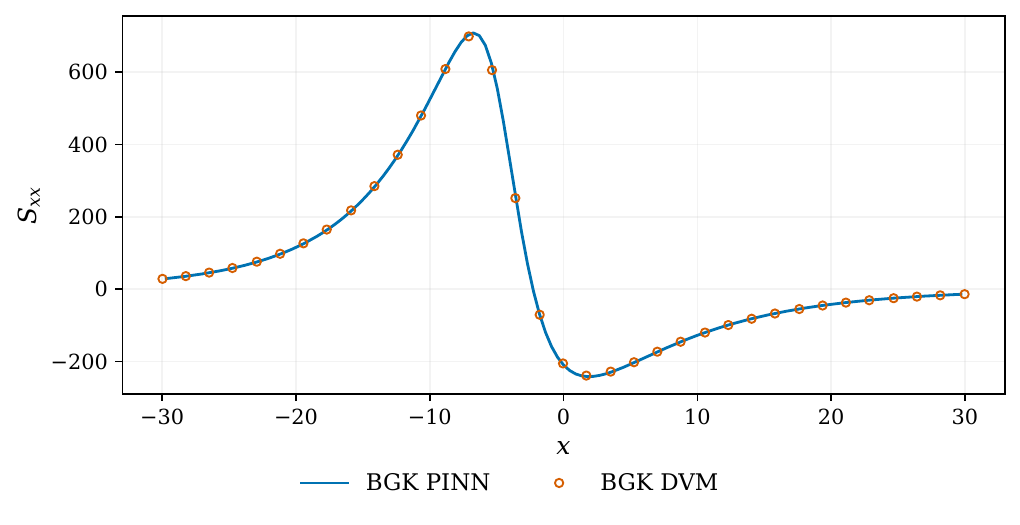}
\caption{Equation-trained Mach-6 BGK contracted fourth-order combination $S_{xx}$ accompanying the step-7450 lower-field profiles. The reference is reconstructed algebraically from its saved raw fourth components, with the same full-distribution convention as the neural output. The relative discrepancy is $1.17553\%$ on the 133 core positions. This agreement does not determine separate $R_{xx}$ or $\Delta$ errors.}
\label{ext:fig:higherM6}
\end{figure}

\begin{figure}
\centering\includegraphics[width=\textwidth]{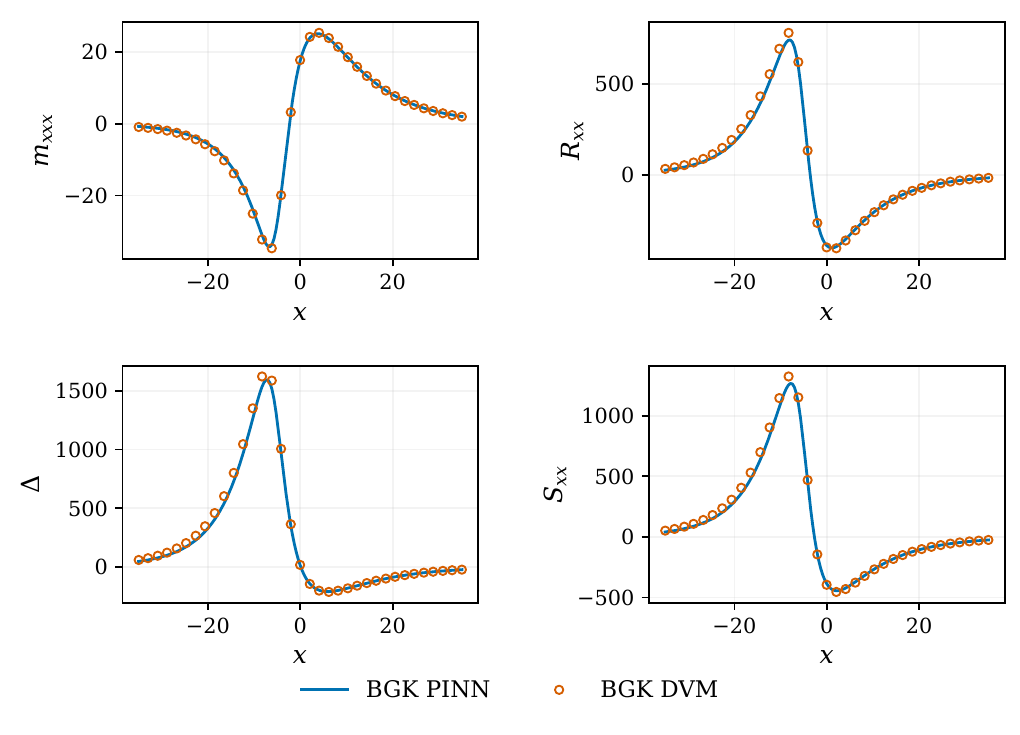}
\caption{Equation-trained Mach-7 BGK neural higher moments: $m_{xxx}$, $R_{xx}$, $\Delta$ and $S_{xx}$, ordered by rows. The same neural solution and matched DVM reference as the lower-field comparison are used. All 161 core positions enter the errors in table~\ref{ext:tab:higherneural}; symbols are only display subsamples.}
\label{ext:fig:higherM7}
\end{figure}

\FloatBarrier

\subsection{Mach-7 characteristic correction}\label{ext:sec:m7char}
The Mach-7 correction starts from the neural macroscopic fields. An analytic RH density-midpoint condition fixes the spatial origin; the initial velocity and temperature are interpolated from the neural fields, and density is set by $\rho=J_m/u$. The neural state and this conditioned initial state are compared separately.

The correction uses 481 spatial points on $[-105,105]$, 24 compact cubic basis functions and 47 free macroparameters, with 512 axial and 144 radial velocity nodes and cutoff 42. Two accepted unit-fraction steps reduce $J_{\rm char}$ from $8.80075\times10^{-7}$ through $1.41095\times10^{-7}$ to $3.16903\times10^{-8}$. There are zero neural-weight updates.

Table~\ref{ext:tab:m7} and figure~\ref{ext:fig:m7profiles} compare the neural state, the conditioned initial state and the characteristic solution at the same 161 positions on $[-35,35]$. Characteristic correction reduces all five discrepancies relative to the neural state, although conditioning alone increases the density error.

\begin{table}
\centering
\caption{Equation-based Mach-7 core relative $L_2$ discrepancies (\%). The reference is not fully convergence-qualified.}
\label{ext:tab:m7}
\begin{tabular}{lrrr}\toprule
Quantity&neural solution&Conditioned baseline&Two-update candidate\\\midrule
$\rho$&0.07193&0.19334&0.04724\\
$u$&0.20134&0.17554&0.06369\\
$T$&0.48911&0.27625&0.07751\\
$q_x$&3.36838&1.40485&0.41606\\
$\sigma_{xx}$&1.34126&0.78520&0.43422\\\bottomrule
\end{tabular}
\end{table}

On 64 Gauss--Legendre diagnostic points, mean kinetic and tail residuals decrease from $2.78185\times10^{-4}$ and $2.37227\times10^{-4}$ to $1.86852\times10^{-5}$ and $1.75616\times10^{-5}$. Maximum flux and projection defects decrease from $4.36608\times10^{-3}$ and $2.12101\times10^{-2}$ to $5.74570\times10^{-4}$ and $5.18651\times10^{-3}$. Density, temperature and the distribution remain nonnegative on the diagnostic grid. The boundary residual increases from $7.40319\times10^{-10}$ to $1.27731\times10^{-9}$. These residuals quantify the remaining self-consistency error.

\IfFileExists{figures/M7_PILOT_PHYSICAL_PROFILES.pdf}{%
\begin{figure}
\centering\includegraphics[width=\textwidth]{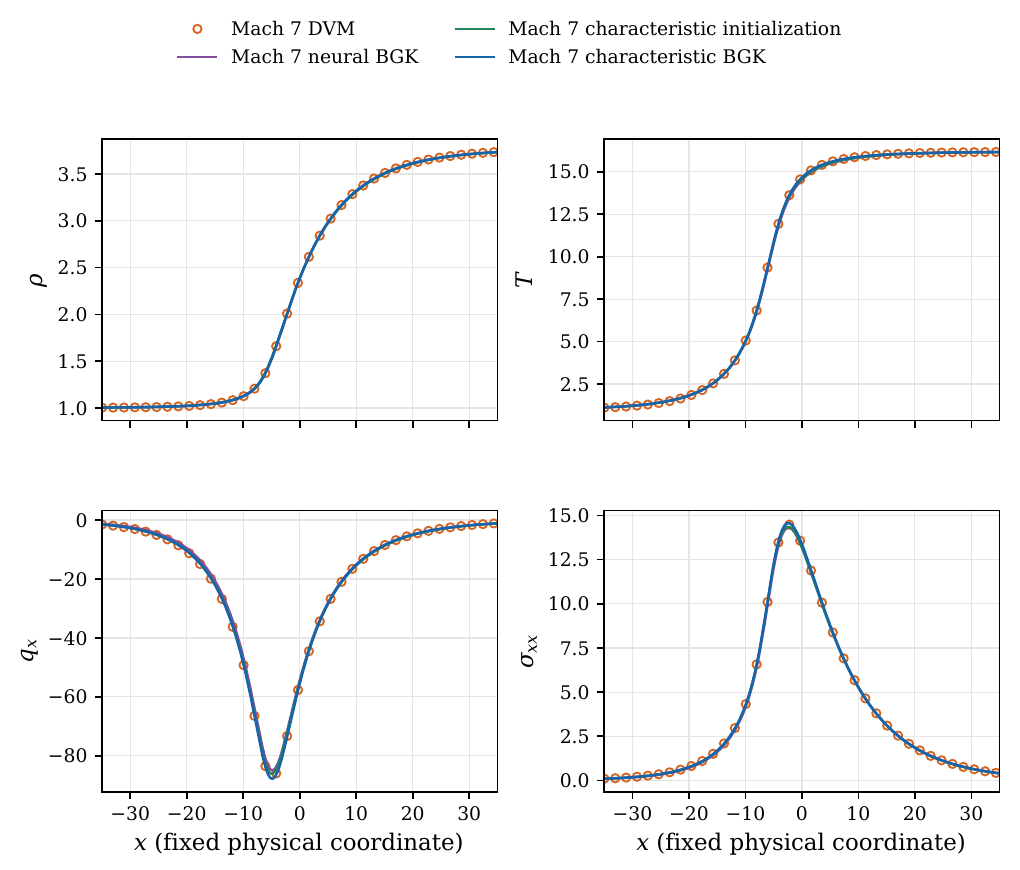}
\caption{Equation-trained neural and equation-based characteristic Mach-7 physical profiles: neural solution, conditioned characteristic baseline, two-update candidate and post-solve reference. The initializer uses only neural macrofields and analytic RH information. No post-hoc model shift is applied. The abscissa is the model's saved coordinate, whose origin is its own Rankine--Hugoniot density midpoint; the reference is centred at its own midpoint (\cref{sec:dvm}).}
\label{ext:fig:m7profiles}
\end{figure}}{}

The characteristic solution also resolves the separately retained fourth-order projections (figure~\ref{ext:fig:charhigher7}). Using the spatial trapezoidal norm on the same 161 core positions gives discrepancies of $0.60249\%$ in $R_{xx}$, $0.54669\%$ in $\Delta$ and $0.58691\%$ in $S_{xx}$. This norm uses $100[\int(a-a^{\rm DVM})^2\,\mathrm dx/\int(a^{\rm DVM})^2\,\mathrm dx]^{1/2}$ with the original saved nodes. With exactly this same norm and denominator, the neural solution discrepancies are $6.19305\%$, $6.45022\%$ and $6.45527\%$, while the conditioned baseline gives $2.38344\%$, $2.30499\%$ and $2.37651\%$, respectively. The characteristic dataset does not retain $m_{xxx}$, so improvement of that separate third-order projection is not inferred.

\begin{figure}
\centering\includegraphics[width=\textwidth]{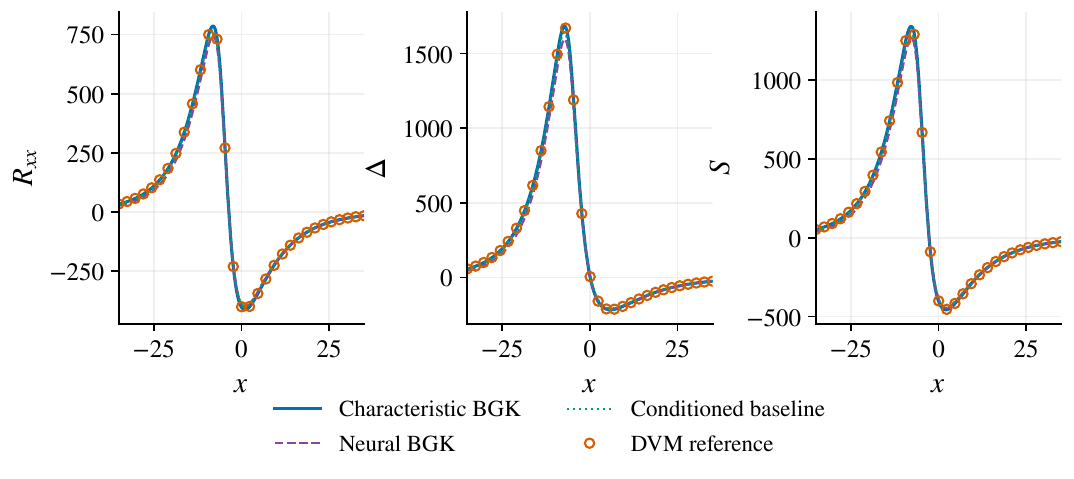}
\caption{Equation-trained neural and equation-based characteristic Mach-7 higher-moment comparison from the retained neural state through characteristic correction. The three panels show $R_{xx}$, $\Delta$ and $S_{xx}$ (labelled $S$ on the right ordinate), with the full-distribution definitions; each curve is identified in the external legend. Reference coordinates and amplitudes are unchanged.}
\label{ext:fig:charhigher7}
\end{figure}

\subsection{Mach-8 overlap comparison}
The overlap comparison starts from the neural Mach-8 state at step 3975. Its own macroscopic fields initialise the positive characteristic correction. The neural weights are not updated by this correction. On the core interval $-39\leq x\leq39$, the recorded relative $L_2$ heat-flux error decreases from $14.1443\%$ to $0.6170\%$ and the stress error from $9.6697\%$ to $0.7918\%$. The density, velocity and temperature errors decrease from $1.9365\%$, $2.8800\%$ and $5.3515\%$ to $0.1435\%$, $0.1239\%$ and $0.2539\%$, respectively. These are errors against the numerical reference on the declared core, not global kinetic residuals. The core comparison contains 87 model positions; the reference is evaluated on the corresponding saved comparison convention, with no additional model alignment.

Table~\ref{ext:tab:m8} compares the Mach-8 field errors before and after characteristic correction.
\begin{table}
\centering
\caption{Equation-based Mach-8 core comparison. All entries are relative $L_2$ errors (\%) on the same declared interval. The correction does not retrain the neural network.}
\label{ext:tab:m8}

\begin{tabular}{lrr}
\toprule
Observable & Neural step 3975 & Characteristic correction\\
\midrule
$\rho$ &1.9365&0.1435\\
$u$ &2.8800&0.1239\\
$T$ &5.3515&0.2539\\
$q_x$ &14.1443&0.6170\\
$\sigma_{xx}$ &9.6697&0.7918\\
\bottomrule
\end{tabular}
\end{table}

The 119-basis Mach-8 solution satisfies the prescribed kinetic, tail, flux, projection, phase, boundary and positive-macro tests, as well as quadrature-sensitivity and restart checks. Its full common-domain stress discrepancy is $1.6681\%$, compared with $0.7918\%$ in the core. Thus the comparison interval materially affects the reported accuracy.

\IfFileExists{figures/M8_characteristic_physical_profiles.pdf}{%
Figure~\ref{ext:fig:m8profiles} compares the Mach-8 neural and characteristic profiles with the matched DVM reference.
\begin{figure}
\centering\includegraphics[width=\textwidth]{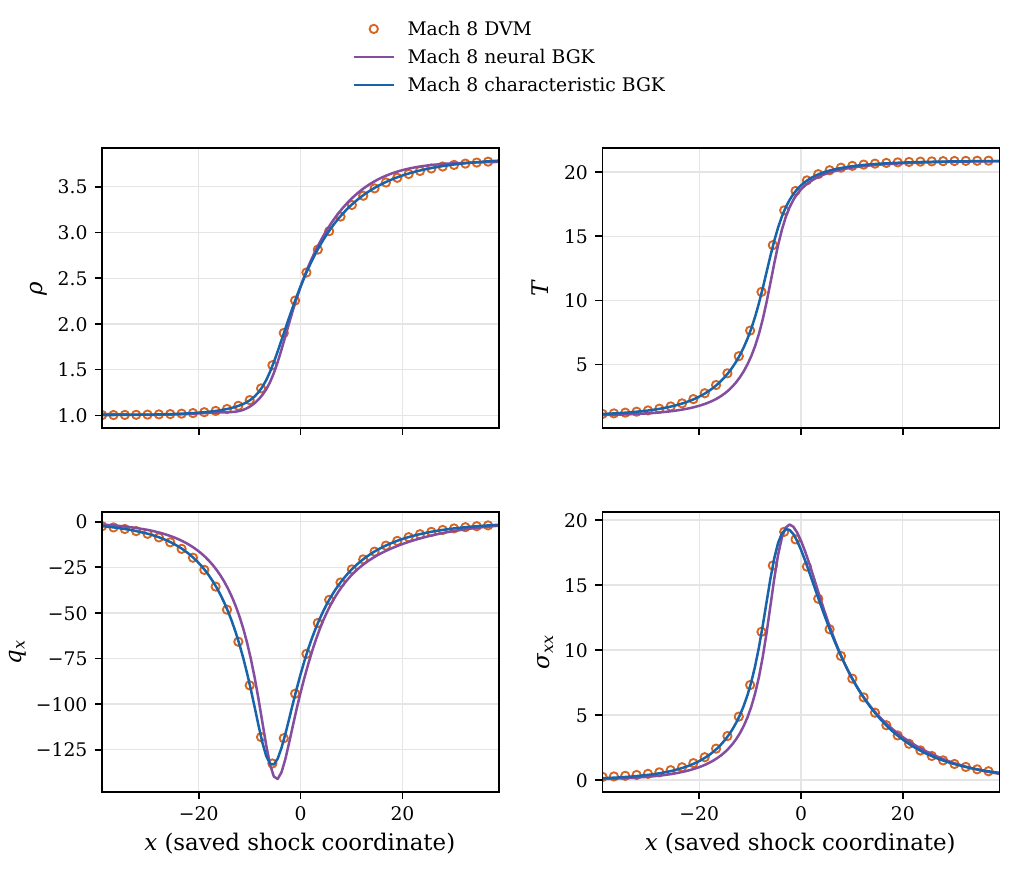}
\caption{Equation-trained neural and equation-based characteristic Mach-8 physical profiles on the core interval. The neural source, characteristic candidate and independent post-solve DVM reference are shown separately. Curves connect actual spatial samples; no additional shift or amplitude adjustment is applied. The abscissa is the model's saved coordinate, whose origin is its own Rankine--Hugoniot density midpoint; the reference is centred at its own midpoint (\cref{sec:dvm}). The selected 119-basis candidate passes the physical checks; global convergence and full reference qualification remain separate.}
\label{ext:fig:m8profiles}
\end{figure}}{}

Figure~\ref{ext:fig:charhigher8} adds the two matched fourth-order combinations retained in both the source and the reference. On the same core, the spatial trapezoidal discrepancies in $S_{xx}$ and $C_{xx}$ decrease from $19.28225\%$ and $11.57459\%$ to $1.07431\%$ and $0.65669\%$, respectively. The reference dataset does not provide the matched separate raw-$f$ $R_{xx},\Delta$ decomposition, and the characteristic dataset does not retain $m_{xxx}$. Agreement of these two combinations is therefore reported directly, without assigning it to an unmeasured component.

\begin{figure}
\centering\includegraphics[width=\textwidth]{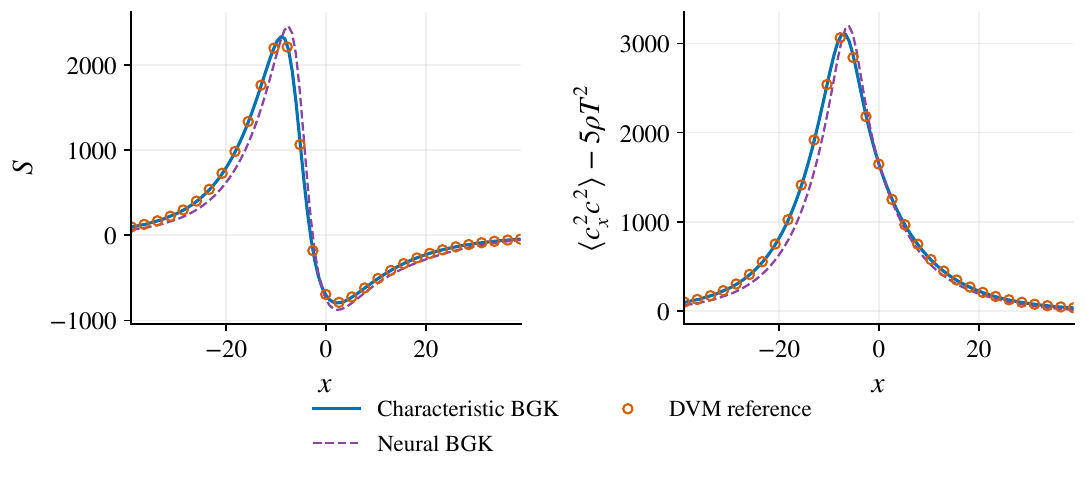}
\caption{Equation-trained neural and equation-based characteristic Mach-8 contracted fourth-order comparisons accompanying the lower-field profiles: $S_{xx}$ (labelled $S$) on the left and $C_{xx}=\langle c_x^2|\boldsymbol c|^2\rangle_f-5\rho T^2$ on the right. The neural source, positive characteristic endpoint and matched DVM reference use the same core $|x|\leq39$. These combinations do not provide a separate $R_{xx}$ error.}
\label{ext:fig:charhigher8}
\end{figure}

\subsection{Mach-12 calculation and reference sensitivity}
The Mach-12 characteristic calculation is initialised from the corrected Mach-8 fields and uses six macroparameter updates. The refined solution passes the prescribed physical and restart checks. Reference sensitivity is substantial: on $|x|\leq61.875$, the heat-flux/stress discrepancies are $1.0266/2.6498\%$ against one DVM reference family, $0.46730/0.48740\%$ against the composite medium-grid reference, and $0.39795/0.41439\%$ against its fine-grid snapshot. The full-domain stress discrepancy against the first family is $5.4549\%$. The composite reference has incomplete collision-law and space--velocity--time qualification, so these values cannot establish a reference-independent accuracy. All comparisons use central stress integrated from the full distribution. Figure~\ref{ext:fig:m12medium} shows the medium-grid comparison.

The Mach-7, Mach-8 and Mach-12 calculations use different approximation spaces and diagnostic grids; their error levels are compared individually rather than interpreted as a common convergence sequence.

The corresponding fourth-order profiles are shown in figure~\ref{ext:fig:charhigher12}. Against the same provisional medium reference, the core spatial trapezoidal discrepancies are $0.71189\%$ in $S_{xx}$ and $0.53328\%$ in $C_{xx}$. The full matched reference decomposition required for separate $R_{xx}$ and $\Delta$ comparisons is unavailable in these data, and no candidate $m_{xxx}$ is retained. These are qualified statements about two contracted observables, not all higher moments.

\begin{figure}
\centering\includegraphics[width=\textwidth]{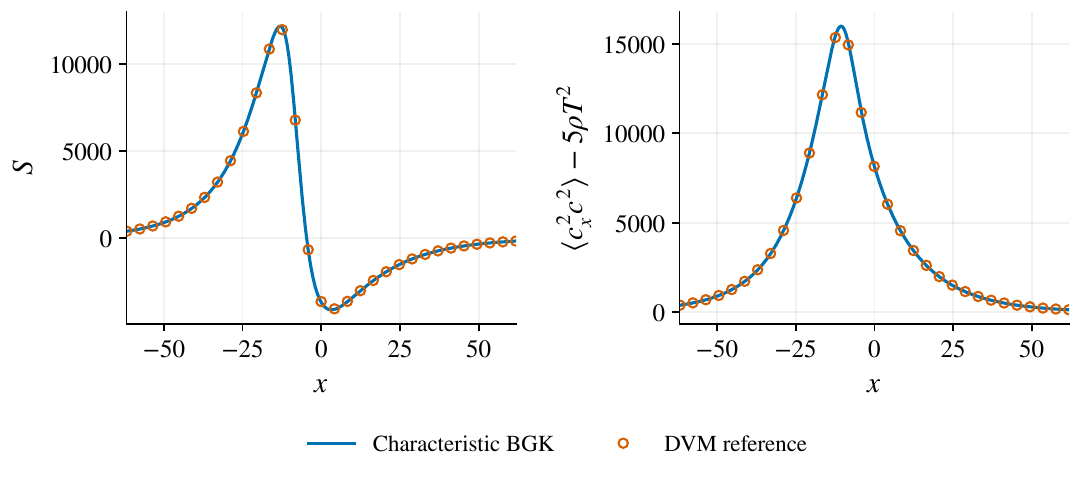}
\caption{Equation-based Mach-12 positive characteristic BGK fourth-order combinations accompanying its physical profiles, compared with the matched provisional medium DVM reference. Left: $S_{xx}$ (labelled $S$); right: $C_{xx}=\langle c_x^2|\boldsymbol c|^2\rangle_f-5\rho T^2$. The original core $|x|\leq61.875$ contains 2281 model positions. Reference qualification remains as stated in the main comparison.}
\label{ext:fig:charhigher12}
\end{figure}

\subsection{Optimisation histories and cost}
Neural and characteristic loss curves represent $J_{\rm NN}$ and $J_{\rm char}$, respectively. The Mach-8 coarse correction has 12 and 8 updates in its two stages; the second-stage objective decreases from $6.8317\times10^{-6}$ to $2.7964\times10^{-7}$. One update in the 119-basis refinement reduces $J_{\rm char}$ from $3.98739\times10^{-9}$ to $5.80251\times10^{-10}$. The phase-specific curves in figures~\ref{ext:fig:recordedloss} and~\ref{ext:fig:charloss} use separate axes because their objectives and approximation spaces differ.

The independent-seed study is reported in \cref{ext:sec:independentseeds}; the higher-Mach curves describe individual optimisation trajectories.

\begin{figure}
\centering\includegraphics[width=\textwidth]{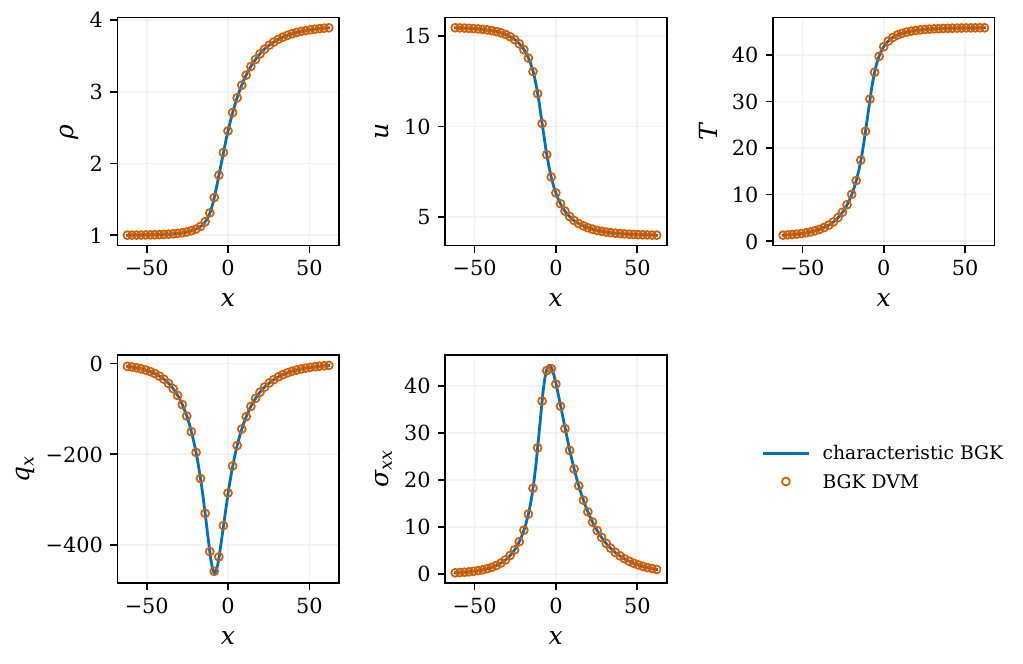}
\caption{Equation-based Mach-12 positive characteristic BGK fields from the 235-basis candidate, compared with the matched provisional medium DVM reference. Panels show $\rho,u,T,q_x,\sigma_{xx}$ with original coordinates and signs; blue lines are the characteristic solution and orange open symbols are native reference samples. No model alignment or amplitude scaling is applied.}
\label{ext:fig:m12medium}
\end{figure}

\FloatBarrier

\section{Discussion}\label{sec:discussion}
Moment supervision and equation training address complementary aspects of kinetic approximation. The shock-tube anchor studies identify the effect of supplying transport moments, and the stationary closure experiments examine the information needed for fourth-order recovery. The constrained continuation improves the supervised shock-tube profiles while controlling field errors, extrema and physical residuals separately. A calibrated closure-head output remains distinct from a fourth moment integrated from the predicted distribution.

The base shock tube shows what is needed to recover third- and fourth-order moments from the distribution of a data-assisted PINN. With five fields these moments are uncontrolled and $R_{xx}$ and $\Delta$ have the wrong sign. With kernels and observations for the three additional moments, the discrepancies fall to $10$--$11\%$, and most of the remaining error lies in the unobserved upstream tail, where neither the observations nor the sampled residual constrain the representation and the result is set by the spatial envelope of the amplitudes. Giving the envelope the exponential far-field decay of the initial condition, and then refining the output layers with the lower-field losses and the physical components protected, reduces the discrepancies to $1.8$--$2.0\%$. The spatial envelope of a positive macro--micro representation is therefore part of the physical model and should respect the known asymptotic behaviour of the solution. The resulting state remains a constrained fit rather than a converged kinetic solution: its sampled balance residuals are between $0.13$ and $0.23$ of the time-derivative terms, its far tails retain relative discrepancies of up to about $50\%$ and, where the moments are smallest, the wrong sign, and at this short time its kinetic residual does not resolve the collision term.

The characteristic shock-tube PINNs recover transport moments using their own collision-source moments and analytic initial and boundary distributions. Heat flux, stress and higher moments are integrated from the same kinetic state. For these equation-trained calculations, the discrepancies in transport and fourth-order moments are comparable to changes produced by doubling the DVM spatial resolution. The velocity-grid and characteristic-quadrature checks address additional numerical errors and are reported separately. Against the 6400-cell reference of the data-assisted base case, the equation-trained BGK prediction has discrepancies of $0.21\%$ and $0.18\%$ in $q_x$ and $\sigma_{xx}$ and of $0.41\%$, $0.85\%$ and $0.70\%$ in $m_{xxx}$, $R_{xx}$ and $\Delta$, each within $0.003$ percentage points of the 3200-to-6400-cell refinement difference of the reference (Appendix~\ref{ext:app:reference}). The data-assisted values, $0.33\%$, $0.67\%$ and $1.8$--$2.0\%$, are larger by factors of $1.6$ to $4.4$, although that state is fitted to this reference. For this short-time problem the equation-trained characteristic PINN is therefore the more accurate of the two, and it uses no reference data; the data-assisted case shows which properties of a positive macro--micro representation limit the recovery of higher moments from observations.

The calculations also show what simultaneous accuracy of the lower and higher moments requires.\label{sec:discussion:challenges} First, the third- and fourth-order kernels weight the velocity tails as $|\bc|^3$ and $|\bc|^4$, so that small distribution errors there are amplified and the velocity cutoff and quadrature must resolve those tails; in a positive data-assisted representation the spatial envelope of the amplitudes must, in addition, follow the far-field decay of the solution (\S\ref{sup:sec:eight:diag}). The stationary Mach-7 neural solution shows the resulting ordering: its density, velocity and temperature are within $0.5\%$ of the reference (table~\ref{ext:tab:m7}), its stress, heat flux and $m_{xxx}$ differ by $1.3\%$, $3.4\%$ and $3.3\%$, and its fourth-order moments by $6.2$--$6.5\%$ (table~\ref{ext:tab:higherneural}); characteristic correction on the same velocity grid reduces the latter to $0.5$--$0.6\%$ (\S\ref{ext:sec:m7char}, figure~\ref{ext:fig:charhigher7}). Second, as is well established, losses on the lower moments do not constrain the higher ones (\cref{lem:observability}): the five-field state matches heat flux and stress to about $3\%$, while $R_{xx}$ and $\Delta$ have the wrong sign in the core of the layer (table~\ref{sup:tab:eightmoment}, figure~\ref{sup:fig:shocktube}), and the calibrated output $\widehat R_{xx}^{\rm cl}$ of \eqref{sup:eq:head} reduces the fourth-order discrepancy without changing the fourth moment of the distribution (\S\ref{sup:sec:closure}). Third, fitting the higher moments alone can degrade the lower ones. In trial runs of the eight-moment refinement, updating only the higher-moment output layer raised the heat-flux and stress losses (Appendix~\ref{ext:app:supervised}); in the residual ablation of table~\ref{sup:tab:rphys}, the $H_R$-weighted residual lowers the $R_{xx}^{\rm cl}$ discrepancy only from $91.6\%$ to $90.5\%$, while those of $q_x$, $\sigma_{xx}$ and $m_{xxx}^{\rm cl}$ rise from $1.96\%$, $4.96\%$ and $10.2\%$ to $2.76\%$, $5.41\%$ and $10.6\%$. The output-layer refinement of the eight-moment state avoids this for the lower fields: the loss of each lower field and of each physical component is constrained not to increase on the training points, and the lower-field residuals at the observation locations enter the least-squares step; the full-domain lower-field discrepancies then also decrease (table~\ref{sup:tab:eightmoment}), although two physical measures on the separate points rise by up to $0.8\%$ (\S\ref{sup:sec:eight:results}). Fourth, the higher moments of the reference carry larger discretisation errors: between 3200 and 6400 cells they change by $0.41$--$0.86\%$, against $0.18$--$0.20\%$ for heat flux and stress, although each difference is about half of its 1600-to-3200-cell value (Appendix~\ref{ext:app:reference}); higher-moment discrepancies of this size must therefore be read against reference refinement.

For the short-time BGK shock tube, the equation-trained PINN obtains both orders from one distribution to within the resolution of the reference: against the first-order extrapolate of the 3200- and 6400-cell solutions, an indicative comparison, its discrepancies in the five non-equilibrium moments are $0.002$--$0.010\%$ (Appendix~\ref{ext:app:reference}). Because $m_{xxx}$ and $S_{xx}$ enter the fluxes \eqref{eq:stress_flux} and \eqref{eq:heat_flux_flux}, the same holds for transport: against the extrapolate, the stress flux \eqref{eq:stress_flux} of this distribution and the non-equilibrium part $\frac12S_{xx}+\frac72T\sigma_{xx}$ of \eqref{eq:heat_flux_flux} differ by $0.005\%$ and $0.003\%$, and those of the data-assisted state by $0.73\%$ (Appendix~\ref{ext:app:reference}).\label{sec:discussion:challenges:end}

A different route closes the higher moments instead of integrating them. \citet{KimJun2026FPPINN} compute a Mach-10 normal shock in argon with a PINN whose only outputs are $\rho$, $u$, $T$, $\sigma_{xx}$ and $q_x$. Their loss is the ellipsoidal-statistical Fokker--Planck equation evaluated with a Grad 26-moment distribution, whose $m_{xxx}$, $R_{xx}$ and $\Delta$ follow from regularised 13-moment constitutive relations, and projected onto moments up to the heat flux; it also prescribes the Rankine--Hugoniot values of all five outputs over upstream and downstream regions five mean free paths wide and fixes the shock centre. The network therefore solves a closed moment system rather than the kinetic equation. Its profiles agree with a numerical solution of the same moment system to within $1\%$, but differ from particle Fokker--Planck simulations by $1.5$--$5.1\%$ in density, velocity and temperature and by $13\%$ and $17\%$ in normal stress and heat flux, which the authors attribute to the closure and the projection. The collision model, gas and Mach number differ from those of the present calculations, so these values are not directly comparable with ours; the comparison shows, however, that when the higher moments are supplied by a closure, the accuracy of the stress and heat flux is bounded by that closure, which is why the higher moments are integrated here from the predicted distribution and assessed against the kinetic reference.

The three independently initialised stationary Mach-2 runs show a small spread in field errors at the stated budget, although not all physical residual tolerances are met. Optimiser histories must therefore be interpreted with the field and residual measurements. A characteristic sweep integrates transport for a supplied source and optimises its self-consistency, whereas a direct neural distribution must also resolve transport derivatives. Their objective values measure different errors; signed profiles and moment budgets provide the common comparison.

The present calculations are one-dimensional. The equation-trained shock-tube interval covers a short fraction of a relaxation time, with one independent transient run per collision model. The data-assisted base case is likewise a single training history, and its observations and its comparison share one reference. The separate Mach-2 initialisation study measures repeatability for its specified stationary Shakhov setting. Extension to multidimensional shocks requires spatial, velocity and boundary-resolution studies. The deep-cavity calculations of \citet{RoohiMcDevitt2026DeepCavity} provide related experience with physics-constrained optimisation and local accuracy assessment.\label{sec:discussion:end}

\section{Conclusions}\label{sec:conclusion}
PINNs with moment supervision and PINNs trained from the kinetic equations recover complementary information about nonequilibrium shock structure. In the data-assisted BGK shock tube, constrained least-squares continuation of the five-field representation matches the heat flux and stress of its coarse reference to $1.18\%$ and $1.34\%$ but leaves the third- and fourth-order moments uncontrolled. An eight-moment representation with a common-tail envelope and an output-layer refinement gives heat-flux and stress discrepancies of $0.33\%$ and $0.67\%$ and third- and fourth-order discrepancies of $1.8$--$2.0\%$ against the 6400-cell reference that also supplies its observations; this state is a constrained fit from a single training history, and it is less accurate than the equation-trained BGK calculation, which gives $0.21\%$ and $0.18\%$ for heat flux and stress and $0.4$--$0.9\%$ for the third- and fourth-order moments against the same reference. The anchor and closure experiments show how supplied moment observations affect transport and fourth-order recovery. That lower-order agreement does not determine the higher moments is established kinetic theory; the aim here has been accuracy of both in one predicted distribution, because the higher moments enter the fluxes of stress and heat flux, and the contribution is to quantify, for data-assisted and equation-trained PINNs, which moments are recovered and which properties of the representation and of the supplied information limit them.\label{sec:conclusion:contribution}

The equation-trained BGK and Shakhov characteristic PINNs give shock-tube heat-flux and stress discrepancies of $0.41\%$ and $0.36\%$ against the 3200-cell references. Fourth-order discrepancies are larger and comparable to reference-refinement changes. Three independently initialised stationary Shakhov Mach-2 calculations give $1.04\pm0.06\%$ for heat flux and $1.27\pm0.02\%$ for stress at the stated budget. Adam-to-SSB continuation improves both fields substantially; quadrature refinement further reduces heat-flux error with a small increase in stress error.

Characteristic correction improves the stationary BGK profiles at Mach~7, 8 and~12. Across these regimes, the physical interpretation follows from the collision model, the training information and the kinetic representation, assessed through signed profiles, higher moments, conservation and reference refinement.

\appendix
\section{Reproducibility and detailed configurations}
\label[appsec]{ext:app:reproducibility}

This appendix gives the objective weights, numerical resolutions, optimisation coefficients and independent-initialisation tests. Each schedule identifies its initial state, optimiser and objective. Changes of objective, quadrature or approximation space are reported as separate stages.

\subsection{Numerical configurations and update counts}

The fixed Mach-6 objective coefficients are listed in table~\ref{ext:tab:m6weights}.

\begin{table}
\centering
\caption{Mach-6 fixed-objective weights. A zero coefficient disables a term; the number of defined terms is not the number of active penalties.}
\label{ext:tab:m6weights}
\begin{tabular}{lr}
\toprule
Term & Weight\\
\midrule
Kinetic, tail, raw & 1 each\\
Flux, phase & 30 each\\
Invariant & 0\\
Stress, heat, fourth deviatoric, fourth scalar & 0.1 each\\
Outgoing boundary & 0\\
\bottomrule
\end{tabular}
\end{table}

The Mach-6 calculation uses seed 101, three hidden layers of width 64, initial mixture width 3 and correction bound $\psi_{\max}=6$. It uses 187 spatial points and a $256\times80$ velocity grid with cutoff 36; the finer velocity test uses $384\times112$ nodes. Acceptance tolerances are $10^{-3}$ for maximum flux and projection defects, $10^{-6}$ for kinetic, tail, boundary and phase residuals, and $10^{-4}$ for quadrature change. The reported solution is evaluated after 7450 neural updates.

Table~\ref{ext:tab:charcapacities} specifies the approximation capacities and update counts for the characteristic stages.

\begin{table}
\centering
\caption{Characteristic stages: capacities and actual update counts. Changed spaces are distinct stages, not one fixed-objective trace.}
\label{ext:tab:charcapacities}
\begin{tabular}{lrrl}
\toprule
Case/stage & Basis & Free coefficients & Updates\\
\midrule
Mach-7 correction & 24 & 47 & 2\\
Mach-8 initial/coarse & 24 & 47 & 12 then 8\\
Mach-8 selected refinement & 119 & 237 & 1 refinement update\\
Mach-12 final refinement & 235 & 469 & 6 total across two capacities\\
\bottomrule
\end{tabular}
\end{table}

The Mach-7 characteristic correction uses 481 spatial nodes on $[-105,105]$, a $512\times144$ velocity grid and cutoff 42. Two unit-fraction updates are followed by diagnostics at 64 Gauss--Legendre points. The two-update solve takes approximately $0.52$ s, excluding diagnostics; this timing covers only the characteristic correction.

The Mach-7 neural continuation keeps its objective coefficients fixed and preserves the inverse metric. Its architecture is $5\to64\to64\to64\to1$ with a scalar logarithmic mixture width, for 8770 trainable parameters. The neural calculation uses 369 spatial points on $[-105,105]$ and a $512\times144$ velocity grid with cutoff 42; the characteristic correction uses 481 spatial nodes.

\subsection{Data-assisted shock-tube continuation}
\label{ext:app:supervised}
The data-assisted shock tubes use the common positive macro--micro representation of \cref{sup:sec:shocktube}. The five-field base case uses 42 fixed heat-flux locations, 42 fixed stress locations and 160 macroscopic locations at $t_f=0.02$. The locations are unchanged during continuation. Three additional reference summaries constrain the global maximum heat flux and the maximum and minimum stress. These summaries and the reference amplitudes used for normalisation are part of the supplied training information.

Table~\ref{ext:tab:supervised_schedule} gives the five-field base-case schedule. The least-squares stages initialise their neural weights from the preceding accepted state and use a fresh damped Gauss--Newton metric. They retain the original $32\times12\times12$ velocity-node values and weights, cast to double precision, on $[-7,7]^3$. The time, relaxation-time and jump-strength cases in table~\ref{sup:tab:shockcases} use the same kinetic representation; the additional continuation in table~\ref{ext:tab:supervised_schedule} is applied to the base case.
\begin{table}
\centering
\caption{Training schedule of the selected data-assisted BGK base shock tube. Update counts refer to the stated stage.}
\label{ext:tab:supervised_schedule}
\small
\renewcommand{\arraystretch}{1.15}
\begin{tabularx}{\textwidth}{@{}>{\raggedright\arraybackslash}p{.29\textwidth}>{\raggedleft\arraybackslash}p{.10\textwidth}@{\hspace{1.2em}}>{\raggedright\arraybackslash}X@{}}
\toprule
Stage & Updates & Controlled quantities\\
\midrule
Final AdamW stage & 5000 & Macrofields and transport-moment targets\\
Constrained least squares & 45 & Field and weighted physical losses\\
Separate physical components & 11 & Raw residual and transport balances\\
Extremum-constrained continuation & 44 & Field, physical and extremum losses\\
\bottomrule
\end{tabularx}
\end{table}
The three least-squares stages require 2288 allocated graphics-processing-unit (GPU) seconds, including their recorded diagnostics; this excludes the preceding AdamW training. Their initial damping multiplier is $10^{-3}$, reduced by a factor of three after acceptance to a minimum of $10^{-8}$ and increased by a factor of ten after rejection. The proposed parameter displacement has Euclidean norm at most $0.1$; backtracking tests $2^{-k}$ for $k=0,\ldots,15$. The final physical stages use 32 fixed space--time collocation points. The evaluation includes a separate 32-point set and a fresh 64-point set, neither used in the training objective. These checks assess the saved neural distribution and its derivatives on the same coordinates for all compared states.

The eight-moment extension of \cref{sup:sec:eightmoment} keeps the 160 macroscopic and 42 moment locations, adds observations of $m_{xxx}$, $R_{xx}$ and $\Delta$ at the 42 moment locations, and takes all 690 observation values from the 6400-cell reference by linear interpolation. All moments, residuals and balances use the $48\times20\times20$ velocity grid on $[-7,7]^3$ in double precision. The 64 training points and the 64 separate points combine 32 equispaced positions on $[-\tfrac12,\tfrac12]$ with 32 on $[-0.12,0.14]$, at random times in $[0,t_f]$; the 64 points used by no stage combine 32 random positions on the domain with 32 on $[-0.12,0.14]$, at random times. The initial-time residual is evaluated at $t=10^{-300}$ rather than at $t=0$, where the implemented bound on $\widehat t$ has a one-sided derivative. Table~\ref{ext:tab:eightmoment_schedule} gives the schedule.
\begin{table}
\centering
\caption{Schedule of the eight-moment data-assisted BGK base shock tube. Accepted updates are counted from the five-field state. Trials are all-parameter trial states in the first three rows and attempted directions in the last row.}
\label{ext:tab:eightmoment_schedule}
\small
\renewcommand{\arraystretch}{1.15}
\begin{tabularx}{\textwidth}{@{}>{\raggedright\arraybackslash}p{.30\textwidth}>{\raggedleft\arraybackslash}p{.09\textwidth}>{\raggedleft\arraybackslash}p{.07\textwidth}@{\hspace{1.2em}}>{\raggedright\arraybackslash}X@{}}
\toprule
Stage & Accepted & Trials & Updated parameters and step\\
\midrule
Original envelopes, 1--17 & 17 & 84 & All 52\,753; projected Levenberg--Marquardt step with a function trust radius\\
Original envelopes, 18--161 & 144 & 292 & Temporal layer (1032), then all parameters; same step\\
Original envelopes, 162--170 & 9 & 16 & All parameters with constraints on coefficient time derivatives and lower-moment groups; temporal-layer corrections\\
Common-tail envelope & 0 & -- & None; envelopes replaced for $|x|>0.135$\\
Output-layer refinement & 35 & 40 & Two output layers (1032); \cref{sup:eq:olgn,sup:eq:olqp}\\
\bottomrule
\end{tabularx}
\end{table}
The fourteen jobs that contributed accepted updates with the original envelopes used 1115 GPU seconds and 6.0 hours on two processor threads; jobs without an accepted update are not included in these times. The step control of this continuation was changed between jobs when progress stalled, and further jobs between them committed no update. Two earlier eight-moment trainings from the five-field state, with 60 and 20 accepted updates and observations from an 800-cell solution, and a ridge fit of the added layers were not adopted. The damping multiplier decreases from $10^{-3}$ to $10^{-8}$ by a factor of three after each acceptance, and is fixed at $10^{-8}$ or $10^{-6}$ in the last nine updates; the function trust radius varies between $0.0125$ and $0.4$. From the 6th update the relative residual on the separate 64 points, and from the 23rd update the three balance mean squares on those points, enter the step as linearised inequality constraints. The acceptance guards are evaluated at 399 reference positions, on the heat-flux and stress extrema and on the separate 64 points, relative to the five-field state, with tolerance $10^{-3}$ for the lower fields, the extrema and the sampled physical measures and strict decrease for the higher moments; from the 162nd update they are also applied on all 6400 cells relative to the current state. The 6400-cell reference therefore serves that stage as a selection aid.

For the common-tail envelope, the integral of \eqref{sup:eq:commontail} over the blending interval is evaluated with a 16-point Gauss--Legendre rule. The step design of the output-layer refinement followed six short trial runs that read only training quantities. Updating all parameters changed the macroscopic outputs at times before $t_f$, and the balance components grew about fivefold for each doubling of the step. Updating only the higher-moment output layer raised the heat-flux and stress losses through the coupling of the moments. Both were abandoned in favour of the two output layers with blocks (ii)--(iv) of \cref{sup:sec:eightmoment}. At the end of the reported refinement the root-mean-square relative changes of the bounded outputs are $2.7\times10^{-4}$ in the tail block and $8.7\times10^{-3}$ in the function-change block. Before the reported run, the acceptance rule, the attempt cap and a tolerance of $0.1\%$ on any increase of the five sampled physical measures on the separate 64 points had been fixed. The reported run exceeded this tolerance (mean-square residual $+0.6\%$, energy-balance root-mean-square value $+0.8\%$). The same check was then applied to the 64 points used by no stage, where the tolerance was also exceeded (relative residual $+0.4\%$, energy balance $+0.2\%$). Two variants therefore imposed the limits also on the five physical measures of the separate 64 points. With the current-loss weights of block (ii) the variant stopped after 13 accepted updates with no admissible step. With weights fixed at their initial values it reached 40 accepted updates and full-domain discrepancies of $3.9\%$, $3.8\%$ and $3.3\%$ for $m_{xxx}$, $R_{xx}$ and $\Delta$, with $0.35\%$ and $0.67\%$ for $q_x$ and $\sigma_{xx}$; it met the tolerance on the separate points but still exceeded it on the unused points (relative residual $+0.3\%$, energy balance $+0.2\%$). The stricter guard thus limits the attainable fit without meeting the tolerance on unused points; the run without it had been designated as the primary one before the variants were run and is the one reported.\label{ext:app:supervised:end}

\subsection{Self-scaling coefficients}
\label{ext:app:ssb}

In this subsection, omit the iteration subscript: $g=\nabla J_{\rm NN}(\theta_k)$, $s=\theta_{k+1}-\theta_k$, $y=\nabla J_{\rm NN}(\theta_{k+1})-g$ and $H=H_k$. For parameter count \(n>1\), define \(Y=y^Ts>0\), \(Q=y^THy>0\), \(S=s^Tg<0\), and accepted step \(\alpha>0\). The dense update uses the self-scaling coefficients in~\cref{eq:self_scaling_coefficients,eq:self_scaling_coefficients-2,eq:self_scaling_coefficients-3,eq:self_scaling_coefficients-4,eq:self_scaling_coefficients-5,eq:self_scaling_coefficients-6}:
\begin{align}
 h&=Q/Y,&
 b&=-\alpha S/Y,&
 a&=bh-1,
 \label{eq:self_scaling_coefficients}\\
 \rho_-&=\min\{1,h[1-\sqrt{\max(|a|/(1+a),0)}]\},
 \label{eq:self_scaling_coefficients-2}\\
 \vartheta_-&=(\rho_--1)/a,&
 \vartheta_+&=1/\rho_-,
 \label{eq:self_scaling_coefficients-3}\\
 \vartheta&=\max\{\vartheta_-,\min[\vartheta_+,(1-b)/b]\},&
 \rho_*&=\min(1,1/b),
 \label{eq:self_scaling_coefficients-4}\\
 \sigma_*&=1+\vartheta a,&
 \phi_k&=(1-\vartheta)/(1+a\vartheta),
 \label{eq:self_scaling_coefficients-5}\\
 \tau_k&=
 \begin{cases}
 \min(\rho_*|\sigma_*|^{1/(1-n)},\sigma_*),
     &\vartheta\leq0,\\
 \rho_*\min(|\sigma_*|^{1/(1-n)},1/\vartheta),
     &\vartheta>0.
 \end{cases}
 \label{eq:self_scaling_coefficients-6}
\end{align}
The displayed formulas describe the nondegenerate branch. The code uses \(\epsilon=10^{-30}\) for its denominator tests, sets \(\vartheta_-=0\) for \(|a|<\epsilon\), and uses explicit guarded branches when \(\rho_-\), \(b\), or \(\sigma_*\) vanish. It rejects invalid curvature inputs, nonfinite coefficients, and \(\tau_k\leq\epsilon\). Here \(\phi_k\) and \(\tau_k\) are the coefficients of \eqref{ext:eq:ssbupdate}, and \(\rho_-\), \(\rho_*\), \(\sigma_*\) and \(\vartheta\) are auxiliary scalars of the update, unrelated to the density, stress and relaxation time. The implementation, including these branches, accompanies the reproducibility material.

\subsection{Independent cold-start Mach-2 calculation}
\label{ext:sec:independentseeds}

Three independent initialisations test the Shakhov Mach-2 solver ($\Pr=2/3$, $\tau=1/\rho$). Each uses 20\,000 Adam updates, 500 SSB updates and 500 SSB updates with the split-half velocity quadrature of \cref{ext:sec:newextension}. The Adam and first SSB stages use a single 160-node axial and 48-node radial Gauss--Legendre rule with cutoff 16; the split-half stage also uses the characteristic incoming gate of \cref{ext:sec:newextension}. Settings and budgets are identical across seeds.

Errors use the 81 common positions on $[-10,10]$ and relative $L^2$ norms evaluated by trapezoidal integration. The reference is centred at its RH density midpoint; the model coordinates are unchanged. Sample standard deviations (SDs) use denominator $n-1$ with $n=3$. Reference discretisation uncertainty remains separate from seed variability.

Table~\ref{ext:tab:seedfinal} reports the final field errors and sample statistics for all three initialisations. Figure~\ref{ext:fig:seedprofiles} shows the physical profiles from the three independent initialisations; figure~\ref{ext:fig:seedlosses} places their stage-specific objectives above the corresponding heat-flux and stress profiles.

\begin{table}
\centering
\caption{Final split-half SSB relative $L^2$ errors (\%).}
\label{ext:tab:seedfinal}
\begin{tabular}{lrrrrr}
\toprule
Seed & $\rho$ & $u$ & $T$ & $q_x$ & $\sigma_{xx}$\\
\midrule
141001 & 0.129693 & 0.132567 & 0.113003 & 1.10627 & 1.28601\\
141002 & 0.120371 & 0.130867 & 0.113351 & 0.990652 & 1.28132\\
141003 & 0.126803 & 0.133233 & 0.112584 & 1.0151 & 1.25607\\
Mean & 0.125622 & 0.132222 & 0.112979 & 1.03734 & 1.27447\\
Sample SD & 0.00477226 & 0.00121986 & 0.000384356 & 0.0609347 & 0.0161008\\
\bottomrule
\end{tabular}
\end{table}

Table~\ref{ext:tab:seedstages} separates the heat-flux and stress statistics at the three optimisation stages.

\begin{table}
\centering
\caption{Stage-specific core errors (\%); distinct objectives and endpoints are not pooled.}
\label{ext:tab:seedstages}
\begin{tabular}{lrrrrr}
\toprule
Stage & $n$ & mean $q_x$ & SD $q_x$ & mean $\sigma_{xx}$ & SD $\sigma_{xx}$\\
\midrule
Adam initialisation & 3 & 20.1871 & 0.296159 & 16.5819 & 1.46827\\
SSB continuation & 3 & 2.11274 & 0.255997 & 1.08211 & 0.202463\\
Split-half SSB & 3 & 1.03734 & 0.0609347 & 1.27447 & 0.0161008\\
\bottomrule
\end{tabular}
\end{table}

The three trajectories show a substantial improvement during quasi-Newton continuation: the mean heat-flux and stress discrepancies decrease from $20.19\%$ and $16.58\%$ after Adam to $2.11\%$ and $1.08\%$ after 500 SSB updates (table~\ref{ext:tab:seedstages}). Adam provides the stochastic initialisation, while the fixed-sample SSB stage uses accumulated curvature information to refine the solution. These observations establish the practical contribution of the second stage in this schedule; they do not separate the effects of the optimiser, additional computation and changed sampling through an equal-budget control. The loss histories in figure~\ref{ext:fig:seedlosses} are therefore interpreted alongside the physical errors.

Table~\ref{ext:tab:seedcost} records the completed update counts, GPU costs and combined physical-test outcomes.

\begin{table}
\centering
\caption{Update counts, computational cost and combined physical-test outcome.}
\label{ext:tab:seedcost}
\begin{tabular}{lrrrrl}
\toprule
Seed & Adam & SSB & Split-half SSB & GPU seconds & All tests pass\\
\midrule
141001 & 20000 & 500 & 500 & 1783 & No\\
141002 & 20000 & 500 & 500 & 1717 & No\\
141003 & 20000 & 500 & 500 & 1676 & No\\
\bottomrule
\end{tabular}
\end{table}

The three runs require 5176 allocated GPU seconds, including diagnostics; a 26-second timing test brings the total to 5202 seconds. These repetitions assess the Shakhov Mach-2 schedule only.

The split-half stage reduces mean heat-flux error from $2.11274\%$ to $1.03734\%$, while mean stress error increases from $1.08211\%$ to $1.27447\%$. The projection and tail-residual tests fail for all three final states; the kinetic-residual test also fails for seeds 141002 and 141003. Flux, phase, quadrature and positive-macro tests pass. Agreement across seeds therefore measures repeatability at this budget rather than satisfaction of all residual tolerances.

\begin{figure}
\centering
\includegraphics[width=\textwidth,height=0.73\textheight,keepaspectratio]{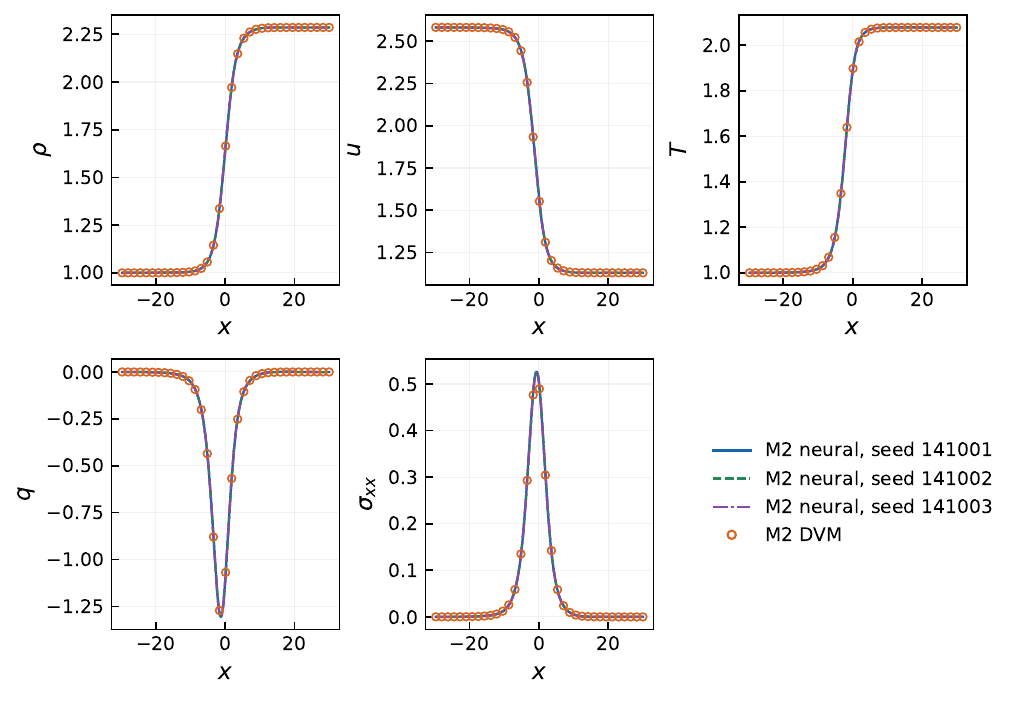}
\caption{Equation-trained Mach-2 Shakhov fields from the three independent cold-start lineages after the final split-half SSB stage, with the matched provisional DVM reference. The curves cover the full comparison interval; the tables above use the separately stated $[-10,10]$ core. Model coordinates and amplitudes are unchanged. The reference uses only its own prescribed RH density-midpoint shift.}
\label{ext:fig:seedprofiles}
\end{figure}

\begin{figure}
\centering
\includegraphics[width=\textwidth]{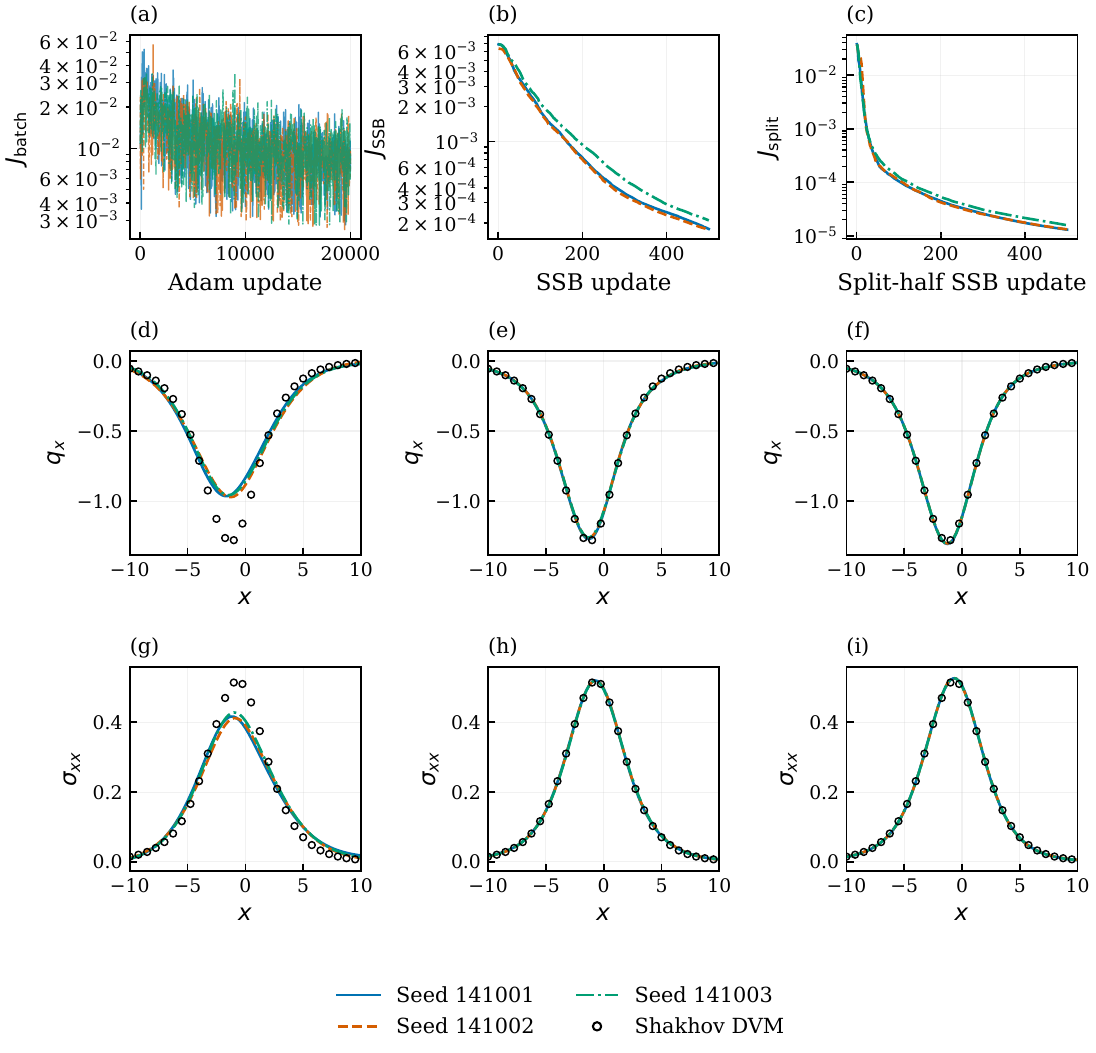}
\caption{Equation-trained Mach-2 Shakhov optimisation and the resulting transport moments for three independent initialisations. Columns correspond to Adam, SSB and split-half SSB. Panels (a--c) show the recorded stage objectives, each formed from the ten terms of the stationary objective of \S\ref{ext:sec:methodA}: $J_{\rm batch}$ on the random eight-point minibatch of each Adam update, recorded before the update, with residual weights ramped over the first 1000 updates; $J_{\rm SSB}$ on 128 fixed points after each accepted update; and $J_{\rm split}$, the same objective rebuilt with the split-half quadrature and the characteristic gate, after each accepted update; (d--f) and (g--i) show endpoint heat flux and stress. The same colours identify seeds 141001, 141002 and 141003 throughout. Open symbols show the DVM reference interpolated onto the common model nodes in $[-10,10]$. Physical ordinate limits are identical across columns. The mean heat-flux/stress errors are $20.19/16.58\%$, $2.11/1.08\%$ and $1.04/1.27\%$, respectively. Stochastic sampling, objective definitions and quadrature differ between stages; loss panels are therefore separate. The figure shows the improvement along the prescribed training schedule, rather than an equal-budget comparison of optimisers.}
\label{ext:fig:seedlosses}
\end{figure}

\FloatBarrier

\section{Reference and refinement studies}
\label[appsec]{ext:app:reference}

Each neural--characteristic comparison uses the same positions and reference denominator: 161 core points on $[-35,35]$ for Mach 7 and 87 on $[-39,39]$ for Mach 8. The Mach-7 residual diagnostics use a separate 64-point Gauss--Legendre grid. Mach-12 comparisons use heat flux and central stress integrated from the full distribution; equilibrium-subtracted stress is a different observable. The reference-grid sensitivity is reported in \cref{ext:sec:extresults}.

For the shock tubes, table~\ref{tab:transient_physics} compares the 1600- and 3200-cell DVM solutions with the same $48\times20\times20$ velocity quadrature and CFL number 0.4. The relative changes in heat flux and stress are $0.4065\%$ and $0.3578\%$ for BGK, and $0.4060\%$ and $0.3578\%$ for Shakhov; the largest reported fourth-order change is approximately $1.7\%$. This is a coupled spatial--temporal refinement check, since reducing the cell width also reduces the time step. The separate tests of characteristic nodes and velocity quadrature in \S\ref{sec:transient_characteristic_pinn} evaluate the fixed neural solution and are not DVM convergence tests. For the data-assisted base case, table~\ref{sup:tab:eightmoment} adds the difference between the 3200- and 6400-cell BGK solutions, with the same convention: $0.2047\%$ and $0.1798\%$ in heat flux and stress, and at most $0.86\%$ in the fourth-order moments. Each difference is close to half of the corresponding 1600-to-3200-cell value, as expected for a first-order scheme. The heat-flux maximum changes from $0.0432021$ at $x=-0.03359375$ to $0.0433034$ at $x=-0.0335156$. The equation-trained BGK prediction of \S\ref{sec:transient_characteristic_pinn} and the final data-assisted state of \cref{sup:sec:eightmoment} can therefore be compared on the same reference. At the 401 saved positions of the equation-trained prediction, its discrepancies against the 6400-cell solution are $0.21\%$, $0.18\%$, $0.41\%$, $0.85\%$ and $0.70\%$ in $q_x$, $\sigma_{xx}$, $m_{xxx}$, $R_{xx}$ and $\Delta$; these lie within $0.003$ percentage points of the differences between the 3200- and 6400-cell solutions at the same positions ($0.20\%$, $0.18\%$, $0.41\%$, $0.85\%$ and $0.70\%$).\label{ext:app:reference:common} The final data-assisted state gives $0.33\%$, $0.67\%$, $1.83\%$, $1.97\%$ and $1.79\%$ at those positions. Against the first-order extrapolate $2a_{6400}-a_{3200}$ of the two finest solutions, the equation-trained discrepancies fall to $0.004\%$, $0.002\%$, $0.007\%$, $0.010\%$ and $0.009\%$, whereas the data-assisted ones are $0.38\%$, $0.67\%$, $1.88\%$, $2.65\%$ and $2.38\%$. For the fluxes \eqref{eq:stress_flux} and \eqref{eq:heat_flux_flux}, the equation-trained discrepancies against the extrapolate are $0.005\%$ for the stress flux and $0.003\%$ for the non-equilibrium part $\frac12S_{xx}+\frac72T\sigma_{xx}$ of the heat-flux flux, against $0.73\%$ for both in the data-assisted state; the full heat-flux flux is dominated by its equilibrium part $\frac52\rho T^2$ and does not discriminate between them. The extrapolation assumes the first-order convergence indicated by the halving of successive differences; these last values are indicative.

For stationary shocks, the saved reference profiles and available refinement comparisons support the stated field discrepancies, but do not establish uniform space--velocity--time convergence for all Mach numbers. In particular, the Mach-12 comparisons in \S\ref{ext:sec:extresults} distinguish the medium and fine snapshots from the separate reference family; differences between families are not treated as a controlled refinement sequence.

\section{Optimisation curves}
\label[appsec]{ext:app:optimisation}

The higher-moment profiles accompany the physical-field comparisons in \cref{ext:sec:extresults}. The following figures give the optimisation curves for each calculation.

\subsection{Data-assisted shock-tube continuation}
Figure~\ref{sup:fig:continuation_losses} shows the three refinement blocks of the five-field base Sod case. The common positive macro--micro representation and the field observations are retained throughout; the active residual and balance terms change as specified in section~\ref{sup:sec:cases}.
\begin{figure}
\centering
\includegraphics[width=\textwidth]{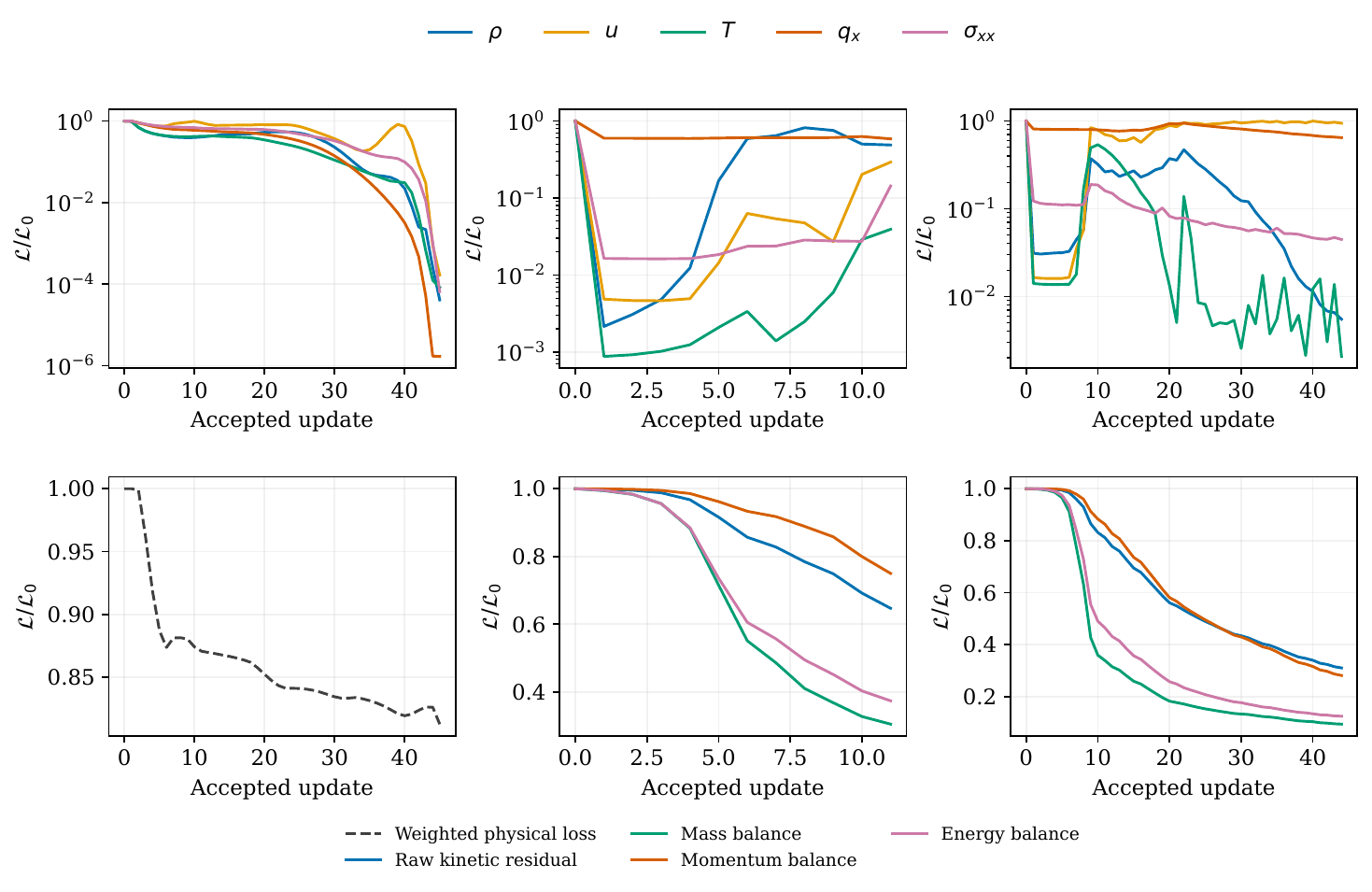}
\caption{Five-field data-assisted BGK shock-tube refinement. Columns show the joint-fit, balance-controlled and extremum-controlled blocks, with 45, 11 and 44 accepted updates, respectively. Upper panels show the five field-loss components. Lower panels show the weighted physical aggregate in the first block and the separate kinetic-residual and mass, momentum and energy balance losses in the remaining blocks. Each component is divided by its value at the start of its own block. The panels therefore display progress within a block, rather than a single unchanged objective across blocks.}
\label{sup:fig:continuation_losses}
\end{figure}

Figure~\ref{sup:fig:outputlayer_losses} shows the output-layer refinement of the eight-moment base case (\cref{sup:sec:eightmoment}). The fifteen components are divided by their values at the start of the refinement, that is, for the common-tail envelope with the weights of the preceding continuation. The eight observation losses decrease together, by $94$--$95\%$ in total. Most of this decrease occurs in four consecutive updates with step lengths between $1/8$ and $1/2$; the other accepted step lengths are at most $1/16$. No protected component exceeds its initial value by more than the relative margin $10^{-6}$; the initial-time residual stays at that margin, and the other six physical components decrease by between $0.2\%$ and $9.5\%$.
\begin{figure}
\centering
\includegraphics[width=\textwidth]{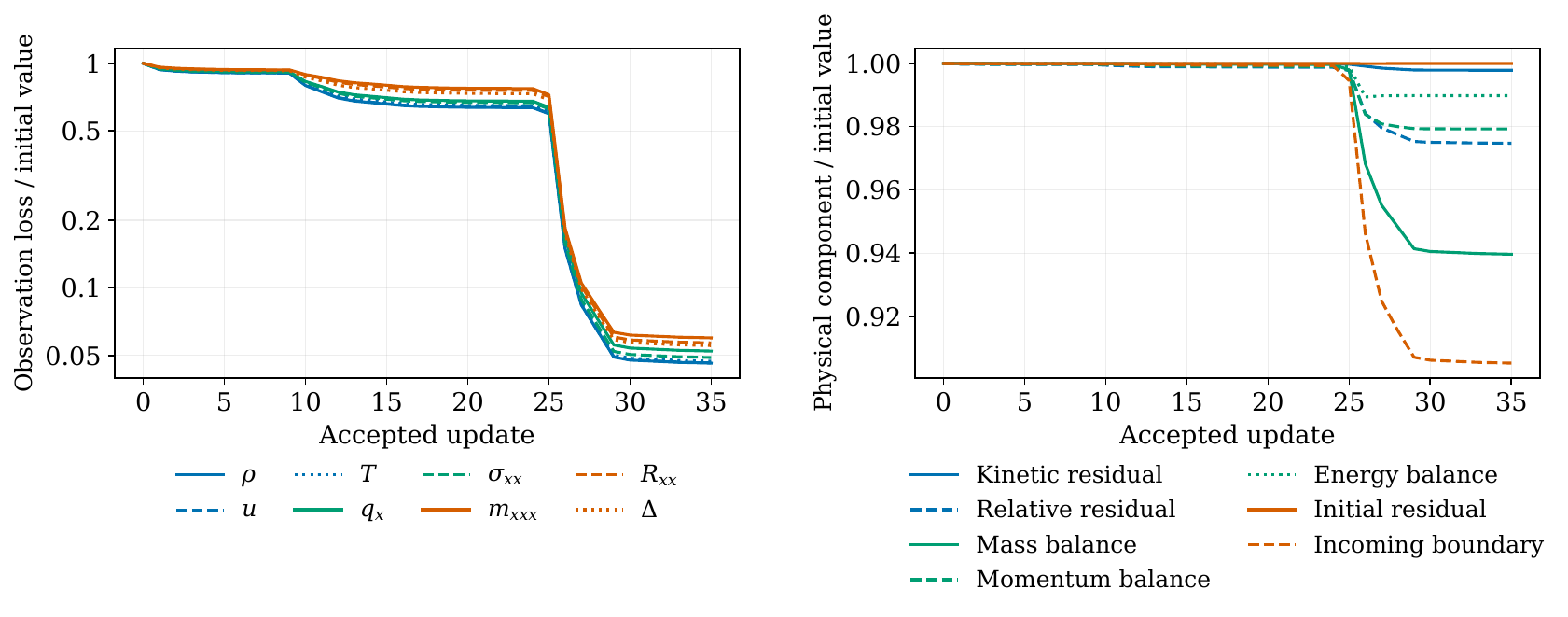}
\caption{Output-layer refinement of the eight-moment data-assisted BGK shock tube: the eight observation losses (left) and the seven physical components (right) after each of the 35 accepted updates, divided by their values at the start of the refinement. The five lower-field losses and the seven physical components are protected; the three higher-moment losses define the objective. The eight loss curves nearly coincide.}
\label{sup:fig:outputlayer_losses}
\end{figure}

\subsection{Optimisation curves for each Mach number}
The curves below cover neural calculations at Mach 2--7, characteristic corrections at Mach 7, 8 and 12. The independent-initialisation curves are shown in figure~\ref{ext:fig:seedlosses}. Each panel represents a specified optimisation stage; changes of objective, quadrature or representation are shown separately.

The Mach-2, Mach-3 and Mach-5 curves contain 400, 1375 and 1375 continuation updates, respectively. Mach 4 contains four optimisation segments spanning updates 1--56875. Mach 6 shows updates 1001--7450, and Mach 7 shows an initial evaluation and three accepted updates. These curves cover the stated stages rather than complete random-initialisation training.

\begin{figure}
\centering
\includegraphics[width=\textwidth]{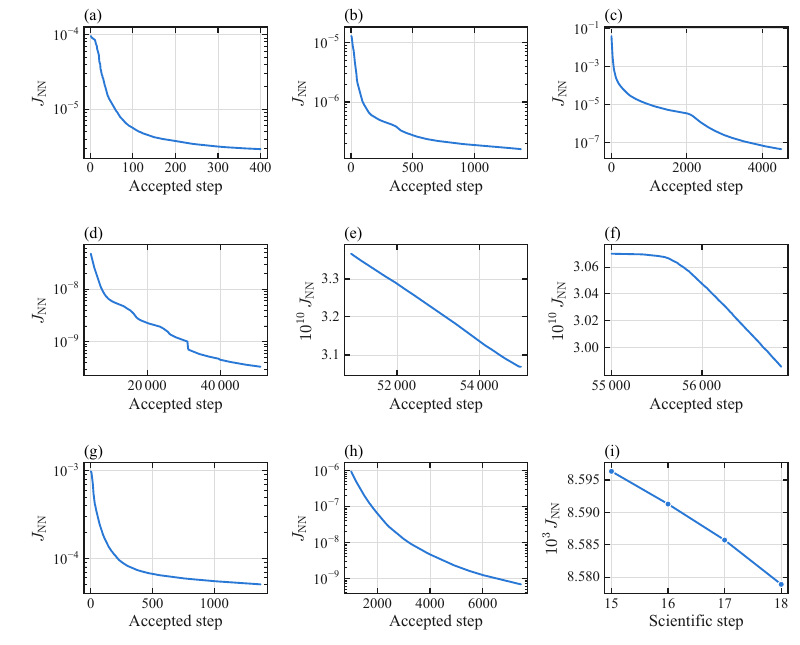}
\caption{Equation-trained neural continuation objectives $J_{\rm NN}$. (a,b) Mach-2 and Mach-3 Shakhov calculations, with 400 and 1375 updates. (c--f) Mach-4 Shakhov optimisation segments: updates 1--4500, 4501--50890, 50891--54999 and 55000--56875. (g) Mach-5 Shakhov continuation, 1375 updates. (h) Mach-6 BGK continuation, updates 1001--7450. (i) Mach-7 BGK initial evaluation and three accepted updates; its abscissa counts the accepted quasi-Newton updates over the lifetime of the Mach-7 network, so that 15 is the inherited state and 16--18 are the three updates. Separate panels preserve changes of objective or inverse metric. Scaled ordinates are used for curves spanning about one decade or less.}
\label{ext:fig:recordedloss}
\end{figure}

\begin{figure}
\centering
\includegraphics[width=\textwidth]{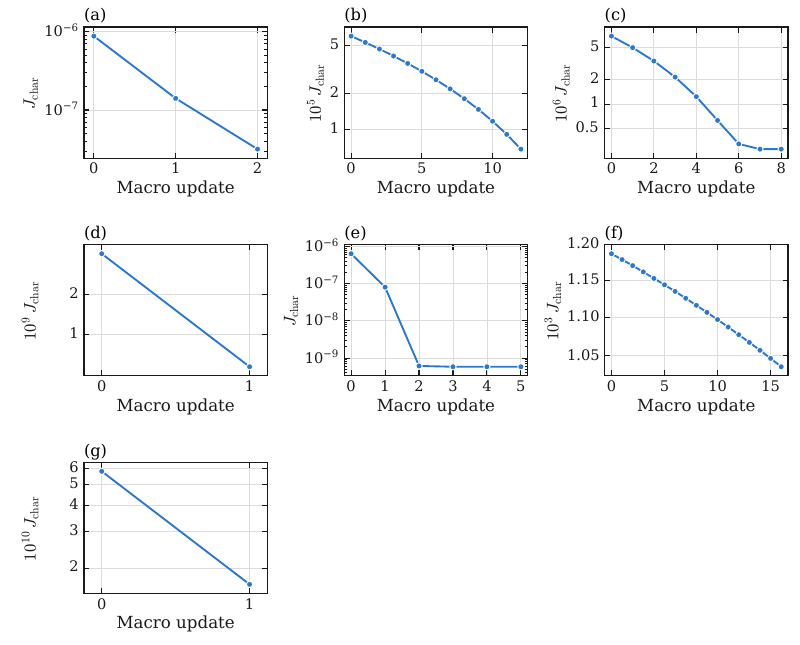}
\caption{Equation-based characteristic self-consistency objectives $J_{\rm char}$. (a) Mach-7 correction, two updates. (b,c) Mach-8 coarse stages, 12 and 8 updates. (d) Mach-8 119-basis refinement, one update. (e,f) Mach-12 initialisation comparison: five updates from transferred fields and 16 from analytic RH fields. (g) Mach-12 235-basis refinement, one update after the five-update transferred-field stage. Each curve uses its stage-specific objective; omitted training intervals are not interpolated.}
\label{ext:fig:charloss}
\end{figure}

Stochastic sampling can produce nonmonotone losses. Objective reduction is interpreted together with the field and residual comparisons.

\subsection{BGK and Shakhov shock-tube optimisation}
The two independently initialised shock-tube calculations of \S\ref{sec:transient_characteristic_pinn} each completed 1000 Adam updates and 300 accepted self-scaled Broyden updates. Figures~\ref{fig:transient_bgk_loss} and~\ref{fig:transient_shakhov_loss} show all recorded objective values. Adam values are evaluated before each update on the current random minibatch of 64 space--time points; Broyden values are evaluated after accepted updates on the fixed 256-point objective. The stages are shown on separate axes and are not connected across the transition. These objectives measure moment self-consistency, not disagreement with reference fields, and their magnitudes do not by themselves certify kinetic convergence.

\begin{figure}
\centering
\includegraphics[width=\textwidth]{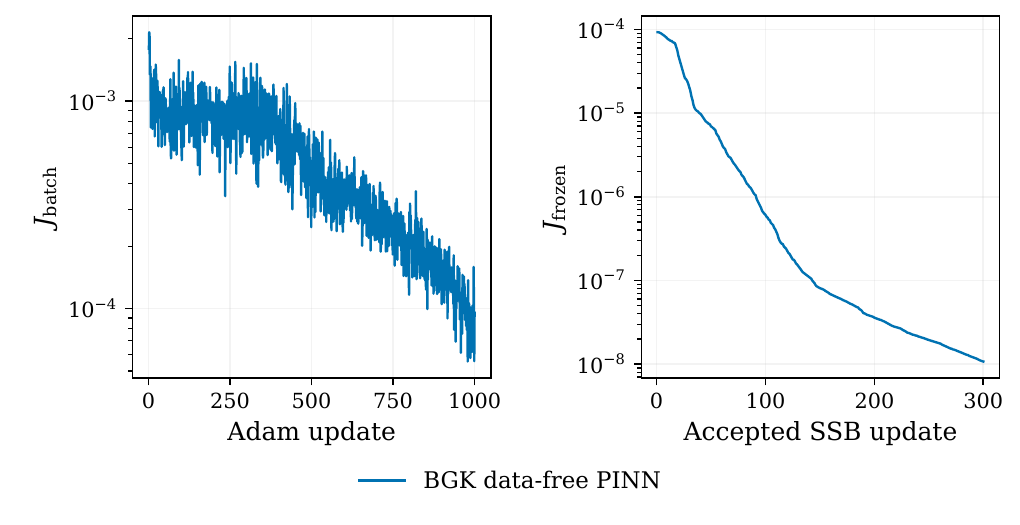}
\caption{Equation-trained PINN objectives for the independent BGK shock-tube calculation: 1000 pre-update stochastic Adam values (left) and 300 post-accepted-update fixed-point Broyden values (right): $J_{\rm batch}$ is the self-consistency objective $J_{\rm sc}$ of \cref{sec:transient_characteristic_pinn} on the 64-point Adam minibatch, evaluated before each update, and $J_{\rm frozen}$ is $J_{\rm sc}$ on the fixed 256-point Broyden set after each accepted update. The two stages are not a single fixed-objective trajectory.}
\label{fig:transient_bgk_loss}
\end{figure}
\begin{figure}
\centering
\includegraphics[width=\textwidth]{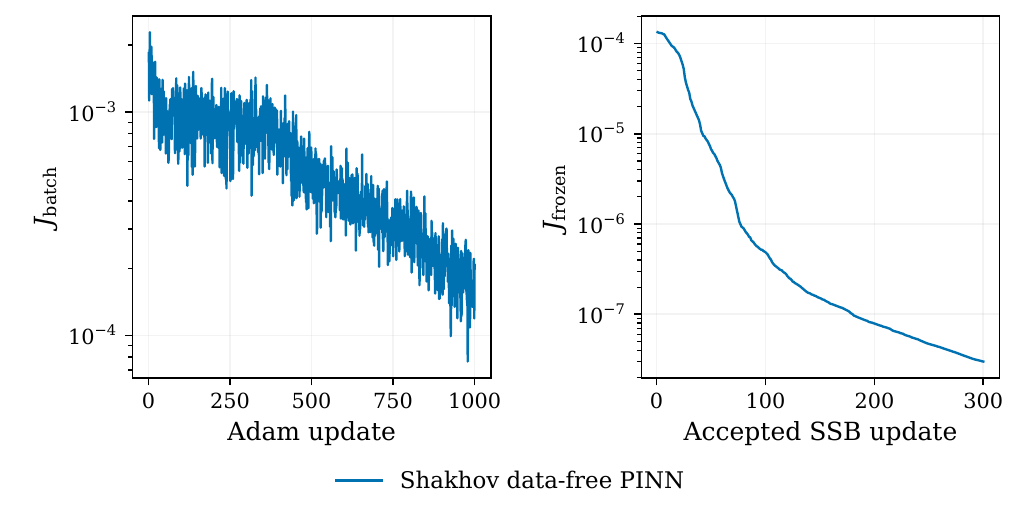}
\caption{Equation-trained PINN objectives for the independent Shakhov shock-tube calculation, with the same stage-specific timing convention as figure~\ref{fig:transient_bgk_loss}. The self-consistency objective includes heat flux in addition to density, momentum and energy; its scale is not an accuracy ranking against the BGK objective.}
\label{fig:transient_shakhov_loss}
\end{figure}

\FloatBarrier

\section*{Funding}
This research received no specific grant from any funding agency, commercial or not-for-profit sectors.

\section*{Declaration of interests}
The author reports no conflict of interest.

\section*{Author contributions}
E.R. conceived the study, developed the numerical methodology, performed the DVM and neural-solver calculations, analysed the results, prepared the figures, and wrote the manuscript.

\section*{Data availability}
The numerical code, configurations, saved field and optimisation data, post-solve reference data and reproducibility instructions are provided in the PINNs-DVM repository at \url{https://github.com/Ehsan-Roohi/PINNs-DVM}. The repository distinguishes data-assisted training observations and calibration quantities from references used solely for evaluation by equation-trained methods. For the eight-moment base case it includes the model state, the envelope and output-layer refinement code as executed, the 6400-cell reference fields, the observation locations, the accepted-update history and the assessment arrays. The deposited manifest records the source and result hashes.

\par\medskip
\noindent\begin{minipage}{\textwidth}
\section*{Declaration of generative artificial intelligence assistance}

Generative artificial intelligence (AI) tools---ChatGPT and Codex (OpenAI), and Claude, including the Claude Fable~5.2 model (Anthropic)---were used to assist with code implementation and debugging, numerical analysis and diagnostic scripts.

\end{minipage}
\clearpage
\bibliographystyle{plainnat}
\bibliography{references}

\end{document}